\documentclass[11pt]{amsart}

\usepackage{macros}
\date{}

\begin{document}
\title{Supersymmetric gauge theory and the Yangian}
\author{Kevin Costello}
\thanks{Partially supported by NSF grant DMS 1007168, by a Sloan Fellowship, and by a Simons Fellowship in Mathematics.}

\address{Department of Mathematics, Northwestern University.}
\email{costello@math.northwestern.edu}

\maketitle

\renewcommand{\L}{\mathscr{L}}

\section*{Introduction}
\newcommand{\F}{\mscr{F}}
\newcommand{\Dirac}{\slashed{\partial}}

In recent years, great progress has been made in performing exact calculations in the $N=2$ and $N=4$ supersymmetric gauge theories.  Exact calculations can be made in the $N=2$ theory by considering the topological twist and using the method of localization \cite{Nek02}.  Inspired by the AdS/CFT correspondence, an integrable structure has been discovered in the $N=4$ gauge theory, leading to many exact calculations \cite{AdaBulMas11, ArkBouCac10}.  However, the $N=1$ gauge theory has proved more difficult; it is not integrable, and does not admit a topological twist. 

The $N=1$ gauge theory does, however, admit a \emph{holomorphic} twist \cite{Joh94}. In this paper, I perform some exact calculations to all orders in perturbation theory (ignoring instanton effects) in a deformation of the $N=1$ gauge theory, using the holomorphic twist.   The twisted deformed theory is a mixture of holomorphic and topological, and is an $N=1$ version of a twisted $N=2$ theory studied by Kapustin \cite{Kap06}.  I show that this twisted, deformed gauge theory is controlled by the Yangian, in the same way that Chern-Simons theory is controlled by the quantum group.  

The main theorem is a rather abstract relationship between the Yangian and the field theory I consider.  However, this abstract theorem is used to give an exact calculation of the vacuum expectation values of a certain net of $n+m$ Wilson operators in the theory compactified on a $2$-torus.   This vacuum expectation value coincides with the partition function of a spin-chain integrable lattice model on an $n \times m$ doubly-periodic lattice.  More generally, the whole spin-chain integrable system is embedded into the Wilson operators of this twisted, deformed $N=1$ gauge theory: the Hilbert space of the spin-chain system correspond to the Hilbert space of the gauge theory in the presence of Wilson operators, and the transfer matrix of the spin-chain system corresponds to a Wilson operator acting on this Hilbert space.  Integrability is evident from the field theory point of view.  

This theorem is stated and proved in the framework of (rigorous) perturbative quantum field theory developed in \cite{Cos11} and \cite{CosGwi11}.  Thus, the expectation values of the Wilson operators are defined from first principles of perturbative QFT: they can be written as a sum of Feynman diagrams which at all orders in perturbation theory is finite but very large.  However, the proof of this result involves the computation of only two simple Feynman diagrams (which are needed to fix a normalization). Instead, the proof relies on formal associativity properties of the operator product developed in \cite{CosGwi11}. It is somewhat remarkable that these associativity properties are strong enough that they can be used to give exact calculations of Wilson operators. 

The appearance of the Yangian in this paper is quite different from the well-known appearance of the Yangian in $N=4$ supersymmetric gauge theory (see \cite{Fer11} for a survey).   In the $N=4$ theory, the Yangian of $\mf{su}(4 \mid 4)$ appears as a symmetry of the gauge theory for $SU(N)$ in the large $N$ limit.   Here, we find the Yangian associated to the gauge group, and not as a symmetry but as part of the algebra of operators.   

Another connection between the Yangian and gauge theory was discovered by Nekrasov and Shatashvili \cite{NekSha09}, and further developed by Nekrasov-Pestun \cite{NekPes12} and Maulik-Okounkov \cite{MauOko12}.  In these papers, it is shown that an $N=2$ quiver gauge theory associated to an ADE or affine ADE quiver is related to the Yangian for the ADE group. Again, this is quite different to what is shown in this paper: here we find the Yangian for the gauge group, not the group associated to the quiver. 

\subsection{}
Let me introduce some notation before stating the main results.  

Let $\g$ be a semi-simple Lie algebra.  The fields of the $N=1$ supersymmetric gauge theory with Lie algebra $\g$ consist of a connection $A \in \Omega^1(\R^4) \otimes \g$, and an adjoint valued spinor $\Psi \in \mscr{S} \otimes \g$, with action functional
$$
S(A,\Psi) = \int \ip{F(A),  \ast F(A)}_{\g}  + \int \ip{\Psi, \Dirac_A \Psi }_{\g}
$$
where $\ip{-,-}_{\g}$ refers to an invariant inner product on $\g$. 

We are interested in a deformation of this theory which adds to it a Chern-Simons term.  Let us choose a complex structure on $\R^4$, and let $z,w$ be complex coordinates.    Let 
$$
\op{CS}(A) = \tfrac{1}{2} \ip{A, \d A} + \tfrac{1}{6} \ip{A, F(A)}. 
$$
be the Chern-Simons three-form associated to the connection $A$.  Our deformed theory has action functional $S(A,\Psi) + S'(A)$, where 
$$
S'(A) = \frac{1}{2 \pi i}\int \d z \op{CS}(A) + c_{\alpha \beta} \int \Psi^\alpha_+ \Psi^{\beta}_- .
$$
Thus, a Chern-Simons Lagrangian is added and some of the fermions are given certain masses (the constants $c_{\alpha \beta}$ will be specified later).  The normalization of $\frac{1}{2 \pi i}$ is to match the standard normalizations in the theory of integrable systems. 

This deformation breaks Lorentz invariance, and breaks some but not all supersymmetry: the theory is invariant under one of the supercharges in $Q \in \mc S_+$.  I call this the Yangian deformation, for reasons that will become clear later.  

\begin{remark}
Integration by parts allows one to rewrite the bosonic part of this deformation as
$$
\int \d z CS(A) = - \frac{1}{2} \int z F(A) \wedge F(A). 
$$
Thus, we have an $N=1$ gauge theory with a $\theta$-angle which varies linearly on space time.  A similar deformation in $N=4$ supersymetric gauge theory was considered in \cite{GaiWit08}.  Gaiotto and Witten consider a theta-angle which varies in a more general way over space time; they prove that the deformation of $N=4$ gauge theory they consider is $\tfrac{1}{2}$ BPS, i.e.\ preserved by half of the supercharges of the theory.  In this paper it is shown that the $N=1$ deformation is at least $\tfrac{1}{4}$ BPS. Probably Gaiotto-Witten's argument applies to show that this $N=1$ deformation is in fact $\tfrac{1}{2}$ BPS. 
\end{remark}

\subsection{}
We are interested in a twist of this deformed $N=1$ gauge theory.  Because twists of $N=1$ theories are not often considered, I should say a little about what I mean by twisting.

The twisting procedure can be described by the following three steps.
\begin{enumerate}
\item Change the action of $\op{Spin}(4)$ on the fields of your theory, using a map from $\op{Spin}(4)$ to the $R$-symmetry group $G_R$.
\item Choose an $S^1$ in the $R$-symmetry group $G_R$, which is used to change the cohomological degrees of the theory. 
\item Choose a $\op{Spin}(4)$-invariant supercharge $Q$ with $[Q,Q] = 0$, which is of weight $1$ under the chosen $S^1 \subset G_R$.  By adding $Q$ to the BRST differential of our theory, we find a new theory, which is simper than the original theory but still contains a lot of information about it. For example, correlation functions between $Q$-invariant observables of the original theory are captured by the twisted theory. 
\end{enumerate}
This procedure does not work with $N=1$ supersymmetry: it is not possible to choose a nilpotent supercharge $Q$ whose stabilizer, in $\op{Spin}(4) \times G_R$, contains a copy of $\op{Spin}(4)$.     This is why twisted $N=1$ theories are rarely considered in the physics literature. 

Nevertheless, we can perform steps 2 and 3 of the twisting procedure. The resulting twisted field theory will not be $\op{Spin}(4)$-invariant, but will be invariant under one of the $SU(2)$'s in the decomposition $\op{Spin}(4) = SU(2) \times SU(2)$.  (The observant reader might also object to step $2$, because the $N=1$ gauge theory has an $R$-symmetry anomaly.  However, it turns out that the anomaly is exact for the supercharge $Q$, so that we can still twist).

\subsection{}
The Yangian deformation of the $N=1$ theory is still acted on by one of the supercharges $Q \in \mc S_+$, and by the $R$-symmetry group $\C^\times$.  Thus, we can consider the twist of the Yangian deformation.

The twist of the Yangian deformation of the $N=1$ gauge theory can be described as a mixture of holomorphic and topological Chern-Simons theory.  It is equivalent to the gauge theory on $\C_z \times \C_w$ where the fields are 
$$A = A_0 \d \zbar + A_1 \d \br{w}  + A_2 \d w$$
where $A_i$ are $\g$-valued smooth functions on $\C_z \times \C_w$, with action functional 
$$
S(A) = \frac{1}{2 \pi i} \int \d z CS(A)
$$
where $CS(A)$ is the Chern-Simons three-form as above. The field $A$ should be interpreted as a partial connection on $\C_z \times \C_w$, giving a $\dbar$-operator in the $z$-direction and a full connection in the $w$-direction.   This theory has the property that solutions to the equations of motion are holomorphic bundles on $\C^2$ (with coordinates $z,w$), together with a flat holomorphic connection in the $w$-direction.      Thus, the theory is holomorphic in the $z$-direction and topological in the $w$-direction.   

One derives this expression for the twisted theory as follows.  Before we twist, there are $7 \op{dim} \g$ independent bosonic fields and $4 \op{dim} \g$ independent fermionic fields. The bosonic components are the three $\g$-valued components of $B \in \Omega^2_+ \otimes \g$ and the four $\g$-valued components of $A \in \Omega^1$; the fermionic ones are the two $\g$-valued components of $\psi_+ \in \mc{S}_+ \otimes \g$ and of $\psi_- \in \mc{S}_- \otimes \g$.  In the twisting procedure, we consider $Q$-closed observables modulo $Q$-exact observables. This amounts to taking the $Q$-cohomology of the space of fields; when we do this, field cancel out in pairs, so we are left with $3$ bosonic components.  These components correspond to the partial connection 
$$A = A_0 \d w + A_1 \d \zbar + A_2 \d \br{w}.$$

In a similar way, I prove that the undeformed, twisted gauge theory is the holomorphic BF theory, with fields $A' \in \Omega^{0,1}(\C^2) \otimes \g$ and $B \in \Omega^{2,0}(\C^2) \otimes \g$, with action functional $\int \ip{B, F^{0,2}(A')}$.  Under the deformation, the field $B$ becomes the $\d w$ component of the partial connection $A$ above. 

\begin{theorem*}
The twisted, deformed $N=1$ gauge theory admits a unique perturbative quantization (compatible with certain symmetries) on any complex surface $X$ equipped with a holomorphic volume form and a closed holomorphic $1$-form.
\end{theorem*}
This is proved in the formalism of \cite{Cos11}, by a cohomological analaysis. 

The case we have discussed above is $\C \times \C$ with volume form $\d z \wedge \d w$ and $1$-form $\d z$.  When the $1$-form is zero, we recover the twisted, undeformed $N=1$ gauge theory.   Thus, this theorem applies to the twisted undeformed gauge theory. 

\subsection{Representations of the Yangian and Wilson operators}
 We are interested in the observables (or operators) of this theory, and the structure of the operator product.   Our main theorem says that the Yangian encodes the structure of the operator product of the twist of the Yangian deformation  of the $N=1$ supersymmetric gauge theory.   The most powerful version of this statement is a little abstract, so I will start by explaining a concrete corollary.

Consider the $N=1$ theory on $\C \times \C^\times$, where $\C$ has coordinate $z$ and $\C^\times$ has coordinate $w$.   Then, we can consider as above the Yangian deformation, involving the Chern-Simons term $\frac{1}{2 \pi i}\d z \op{CS}(A)$, and the twist of this deformation.    The twisted theory is holomorphic in the $z$-direction and topological in the $w$-direction.

We are interested in circle operators or observables in this theory, supported on the circles where $z$ is fixed at $z_0$ and $\abs{w} = r$.  We will call the space of such observables $\Obs( z_0 \times \{\abs{w} = r\})$, where $\Obs$ is short for observables\footnote{Strictly speaking, the theory of factorization algebras only assigns cochain complexes to open subsets.  In the introduction we will be more informal, and talk at various points about observables supported on closed subsets.  Precise definitions of what I mean by observables on closed subsets are given in the body of the paper.}.

The BRST/BV differential makes $\Obs(z_0 \times \{\abs{w} =r \})$ into a cochain complex; the space of physically relevant observables is $H^0$ of this cochain complex.  This cochain complex forms part of a rich structure called a \emph{factorization algebra}, which I will describe later.  For now, we are only interested in a small part of this structure, namely the operator product on the spaces $\Obs(z_0 \times \{\abs{w} = r\})$ for fixed $z_0$. 

 The operator product in the $w$-direction gives the complex $\Obs(z_0 \times \{\abs{w} = r\})$ the structure of an associative algebra.   Let me explain this briefly.  The fact that our theory is topological in the $w$-direction tells us that this space of operators is independent of $r$.  Then, if $\alpha,\beta \in \Obs( z_0 \times \{\abs{w} = r\})$ are such operators, we define the product $\alpha \cdot \beta$ as follows.  View $\alpha$ as an operator on $z_0 \times\{\abs{w} = r\}$, and $\beta$ is an operator on $z_0 \times \{\abs{w} = r'\}$ where $r' > r$.  Then, the operator product of $\alpha$ and $\beta$, as $r' \to r$, gives a new operator in $\Obs(z_0 \times \{\abs{w} = r\})$, which represents the product $\alpha \cdot \beta$. 

The first theorem I will state identifies this associative algebra with an associative algebra constructed from the Yangian.   Let $C(\g)$ be the space of $\C[[\hbar]]$-linear maps
$$
l : Y(\g) \to \C[[\hbar]]
$$
satisfying
$$
l ([a,b] ) = 0
$$
for all $a,b \in Y(\g)$.   The coproduct on the Yangian makes this space $C(\g)$ into an associative sub-algebra of the dual of the Yangian, $Y^\ast(\g)$. 
\begin{theorem*}
\label{theorem_circle_observables}
There is an isomorphism of associative algebra
$$
H^0 (\Obs( z_0 \times \{\abs{w} = r\} )) \iso C(\g) 
$$
between $C(\g)$ and the space of physical circle operators in the twisted, deformed $N=1$ gauge theory. 
\end{theorem*}
As a corollary, we find the following.
\begin{corollary*}
\begin{enumerate}
\item Every finite dimensional representation $V$ of the Yangian gives a circle operator 
$$
\chi_V \in H^0(\Obs(z_0 \times \{\abs{w} = r\} )) .
$$
\item Tensor product of such finite dimensional representations (using the Hopf algebra structure on the Yangian) corresponds to the associative operator product on $H^0(\Obs(z_0 \times \{\abs{w} = r\} ))$:
$$
\chi_{V \otimes W} = \chi_V \cdot \chi_W. 
$$ 
\end{enumerate}
\end{corollary*}
\begin{proof}
If $V$ is a finite-dimensional representation of $Y(\g)$, then taking trace in $V$ gives a linear map $\op{Tr}_V : Y(\g) \to \C[[\hbar]]$ satisfying $\op{Tr}([a,b]) = 0$. Thus, $\op{Tr}_V$ is an element of the algebra $C(\g)$, which the theorem identifies with the algebra of circle operators.  The second statement in the corollary is immediate: the tensor product of $Y(\g)$-modules and the product on $C(\g)$ are both defined using the coproduct on $Y(\g)$. 
\end{proof}
\begin{remarks}
\begin{enumerate}
\item
 I use a completed version of the Yangian, which quantizes $U(\g[[z]])$ rather than $U(\g[z])$ as is more usual.   Finite-dimensional representations of this completed Yangian are representations of the ordinary one where, modulo $\hbar$, $z^k \g$ acts trivially for sufficiently large $k$.  Also, by finite-dimensional I mean free of finite rank as $\C[[\hbar]]$-modules. 

\item As usual, one expects that observables (or operators) for the twisted theory correspond to $Q$-closed observables of the physical theory, modulo $Q$-exact observables.  (More precisely, one expects to find a spectral sequence starting with observables of the untwisted theory, with their BRST differential, and converging to observables of the twisted theory). 

At the classical level, the circle operator associated above to a representation of $\g$ (viewed as a representation of the classical limit of the Yangian $U(\g[[z]])$) correspond to ordinary Wilson operators: they come from the holonomy of a connection in the $w$-plane.  In the deformed, untwisted theory, this connection is given by adding to the connection $A$ a one-form which is one of the components of $B$ times $\d w$.  This combined connection, when restricted to the $w$-plane, is invariant under our chosen supercharge.  (This is in contrast to the undeformed untwisted theory, which has no supersymmetric Wilson operators). 

One thing that is missing from the story is a rigorous proof that the deformation of the $N=1$ theory (including these supersymmetric Wilson operators) exists at the quantum level. In this paper I only prove existence at the quantum level after twisting.  Si Li has constructed the untwisted $N=1$ super-symmetric gauge theory using the techniques of \cite{Cos11, CosGwi11}.  Hopefully, his approach can be generalized to include the deformation of the $N=1$ gauge theory we use, thus providing a rigorous proof that the Yangian computes information about the operator products in the untwisted theory. 

\item The algebra $C(\g)$ is the $\C[[\hbar]]$-linear dual to the $0^{th}$ Hochschild homology of $Y(\g)$.  The entire differential graded algebra $\Obs(z_0 \times\{ \abs{w} =r\})$ is isomorphic to the dual of the whole Hochschild chain complex of $Y(\g)$. 
\end{enumerate}
\end{remarks}

\subsection{Wilson loops and lattice models}
The Yangian plays a key role in the theory of integrable lattice models: for instance, $Y(\mf{sl}_2)$ controls the $XXX$ model, and Yangians for other semi-simple Lie algebras play the same role in more  general spin-chain models.   It turns out that the partition functions of spin-chain lattice models coincide with expectation values of certain Wilson operators for the twisted, deformed $N=1$ gauge theory.  

We have seen that the twisted, deformed theory can be defined at the quantum level on any complex surface $X$ with a holomorphic volume form and a closed holomorphic $1$-form.  Note that a closed holomorphic $1$-form is the same as a holomorphic vector field preserving the holomorphic symplectic form. 

It turns out that the theory can be defined more generally.  Suppose that $X$ is a complex surface with a reduced effective divisor $D$ and a trivialization of $K_X(2 D)$.  Thus, we have a meromorphic volume form on $X$ with a quadratic pole along $D$, and no other poles or zeroes. Suppose further that we have a holomorphic vector field $V$ on $X$, which preserves $D$ and the meromorphic volume form.  Then, we can consider a variant of the twisted, deformed $N=1$ theory where the solutions to the equations of motion are holomorphic $G$-bundles on $X$, trivialized along $D$, and equipped with a holomorphic (and therefore flat) connection in the direction given by the vector field $V$.  This variant of our theory has a unique perturbative quantization which respects certain symmetries. 

The example we are interested in is $\mbb{P}^1 \times E$, where $E$ is an elliptic curve. Let $z,w$ denote coordinates on $\mbb{P}^1$ and $E$, respectively.  The divisor we consider is $\infty \times E$; the holomorphic volume form is $\d z \wedge \d w$, which has quadratic poles along the divisor; and the vector field we consider is $\tfrac{\partial}{\partial w}$.    One can check in this case that there are no non-trivial stable solutions to the equations of motion.  Thus, in this case, perturbative calculations are exact.  

Away from $\infty$, the theory on $\mbb{P}^1_z \times E_w$ coincides with the compactification of the theory on $\C_z \times \C_w$ along the elliptic curve $E_w$. 

Given any finite-dimensional representation $V$ of the Yangian $Y(\g)$, there is a corresponding $2$-dimensional integrable lattice model, discussed for example in \cite{JimMiw94}. We have seen that representations of the Yangian give rise to Wilson operators.  We will now explain how expectation values of certain Wilson operators coincide with partition functions of the corresponding integrable lattice model.

Consider the theory on $\mbb{P}^1_z \times E_w$, with divisor $\infty \times E$ as above.  For every representation $V$ of the Yangian, let 
$\chi_V^a(z)$ be the Wilson operator associated to $V$ placed on the circle $z \times S^1_a$ where $S^1_a \subset E_w$ is an $a$-cycle. Similarly, let $\chi_V^b(z)$ be the operator placed on a $b$-cycle.  (It doesn't matter which cycle we use in a given isotopy class). 

Then, we can consider the expectation value
$$
\ip{\chi_V^{a_1}(z) \dots \chi_V^{a_m}(z) \chi_V^{b_1}(0) \dots \chi_V^{b_n}(0)  }_{\mbb{P}^1 \times E}
$$
where $a_1,\dots, a_m$ are disjoint $a$-cycles, $b_1,\dots, b_m$ are disjoint $b$-cycles, and $z \in \mbb{P}^1 \setminus \{0,\infty\}$. 

This expectation value coincides with the partition function of a certain two-dimensional integrable lattice model, which I will now describe.  Let $V$ be a representation of the Yangian.  Choose a basis for $V$. A state of the model is a way of labeling every edge of a lattice by a basis vector in $V$. 

There is an interaction at every vertex, which is determined by the $R$-matrix of the Yangian. This is an element
$$R \in Y(\g) \otimes Y(\g) [z^{-1}].$$
More properties of the $R$-matrix will be discussed later.  For now, all we need to know is that if $V$ is a representation of the Yangian, $R$ gives a $z$-dependent endomorphism $R_{V\otimes V}$ of $V\otimes V$. The matrix entries of $R_{V\otimes V}$ provide the interaction at the vertex in the lattice model. The parameter $z$ is called the spectral parameter of the model.

The simplest case of this construction is when $\g = \sl _2$ and $V$ is the fundamental representation of $\sl _2$, turned into a representation of the Yangian using the evaluation homomorphism.  This leads to the $XXX$ model; see \cite{JimMiw94}.  In this case, the $R$-matrix is the matrix 
$$
1 + \frac{\hbar c}{z} \in \op{End}(\C^2 \otimes \C^2)
$$
where $c \in \sl_2 \otimes \sl_2$ is the quadratic Casimir. 

Let us denote the partition function of this integrable lattice model on a doubly periodic $m\times n$ lattice by $Z( m,n; V ,z) $.  As usual, the partition function is defined to be the sum over all states, where a state is weighted by the product over all vertices of the interaction at that vertex.

\begin{theorem*}
We have
$$
\ip{\chi_V^{a_1}(z) \dots \chi_V^{a_m}(z) \chi_V^{b_1}(0) \dots \chi_V^{b_n}(0)  }_{\mbb{P}^1 \times E} = Z( m,n; V ,z)  \in \C[[\hbar]]. 
$$
\end{theorem*}

\begin{remark}
\begin{enumerate}
\item The normalization $\frac{1}{2 \pi i} \d zCS(A)$ in our Lagrangian appears so that the coordinate $z$ is the usual spectral parameter for the integrable lattice model.  Different normalizations are related by a linear change of coordinate on the $z$-plane.  
\item Our expectation value is computed in perturbation theory, which is why we find an answer in $\C[[\hbar]]$.  The integrable system, however, will give us an answer which is a polynomial rather than a power series in $\hbar$.   Thus, this theorem implies that the expectation value is a polynomial in $\hbar$. 

I believe that there are no instanton corrections to these expectation values.  If this is so, then this is an exact non-perturbative calculation.  
\end{enumerate}
\end{remark}

\subsection{ The Hilbert space and the transfer matrix.}
Let me now explain how other aspects of of theory of integrable  systems can be seen from the point of view of $4$-dimensional gauge theory.  We will see that the Hilbert space of our theory on $\mbb{P}^1 \times E$, with Wilson operators inserted, gives us the Hilbert space of the integrable lattice model. Wilson operators will give rise to the transfer matrix; commutativity of the transfer matrix (i.e.\ integrability) is a formal consequence of the fact that our theory is a mixture of topological and holomorphic. 

Before I explain this, I should remark that discussion of the Hilbert space lies a bit outside the language of factorization algebras \cite{CosGwi11} in which the main results of this paper are proved.  Factorization algebras only know about observables, or operators, of the theory.  The statements I will explain about the Hilbert space will be justified later by a positing that a version of the state-operator correspondence holds. 

Thus, let us consider our theory on $\mbb{P}^1 \times S^1$, where, as above, we use a divisor $D = \infty$ on $\mbb{P}^1$.   Let $V$ be a representation of the Yangian, and let $\chi_V$ be the corresponding Wilson operator. 
\begin{lemma*}
The Hilbert space of our theory on $\mbb{P}^1 \times S^1$ with insertions of the Wilson operator $\chi_V$ at $n$ distinct points in the circle $0 \times S^1$ is isomorphic to $V^{\otimes n}$.
\end{lemma*} 
In the absence of Wilson operators on the boundary, all observables on $\mbb{P}^1 \times S^1 \times [0,1]$ will act on the Hilbert space. If we include such Wilson operators, however, the algebra of observables is modified along $0 \times S^1 \times [0,1]$.  Observables which are supported away from $0 \in \mbb{P}^1$ will continue to act.  In particular, the Wilson operator $\chi_V(-z)$ placed on the circle $z \times S^1 \times \{\tfrac{1}{2}\}$ for any $z \in \C^\times$ will act on the Hilbert space $V^{\otimes n}$.   
\begin{theorem*}
The linear map
$$
\chi_V(z) : V^{\otimes n} \to V^{\otimes n}
$$
is the transfer matrix $T(z)$ of the spin-chain integrable system.  
\end{theorem*}
Note that this is consistent with the theorem above stating that the expectation value of a net of Wilson operators is the partition function of the spin-chain system: the partition function is the trace of $T(z)^m$.

Next, we need to show that the transfer matrices for different values of $z$ commute.  This will show that the lattice model is integrable. 

We can put the operator $\chi_V(z)$ on any of the parallel circles $z \times S^1 \times a$ for $a \in (0,1)$. Because our theory is topological in the $S^1 \times (0,1)$ factor, these operators are all the same.  We will let $\chi_V(z,a)$ denote the operator viewed as being on the circle $z \times S^1 \times a$. 

For another $z' \in \C^\times$, the composition $T(z) T(z')$ of the two transfer matrices is represented by the Wilson operator associated to the representation $V$ but placed on the disjoint union of two circles 
$$
\{z \times S^1 \times a\} \amalg \{z' \times S^1 \times a'\}
$$
with $a' < a \in (0,1)$.  The composition $T(z') T(z)$ has a similar representation, except that we take $a < a'$.  Since the two circles are disjoint, we can slide one past the other to conclude that 
$$
T(z) T(z') = T(z') T(z). 
$$

This proof of commutativity of the transfer matrix does not rely at all on the details of the model we are considering. Rather, this commutativity will hold any time we have a $4$-dimensional field theory which is topological in two directions and holomorphic in the other two.

In fact, in this situation, we always find a local transfer matrix,  i.e. one built up from interactions at a single vertex in the lattice. We can see this immediately by cutting $\mbb{P}^1 \times S^1 \times I$, with its $n$ vertical and one horizontal Wilson operator, into several copies of $\mbb{P}^1 \times I \times I$, each of which contains one vertical and one horizontal Wilson operator.   The Segal-style axioms of field theory guarantee that the transfer matrix is a product of local contributions coming from such square regions.

\subsection{Integrability in two dimensions}
We have seen that our twisted, deformed $N=1$ gauge theory leads to an integrable lattice model.  In this subsection I will explain how, upon compactification to two dimensions, it also leads to an integrable two-dimensional holomorphic field theory: i.e. to a vertex algebra with a maximal set of commuting operators.  The two appearances of integrability are closely related.   The results explained in this subsection only hold for $\g = \sl_n$ (but probably an expert on the Yangian would be able to lift this restriction). 

To understand this integrability, we need to consider our field theory compactified on a torus.  Thus, we let $E$ be an elliptic curve, and we work on $\C_z \times E$, where $w$ is a local coordinate on the elliptic curve $E$.  We use the one-form $\d z$ when defining the Chern-Simons term in the deformation of the $N=1$ gauge theory; so that when we twist, the theory is holomorphic in the $z$-direction  topological in the direction of the elliptic curve $E$.  After twisting, the theory does not depend on the complex structure on $E$. 

We will, as before, let $\Obs(z_0 \times E)$ denote the cochain complex of observables of the twisted deformed theory on $z_0 \times E$.  Because our theory is holomorphic in the $z$-direction, this cochain complex has a structure like that of a vertex algebra\footnote{Strictly speaking, we find an analog of axioms of a vertex algebra based on holomorphic functions rather than Laurent series.   More details are given in the body of the paper.}.   Let us choose generators $a$ and $b$ for $\pi_1(E)$, represented by circles $S^1_a$ and $S^1_b$ in $E$.  We have a sub-cochain complex 
$$
\Obs(z_0 \times S^1_a) \subset \Obs(z_0 \times E)
$$
of observables on the $a$-circle.  We have seen how to describe $H^0 ( \Obs(z_0 \times S^1_a))$ in terms of the Yangian.

Recall that a quantum-mechanical system is completely integrable if there is a maximal commutative sub-algebra of the algebra of observables.  We will find a similar structure: inside the vertex algebra $H^0 ( \Obs( z_0 \times E))$ we will find that $H^0(\Obs(z_0 \times S^1_a))$ is a maximal commutative sub-vertex algebra. (A vertex algebra is commutative if all operator products are non-singular). 

\begin{theorem*}
The subspace 
$$
H^0 ( \Obs(z_0 \times S^1_a)) \subset H^0 ( \Obs (z_0 \times E))
$$
is closed under the operator product in the $z$-direction. 

Further, the operator product between two elements 
$\alpha \in H^0 ( \Obs(z_0 \times S^1_a))$ and $\beta \in H^0(\Obs(z_1 \times S^1_a))$ is non-singular as $z_1$ approaches $z_0$, so that $H^0(\Obs(z_0 \times S^1_a))$ is a commutative sub-vertex algebra of $H^0(\Obs(z_0 \times E))$.  

Finally, if $\g = \sl_n$, $H^0(\Obs(z_0 \times S^1_a))$ is a maximal commutative sub-vertex algebra. 
\end{theorem*}
\begin{proof}
Let me sketch the proof of commutativity now (the argument is very similar to the argument for commutativity of the transfer matrix sketched above).  Full details are presented in the body of the text, where this argument is formalized using the language of factorization algebras.  Let $\alpha \in H^0(\Obs(z_0 \times S^1_a))$, and let $\beta \in H^0(\Obs^q(z_1 \times S^1_a))$.  The operator product expansion is an asymptotic expansion for the operator $\alpha \cdot \beta$ as $z_1$ approaches $z_0$.  

Let $S^1_{a'}$ be a circle obtained by displacing $S^1_a$ a small amount; so $S^1_{a}$ and $S^1_{a'}$ don't intersect. Since our theory is topological in the $w$-direction, the space of operators on $z_0 \times S^1_{a}$ and on $z_0 \times S^1_{a'}$ are the same, so we may as well view $\beta$ as an operator on $z_1 \times S^1_{a'}$.  Then, $\lim_{z_1 \to z_0} \alpha \cdot \beta$ has no singularities, because the operators $\alpha$ and $\beta$ have disjoint support even if we set $z_0 = z_1$.  
\end{proof}
\begin{remark}
\begin{enumerate}
\item We have seen that $H^0 ( \Obs(z_0 \times S^1))$ is the algebra $C(\g)$ of linear maps $l : Y(\g) \to \C[[\hbar]]$ satisfying $l([a,b]) = 0$.  Probably the algebra $C(\g)$ should be called the ``Bethe sub-algebra of the dual Yangian''.  
\item If $\g = \sl_n$, the algebra $C(\sl_n)$ has a concrete description.  We can apply any invariant polynomial on $\sl_n$ to the holonomy around the circle of a field configuration; this gives elements $P_2,\dots,P_n \in H^0 ( \Obs(z_0 \times S^1))$, where $P_i$ is the invariant polynomial $\op{Tr} A^i$.  Although the description in terms of holonomy is classical, I show that these observables can be defined at the quantum level also.   We can also take the derivative of these observables in the $z$-direction: classically, this corresponds to taking the derivative of the holonomy around the circle $z \times S^1$ in $z$.  Then, $C(\sl_n)$ is freely generated by $\partial_z^k P_i$, for $k = 0,\dots,\infty$ and $i = 2,\dots, n$.  
\item The restriction to type $A$ arises because I don't know how to describe $C(\g)$ explicitly for arbitrary $\g$: in particular, I don't know that it's ``big enough''.  Chervov \cite{ChervovMathOverflow} suggests that for arbitrary semi-simple $\g$, $C(\g)$ can be described in terms of invariant polynomials on $\g$ in the same way as in type $A$.  This would prove integrability for arbitrary semi-simple $\g$.  
\end{enumerate}
\end{remark}

\subsection{Factorization algebras}
All of these results follow from a sequence of abstract results phrased in the language of factorization algebras.  Factorization algebras are a language for talking about the operator products in quantum field theories. Here, I will give a brief introduction to the idea: a more thorough discussion is given later,  and many more details are available in the book \cite{CosGwi11}.  

\begin{definition*}
A prefactorization algebra $\F$ on $\R^n$ is an assignment to every open subset $U \subset \R^n$ a cochain complex $\F(U)$; and to every inclusion $U_1 \amalg \dots \amalg U_n \into V$ of disjoint open set $U_i$ into $V$, a cochain map
$$
\F(U_1) \otimes \dots \otimes  \F(U_n) \to \F(V),
$$ 
which is independent of the ordering chosen on the set $U_i$, and which satisfies the following associativity condition.  Let $U_i$ and $V_j$ be finite collections of open subsets of $\R^n$, where $U_i \cap U_{i'} = \emptyset$ and $V_{j} \cap V_{j'} = \emptyset$. Let $W$ be another open subset of $X$. Suppose that
$$
\amalg U_i \subset \amalg V_j \subset W.
$$
Then the following diagram commutes:
$$
\xymatrix{
\otimes_{i} \F(U_i) \ar[r] \ar[dr] & \otimes_{j} \F(V_j) \ar[d] \\
& \F(W) ,
}
$$
where the arrows are the structure maps of the prefactorization algebra. 

A factorization algebra is a prefactorization algebra satisfying a certain ``descent'' condition, saying that $\F(V)$ for any open subset $V$ is determined from the value of $\F$ on a sufficiently fine open cover of $V$. (The descent condition will not be used in this paper).
\end{definition*}
The cochain complex $\F(V)$ should be thought of as being the observables supported on the subset $V \subset \R^n$: that is, the measurements we can make in a field theory which only depend on the behavior of the fields on $V$.   

Let me explain the relationship between factorization algebras and the operator product expansion.  Suppose we have a translation-invariant factorization algebra $\F$ on $\R^n$, corresponding to a translation-invariant quantum field theory.  We are mainly interested in the cochain complexes $\F(D(x,r))$ associated to the disc of radius $r$ around $x \in \R^n$ (by the descent axiom, this encodes all the information about the factorization algebra).  The theory we are interested in is translation-invariant, so that we are given an isomorphism between $\F(D(x,r))$ and $\F(D(x',r))$ for any other $x' \in \R^n$; so that we can use the notation $\F_r = \F(D(x,r))$.  Then, the structure of a factorization algebra gives us product maps
$$
m_{x_1,x_2} : \F_{r_1} \times \F_{r_2} \to \F_{r_3}
$$
for every $x_1,x_2 \in D(0,r_3)$ in the disc of radius $r_3$ about the origin, with the property that $D(x_1,r_1)$ and $D(x_2,r_2)$ are disjoint and contained in $D(0,r_3)$.  This product varies smoothly with the positions of $x_1$ and $x_2$, and satisfies an associativity axiom which is a special case of the associativity axiom described above.   

The products $m_{x_1,x_2}$ encode the same data as the familiar operator product expansion.  The operator product expansion (when it exists\footnote{I don't know in general that such an asymptotic expansion exists, except for theories of a holomorphic nature.}) is an asymptotic expansion of $m_{x_1,x_2}$ as $x_1,x_2$ approach each other and the radii $r_1,r_2$ tend to zero.  It seems to me that the actual operator product -- encoded by the factorization algebra -- is more flexible and easier to work with than the operator product expansion.

\subsection{}  
The only factorization algebra we will consider in this paper is the factorization algebra on $\C_z \times \C_w$ of observables of the twisted, deformed $N=1$ gauge theory.  We will denote this factorization algebra by $\Obs$. 

\renewcommand{\F}{\Obs}
This is a factorization algebra on $\C_z \times \C_w$, which behaves like a vertex algebra in the $z$-direction and like the algebra of observables of a topological field theory in the $w$-direction.  Let us denote by $\F_{z = 0}$ the restriction of $\F$ to a factorization algebra on the complex line $z = 0$.   This means we consider observables of our theory which are supported on $z = 0$, and consider the operator products in the $w$-directions.

Because our theory is topological in the $w$-direction, the factorization algebra $\F_{z = 0}$ behaves like the observables of a topological field theory.  A naive analysis suggests that the observables of a topological field theory are simply a commutative algebra.  Indeed, the operator product does not depend on the distance between points, or on the direction in which points collide.  A more careful analysis, however, shows that the observables of a TFT on $\R^n$ form an $E_n$ algebra: this is an algebra which is only ``partially commutative''.  (More details on $E_n$ algebras are given in section \ref{e_2_hopf}).  This extra structure appears because of the subtle way in which the operator product depends on the direction in which points come together.  The operator product is independent of this direction, up to $Q$-exact terms ($Q$ denoting the BRST differential on observables).    Taking account of the way in which things are made $Q$-exact leads to this richer structure.    A theorem of Lurie \cite{Lur12} makes this intuition precise: he shows that factorization algebras on $\R^n$ satisfying an additional ``locally-constant'' axiom are the same as $E_n$ algebras. 

Thus, we see that the factorization algebra for the twist of the Yangian deformation, when restricted to the line $z = 0$, is an $E_2$ algebra.   A well-known result of Tamarkin \cite{Tam03b} shows that there is a close relationship between $E_2$ algebras and Hopf algebras: the Koszul dual of an $E_2$ algebra is a Hopf algebra.   

Let me briefly try to explain this result. If $A$ is a commutative algebra, then the category $A-\op{mod}$ of left $A$-modules is a symmetric monoidal category, with monoidal structure given by tensoring over $A$.  An $E_2$ algebra is partially commutative: for instance, if $A$ is an $E_2$ algebra, there is an isomorphism $A \iso A^{op}$, so that we can identify left and right $A$-modules.  It turns out that an $E_2$ algebra $A$ has just enough commutativity so that the dg category $A-\op{mod}$ of left $A$-modules is a monoidal category, with monoidal structure defined by $M \otimes_A N$ (where we view $N$ as a right $A$-module).  This monoidal structure is not, in general, symmetric.  

This suggests already that $E_2$ algebras and Hopf algebras are closely related: if $B$ is a Hopf algebra, then the category $B-\op{mod}$ of left $B$-modules is monoidal.  The key feature of Hopf algebras, however, is that the forgetful functor $F: B-\op{mod} \to \op{Vect}$ (sending a $B$-module to its underlying vector space) is monoidal.  The Tannakian formalism tells us that we can recover the Hopf algebra $B$ from the category $B-\op{mod}$ together with the monoidal functor $F$.

If $A$ is an $E_2$ algebra, the monoidal category $A-\op{mod}$ is not automatically equipped with a monoidal functor to the category $\op{dgVect}$ of cochain complexes.  We can get such a functor, however, if we choose an augmentation on $A$: this is a homomorphism of $E_2$ algebras from $A$ to the ground ring $\C$.   If $\phi : A \to \C$ is an augmentation, we get a monoidal functor $F : A-\op{mod} \to \op{dgVect}$ which sends a module $M$ to
$$
F(M) = M \otimes_A \C. 
$$The Koszul dual Hopf algebra $B$ to $A$ is characterized by the fact that there is an equivalence of monoidal dg categories $B-\op{mod} \simeq A-\op{mod}$, compatible with monoidal functors to $\op{dgVect}$.   Concretely, as an algebra, $B = \R \op{Hom}_A(\C, \C)$. 

The first main result of this paper is the following.
\begin{theorem*}
The Koszul dual of the $E_2$ algebra $\F_{z = 0}$ is the Yangian Hopf algebra $Y(\g)$. 

In particular, if $\op{Fin}(Y(\g))$ denotes the monoidal dg category of finite-rank $Y(\g)$-modules, and $\op{Perf}(\F_{z = 0})$ denotes the monoidal dg category of perfect (e.g. free of finite rank) $\F_{z=0}$-modules, then there is a quasi-equivalence of dg categories
$$
\op{Fin}(Y(\g)) \simeq \op{Perf}(\F_{z = 0}). 
$$
\end{theorem*}
\begin{remark}
For technical reasons, in the body of the paper we work mostly with a ``dual'' formulation of Koszul duality. This dual formulation associates a Hopf algebra $A^!$ to an $E_2$ algebra, in such a way that there is an equivalence of monoidal categories between $A$-modules and $A^!$-comodules. The formulation explained above amounts to replacing $A^!$ by its linear dual, which is a topological Hopf algebra.   The topological nature of $Y(\g)$ means the correct definition of $\op{Fin}(Y(\g))$ is not completely obvious; we need to use a certain category of complexes of topological modules with finite-rank cohomology. Details are given in the body of the paper. 
\end{remark}

It will be important later that we understand the augmentation in field-theoretic terms.  Our theory can be defined on $\mbb{P}^1_z \times \C_w$, where the bundles that appear are trivialized at $\infty$ in $\mbb{P}^1$.   Thus, for any $D \subset \C_w$ we have a cochain complex $\Obs(\mbb{P}^1 \times D)$.  The factorization algebra structure makes this into an $E_2$ algebra which we call $\Obs_{\mbb{P}^1}$.  The following lemma will be easy to verify.
\begin{lemma*}
There is a quasi-isomorphism of $E_2$ algebras
$$
\Obs_{\mbb{P}^1} \simeq \C[[\hbar]].
$$
\end{lemma*}
The factorization algebra structure gives us a natural map of $E_2$ algebras
$$
\Obs_{z = 0} \to \Obs_{\mbb{P}^1} \simeq \C[[\hbar]].
$$ 
This map is the augmentation we use to construct the Koszul dual to $\Obs_{z = 0}$. 

Every map of $E_2$ algebras leads to a functor of the corresponding monoidal categories.  The statement that $\Obs_{\mbb{P}^1}$ gives us the augmentation of $\Obs_{z = 0}$ tells us that the following diagram of monoidal categories commutes:
$$
\xymatrix{
\op{Perf}( \Obs_{z = 0} ) \ar[r]^{\simeq} \ar[d] & \op{Fin}(Y(\g)) \ar[d] \\
\op{Perf}( \Obs_{\mbb{P}^1}) \ar[r]^{\simeq} & \op{Perf}(\C[[\hbar]]).
}
$$
The bottom right arrow is the forgetful functor from the category of $Y(\g)$-modules to the category of dg $\C[[\hbar]]$-modules of finite rank. 

\subsection{}
The Koszul duality theorem stated above allows us to calculate the spaces of Wilson operators and surface operators in terms of the representation theory of the Yangian.  Consider our theory on $\C_z \times \C^\times_w$, as before. We can thus consider the space $\Obs(0 \times S^1_w)$ of observables on the circle $S^1_w =\{ \abs{w} = 1\}$.  As we vary the radius of the circle, we can view this space as part of a factorization algebra on $\R_{> 0}$, obtained from pushing forward the factorization algebra on $\C^\times$ to $\R_{> 0}$ along the radial projection.  Locally constant factorization algebras on $\R_{> 0}$ are the same as associative algebras. 

Abstract results about factorization algebras (proved by Lurie and Francis) imply that this space is the Hochschild homology of the $E_2$ algebra $\F_{z = 0}$.  Further, the Hochschild homology of any $E_2$ algebra is an associative algebra, and the associative product is the one described in the previous paragraph.  

It turns out that Hochschild homology plays well with Koszul duality (this result is due, in various different contexts, to John Francis, Jonathan Campbell, and Takuo Matsuoko; a version of this statement relevant to this context is proved here).  Thus, we find an isomorphism of associative algebras
$$
HH_\ast (\F_{z = 0} ) = H^\ast ( \Obs (0 \times S^1) ) = HH_\ast (Y(\g))^\vee 
$$
where on the right hand side we take the linear dual.  If $Y(\g)-\op{mod}$ refers to the category of $Y(\g)$-modules which are of finite rank over $\C[[\hbar]]$, we also find an isomorphism
$$
H^\ast (\Obs (0 \times S^1 )) = HH_\ast ( Y(\g)-\op{mod} ) . 
$$
The monoidal structure on $Y(\g)-\op{mod}$ induces an associative structure on $HH(Y(\g)-\op{mod})$, and this is an isomorphism of associative algebras.  Every $Y(\g)$-module has a character, which is an element of $HH_\ast(Y(\g)-\op{mod})$; this gives an alternative and equivalent way to associate an element of $H^\ast(\Obs(0 \times S^1))$ to a representation of the Yangian.

\subsection{}
There is a similar description of the space of surface operators, which will be essential in the proof that the operator product of Wilson operators coincides with the partition function of the integrable lattice model. 

So far, we have seen that the vertex algebra of surface operators $H^0 ( \Obs ( z_0 \times E))$ is completely integrable, but we have not given a more detailed description of this algebra.  In this section I will explain how to calculate $H^\ast(\Obs(z_0 \times E))$ from the Yangian. In the next section we will see that the operator product on $H^\ast(\Obs(z_0 \times E))$ is encoded by the $R$-matrix of the Yangian. We have already seen one version of this statement in the relationship between Wilson loops and integrable lattice models. 

Let $HC_\ast(Y(\g))$ denote the Hochschild chain complex of the Yangian, and $HC_\ast(Y(\g))^{\vee}$ its $\C[[\hbar]]$-linear dual.  Because the Yangian is a coalgebra, and the functor of taking Hochschild chains is symmetric monoidal (in a homotopical sense), the chain complex $HC_\ast(Y(\g))^{\vee}$ has the structure of a differential graded algebra.  Thus, we can form the Hochschild homology of the dg algebra $HC_\ast(Y(\g))^{\vee}$. 
\begin{proposition*}
There is an isomorphism of graded $\C[[\hbar]]$-modules 
$$
H^\ast ( \Obs(z_0 \times E)) \iso HH_\ast ( HC_\ast(Y(\g))^{\vee} ).
$$
\end{proposition*}
Another formulation of this result is as follows.  Let $\op{Fin}(Y(\g))$ denote the dg category consisting of those dg modules over $Y(\g)$ whose cohomology is finitely generated as a $\C[[\hbar]]$-module.    Then, $\op{Fin}(Y(\g))$ is a monoidal dg category, so its Hochschild chain complex is a homotopy-associative algebra.
\begin{proposition*}
There is an isomorphism\footnote{The statement of this proposition is a little imprecise: the precise statement involves taking Hochschild homology in a certain category of filtered cochain complexes.  I will ignore all such irritating technicalities throughout the introduction. } of graded $\C[[\hbar]]$-modules 
$$
H^\ast ( \Obs(z_0 \times E)) \iso HH_\ast ( HC_\ast(\op{Fin}(Y(\g))) .
$$
\end{proposition*}

\subsection{The $R$-matrix and the operator product in the $z$-direction}
So far, we have understood that the operator product of our theory in the $w$-direction encodes the Yangian as a Hopf algebra.  In this section we will see that we can understand the operator product in the $z$-direction using another structure on the Yangian: Drinfeld's universal $R$-matrix. 

The universal $R$-matrix is an element\footnote{The $R$-matrix is usually considered as a Laurent series in $\lambda^{-1}$, but because we are using a completed version of the Yangian we can view it as a Laurent series in $\lambda$. }
$$
R (\lambda) \in Y(\g) \otimes Y(\g)((\lambda))
$$
satisfying various identities detailed in Section 12 of \cite{Dri87} and in  Theorem 12.5.1 of \cite{ChaPre95}. I will explain a categorical interpretation of the $R$-matrix.  

Recall that $Y(\g)$ quantizes $U(\g[[z]])$.  Sending $z \mapsto z + \lambda$ gives an embedding $U(\g[[z]]) \into U(\g[[z]])((\lambda))$.  This embedding quantizes to an embedding of Hopf algebras
$$
T_\lambda : Y(\g) \into Y(\g)((\lambda)).
$$
We will let $T_0 : Y(\g) \into Y(\g)\otimes \C((\lambda))$ denote the obvious embedding, sending $\alpha \mapsto \alpha \otimes 1$.  

If $V$ is a $Y(\g)$-module, we will let 
$$
T_\lambda(V) = Y(\g)((\lambda)) \otimes_{Y(\g)} V
$$
where the tensor product is taken using the homomorphism $T_\lambda : Y(\g) \to Y(\g)((\lambda))$.   We let $T_0 (V) = V((\lambda))$, i.e.\ the extension of $V$ to a $Y(\g)((\lambda))$-module using the homomorphism $T_0 : Y(\g) \to Y(\g)((\lambda))$.   

Define a functor 
\begin{align*}
F_R : \op{Fin}(Y(\g)) \times \op{Fin}(Y(\g)) &\to \op{Fin}(Y(\g)((\lambda)) )\\
V \times W & \mapsto T_0(V) \otimes T_\lambda(W).
\end{align*}
(Here $\op{Fin}(Y(\g)((\lambda)) )$ denotes the category of $Y(\g)((\lambda))$-modules which are of finite rank over $\C((\lambda))[[\hbar]]$: this is a monoidal category over $\C((\lambda))$). 

\begin{theorem*}
The universal $R$-matrix $R(\lambda)$ encodes the natural isomorphism and coherences making the functor $F_R$ monoidal. 
\end{theorem*}
The converse is also true: we can construct an $R$-matrix from any natural isomorphism making the functor $m_\lambda$ monoidal.  Equations satisfied by the $R$-matrix correspond to the coherence conditions that a monoidal functor must satisfy. 

This functor will encode the operator product in the $z$-direction.  
\begin{theorem*}
The operator product expansion in the $z$-direction gives rise to a map of $E_2$ algebras
$$
\Obs_{z = 0} \otimes \Obs_{z = 0} \to \Obs_{z = 0} ((\lambda))
$$
(where $\lambda$ is a coordinate on $\C_z$).  

This homomorphism of $E_2$ algebras induces a monoidal functor
$$
F_{OPE} : \op{Perf}(\Obs_{z = 0} ) \times \op{Perf}( \Obs_{z = 0} )\to \op{Perf} ( \Obs_{z = 0} ((\lambda)) ).
$$ 
The following diagram of monoidal categories and functors commutes:
$$
\xymatrix{
\op{Perf}(\Obs_{z = 0} ) \times \op{Perf}( \Obs_{z = 0} )\ar[r]^{F_{OPE}} \ar[d]_{\simeq} & \op{Perf} ( \Obs_{z = 0} ((\lambda)) ) \ar[d]_{\simeq} \\
\op{Fin}(Y(\g)) \times \op{Fin}(Y(\g))  \ar[r]^{F_R} &  \op{Fin}(Y(\g)((\lambda)) )
}
$$
The vertical functors are the equivalences discussed earlier. 
\end{theorem*}
Since the $R$-matrix encodes the bottom horizontal functor, we see that the OPE in the $z$-direction of our theory is completely encoded by the $R$-matrix. The proof of this result uses Drinfeld's theorem that the $R$-matrix is unique. 

This result allows us to calculate the operator product expansion on surface operators.  We have seen that there is an isomorphism
$$
H^\ast(\Obs(0 \times E)) = HH_\ast^{(2)} (\op{Alg}(\op{Fin}(Y(\g))) ).
$$
Higher Hochschild homology of a $2$-category is functorial, so that applying the functor $F_R$ gives us a map
$$
HH_\ast^{(2)} (\op{Alg}(\op{Fin}(Y(\g))) ) \otimes HH_\ast^{(2)} (\op{Alg}(\op{Fin}(Y(\g))) ) \to HH_\ast^{(2)} (\op{Alg}(\op{Fin}(Y(\g))) ) ((\lambda)).
$$
This map is the operator product in the $z$-direction of surface operators.

\subsection{Higher-categorical interpretation of surface operators}

The monoidal category $\op{Fin}(Y(\g))$ can be viewed as a $2$-category with one object, whose $1$-morphisms are objects of  $\op{Fin}(Y(\g))$.  Let us call this $2$-category $\op{Fin}^{(2)}(Y(\g))$. 

Ayala-Rozenblyum \cite{AyaRoz13} and Morrison-Walker \cite{MorWal10} have developed a theory of higher Hochschild homology for $2$-categories. I will use the theory developed by Ayala-Rozenblyum, as Morrison-Walker's theory only applies to pivotal two-categories.  For the $2$-category associated to a monoidal category,  the Ayala-Rozenblyum construction gives the iterated Hochschild homology of the monoidal category discussed above.  Let $HH^{(2)} (\op{Fin}^{(2)} (Y(\g)))$ be the Ayala-Rozenblyum higher Hochschild homology applied to the $2$-category $\op{Fin}^{(2)} (Y(\g))$. We thus have an isomorphism
$$
HH^{(2)} ( \op{Fin}^{(2)} (Y(\g)) ) \simeq H^\ast (\Obs (z_0 \times E)).
$$
I will use this interpretation in our proof of the relationship between Wilson operators and lattice models, which I will now sketch.

\subsection{Sketch of the proof of the relationship between lattice models and Wilson operators}
The relationship between the operator product of Wilson operators and tensor product of modules for the Yangian allows us to reduce the calculation of the expectation value of a net of Wilson operators on $a$- and $b$-cycles to the case when we have one Wilson operator $\chi_{V^{\otimes n}}(z, S^1_a)$ on the $a$-cycle of $z \times E_w$, and another Wilson operator $\chi_{V^{\otimes m}}(0, S^1_b)$ on the $b$-cycle of $0 \times E_w$.  We thus need to prove the following.
\begin{proposition*}
For any two representations $V,W$ of the Yangian, we have
$$
\ip{\chi_V(z, S^1_a), \chi_W(0, S^1_b}_{\mbb{P}^1 \times E} = \op{Tr}_{V \otimes W} (R_{V,W}(z)) 
$$
is the trace of the $R$-matrix
$$
R_{V,W} (z) : V \otimes W \to V \otimes W. 
$$
\end{proposition*}
The proof of this relies on Ayala-Rozenblyum's model for higher Hochschild homology.  In their model, applied to the two-category $\op{Fin}^{(2)} (Y(\g))$ associated to the monoidal category $\op{Fin}(Y(\g))$, a $0$-chain arises from the following data.
\begin{enumerate}
\item A $1$-cell complex $\Gamma$ embedded in the two-torus $T^2$, with certain framing data. In particular, every $1$-cell is oriented.
\item Every component of $T^2 \setminus \Gamma$ is labeled by an object of the $2$-category. In this case, there is only one object. 
\item Every $1$-cell in the graph is labeled by a $1$-morphism, which in this case is an object of $\op{Fin}(Y(\g))$.  The direction the one-morphism goes is determined by orientation on the $1$-cell.  
\item Every $0$-cell is labeled by a $2$-morphism, relating the tensor product of the $1$-morphisms on the incoming and outgoing $1$-cells attached to this $0$-cell. 
\end{enumerate}
In particular, suppose we have an $a$-cycle $S^1 \subset T^2$. We can label the complement of this circle by the trivial $Y(\g)$-algebra $\C[[\hbar]]$, and label the circle by a representation $V$ of $Y(\g)$.  This gives a zero-cycle in $HH^{(2)} (\op{Fin}^{(2)} (Y(\g))$.  This zero-cycle represents the Wilson operator we discussed earlier.  Similarly, we can label a $b$-cycle by a representation of the Yangian, leading to a Wilson operator. 

We have seen that the operator product in the $z$-direction leads to a monoidal functor
$$
F_{OPE} : \op{Fin}(Y(\g)) \times \op{Fin}(Y(\g)) \to \op{Fin}(Y(\g)((z)) ),
$$ 
and that this functor is encoded by the $R$-matrix. Applying this functor to the Ayala-Rozenblyum higher Hochschild homology gives a map
$$
HH^{(2)}(F_{OPE}) : HH^{(2)} ( \op{Fin}^{(2)}(Y(\g)) ) \otimes  HH^{(2)} ( \op{Fin}^{(2)}(Y(\g) )) \to  HH^{(2)} ( \op{Fin}^{(2)}(Y(\g) ) ) ((z)).
$$
This map is the operator product expansion for surface operators. 

In particular, we can apply $HH^{(2)}(F_{OPE})$ to Wilson operators on $a$ and $b$-cycles, to get an element
$$
HH^{(2)}(F_{OPE}) (\chi_V(z, S^1_a) \otimes \chi_W(0, S^1_b) ) \in HH^{(2)}(  \op{Fin}^{(2)} (Y(\g))  ((z)).
$$
A small calculation shows that this element is represented by the union of $S^1_a$ and $S^1_b$, with the following labels: on $a$-cycle we put the representation $T_0(V)$ of the Yangian $Y(\g)((z))$. On the $b$-cycle we put the representation $T_z(W)$.  On the intersection, we put the $R$-matrix, which gives a map of $Y(\g)((z))$-modules
$$
\sigma \circ R : T_0(V) \otimes T_z(W) \to T_z(W) \otimes T_0(V).
$$
(Here $\sigma$ is the isomorphism which flips the factors).  The isomorphism $\sigma \circ R$ can be interpreted as a $2$-morphism in the $2$-category associated to $\op{Fin}(Y(\g))$.

This nearly completes the proof that the expectation values of Wilson operators coincides with the partition function of lattice models.   So far, we have computed the operator product expansion of Wilson operators. To compute the partition function on $\mbb{P}^1 \times E$, we need to apply the map
$$
H^\ast (\Obs (0 \times E)) \to H^\ast (\Obs ( \mbb{P}^1 \times E) ) = \C[[\hbar]].
$$
This map corresponds to the forgetful monoidal functor
$$
\op{Fin}(Y(\g)) \to \op{Fin}(\C[[\hbar]])
$$
to the category of finite-rank $\C[[\hbar]]$-modules. 

Thus, we find that the expectation value $\ip{\chi_V(z, S^1_a), \chi_W(0, S^1_b}_{\mbb{P}^1 \times E}$ coincides with a class in $HH^{(2)} (\op{Fin}^{(2)}(\C[[\hbar]]))((z))$ described by putting $V$ on the $a$-cycle, $W$ on the $b$-cycle, and the map
$$
\sigma \circ R : V \otimes W \to W \otimes V ((z))
$$
on the intersection.  

Now,
$$
HH^{(2)} (\op{Fin}^{(2)} (\C[[\hbar]]))((z)) = \C[[\hbar]]((z))
$$
and it is easy to compute that this element is $\op{Tr} (R_{V \otimes W})$. 

\subsection{Extended TFTs and the Hilbert space}
In this subsection I will sketch a derivation of the statements about the Hilbert space and the transfer matrix presented earlier, from the abstract results I have explained so far.

The formalism of factorization algebras doesn't, in general,  know about the Hilbert space: it only knows about the space of operators or observables.  However, in some cases, one can reconstruct the Hilbert space from the space of operators.  A famous example is in $2$-dimensional conformal field theory, where the state-operator correspondence tells us that the Hilbert space for a circle is the space of operators on a disc bounding that circle.  Thus, the factorization algebra associated to a $2$-dimensional conformal field theory knows about the Hilbert space.

More generally, in Segal's axiomatic approach to quantum field theory \cite{Seg99a, Seg04} he defines the space of observables on a Riemannian ball as being the Hilbert space associated to the sphere bounding that ball.  (Observables on a point are obtained as a limit as the radius of the sphere tends to zero). 

As a warm up for our construction, let us consider a $2$-dimensional holomorphic field theory defined on Riemann surfaces equipped with a nowhere-vanishing holomorphic $1$-form. In this case, we would expect that a Hilbert space assigned to any germ of an annulus equipped with a holomorphic $1$-form, and these Hilbert spaces will all be different. 

The state-operator correspondence would only be expected to hold for the case of a holomorphic $1$-form which extends to a nowhere-vanishing holomorphic $1$-form on the disc. That is, we should be able to identify the Hilbert space associated to $(S^1, \d z)$ with the space of operators for the disc $(\abs{z} \le 1, \d z)$. 

Let us now consider how this applies in our example. We will view our twisted, deformed $N=1$ gauge theory as a holomorphic $2$-dimensional field theory valued in topological $2$-dimensional field theories.  Thus, we will think of our theory as a twice-categorified $2$-dimensional holomorphic field theory.  That is, we expect to find something like a $2$-dimensional holomorphic field theory where the Hilbert space is a $2$-category instead of a vector space.  In the holomorphic direction, the theory is defined on Riemann surfaces equipped with a non-vanishing holomorphic $1$-form (or more generally, a meromorphic $1$-form with quadratic poles and no zeroes). 

Thus, we would expect that the Hilbert $2$-category assigned to the circle $(S^1, \d z)$ is associated to observables on the disc $(\abs{z} \le 1, \d z)$.  This space of observables forms an $E_2$ algebra, as we discussed earlier: and this $E_2$ algebra is Koszul dual to the Yangian. 

$E_2$ algebras (and Hopf algebras) are closely related to $2$-categories.  In particular, we have a $2$-category $\op{Fin}^{(2)} (Y(\g))$ with one object, and whose $1$-morphisms are representations of the Yangian of finite rank over $\C[[\hbar]]$. 

We thus make the following definition:
\begin{definition*}
The Hilbert $2$-category associated to the circle $(\abs{z} = \eps, \d z)$, where $\eps$ is an infinitesimally small parameter, is the $2$-category $\op{Fin}^{(2)}(Y(\g))$. 
\end{definition*}
We need to use an infinitesimally small circle because the Yangian is the object associated to a formal disc around $0$. We could also view this Hilbert $2$-category as being associated to the formal punctured disc.  

The reader who prefers to work with circles of non-infinitesimal radius can modify this definition as follows.  The general theory of factorization algebras produces an $E_2$ algebra $\Obs_{D_r}$ associated to the disc $D_r \subset \C_z$ of radius $r$.  Taking $r \to 0$ produces the $E_2$ algebra Koszul dual to the Yangian.  The larger $E_2$ algebra associated to a disc of positive radius has a monoidal category of modules which contains $\op{Fin}(Y(\g))$ as a full dg monoidal subcategory.  This larger monoidal category is what should be associated to a circle of positive radius. Similar remarks hold for other constructions in this subsection. 

We would expect $2$-dimensional holomorphic cobordisms to represent various $2$-functors.  In our story, we are allowed to consider Riemann surfaces with a meromorphic $1$-form with quadratic poles and no zeroes.  

We have seen that the $E_2$ algebra of observables on $(\mbb{P}^1, \d z)$ is the trivial $E_2$ algebra $\C[[\hbar]]$, and that the map
$$
\Obs_{z = 0} \to \Obs_{\mbb{P}^1} \simeq \C[[\hbar]]
$$
corresponds to the forgetful monoidal functor
$$
\op{Fin}(Y(\g)) \to \op{Fin}(\C[[\hbar]]). 
$$
This suggests the following:
\begin{definition*}
The $2$-dimensional cobordism $(\mbb{P}^1 \setminus \{\abs{z} < \eps\}, \d z)$ is associated to the forgetful functor of $2$-categories
$$
\op{Fin}^{(2)}(Y(\g)) \to \op{Fin}^{(2)}(\C[[\hbar]]). 
$$
\end{definition*}
Again, we view $\eps$ as being an infinitesimally small parameter. 

We have seen that the operator product expansion in the $z$-direction on the algebra of observables is encoded by the $R$-matrix. More precisely, the $R$-matrix gives rise to a monoidal functor
$$
F_R : \op{Fin}(Y(\g)) \times \op{Fin}(Y(\g)) \to \op{Fin}(Y(\g)((\lambda))) 
$$
which encodes the OPE.  

Because $F_R$ is a monoidal functor, it can be viewed as a functor of $2$-categories, which I call $F_R^{(2)}$.  We thus arrive at the following picture:
\begin{definition*}
Let $\lambda$ be a point in the formal punctured disc.  Then, the $2$-dimensional cobordism
$$
\{ \abs{z- \lambda} \ge \eps, \abs{z} \ge \eps, \abs{z} \le \eps' \}
$$
(where, again, $\eps,\eps'$ should be interpreted as being very near $0$) is associated to the $2$-functor $F_R^{(2)}$. 
\end{definition*}

Formal arguments using the cobordism hypothesis of Lurie, or the Ayala-Rozenblyum \cite{AyaRoz13} theory of composition homology, allows us to calculate from these definitions the Hilbert space associated to various other configurations.  

The $2$-category $\op{Fin}^{(2)}(Y(\g))$ has one object.  There is a $1$-category associated to the configuration $(\abs{z} = \eps, \d z) \times [0,1]$, where the interval is in the topological direction, and the two ends of the interval are labeled by the unique object of $\op{Fin}^{(2)} (Y(\g))$.  This $1$-category is the category $\op{Fin}(Y(\g))$.  The monoidal structure on this category arises from composition of the $1$-dimensional cobordism $[0,1]$ with itself. 
We want to include Wilson operators into the field theory. By analogy with Chern-Simons theory, we make the following definition. \begin{definition*}
The cobordism $(\abs{z} < \eps, \d z) \times [0,1]$, with the boundary of a Wilson operator $V$ placed at a point in the interval $0 \times [0,1]$, is associated to the object $V \in \op{Fin}(Y(\g))$. 
\end{definition*}

These definitions are enough to derive the following.
\begin{lemma*}
The Hilbert space associated to $\mbb{P}^1 \times S^1$, with the boundaries of $n$ Wilson operators on points in $0 \times S^1$, each associated to the representation $V$, is $V^{\otimes n}$.
\end{lemma*}
\begin{proof}
This follows immediately by decomposing $\mbb{P}^1 \times S^1$ into the various elementary pieces described above.
\end{proof}
\begin{lemma*}
Consider the $\mbb{P}^1 \times S^1 \times [0,1]$, with $n$ Wilson lines placed on the intervals $0 \times \{\theta = 2 \pi k / n\} \times [0,1]$, and a single Wilson operator placed on the circle $z \times S^1 times \tfrac{1}{2}$.  All the Wilson operators are associated to the representation $V$.

Then, the linear map $V^{\otimes n} \to V^{\otimes n}$ assigned to this configuration is the transfer matrix of the integrable system associated to the representation $V$ of the Yangian.
\end{lemma*}
\begin{proof}
Again, this follows formally from decomposing this configuration into elementary pieces. The main point is that the transfer matrix is defined as a trace of the composition of $n$ copies of the $R$-matrix acting on $V \otimes V$, where the trace and composition are defined over the first factor.  The $R$-matrix arises, as we have seen above, from the holomorphic pair of pants. 
\end{proof}

\subsection{Generalizations}
Let us call the twisted, deformed $N=1$ gauge theory the Yangian theory.  Then, the Yangian theory makes sense for every dg Lie algebra with an invariant pairing, and more generally for every $L_\infty$ algebra with an invariant pairing.  The construction of the Yangian theory at the quantum level also works in this generality. 
It turns out that the Yangian theory for various different Lie algebras can be interpreted in terms of twisted $N=2$ and $4$ gauge theories.
\begin{claim}
There is a twist of the $N=2$ pure gauge theory with Lie algebra $\g$, which gives the Yangian theory for $\g \oplus \g^\ast$.  There is a twist of the $N=4$ pure gauge theory which gives the Yangian theory for $\g[\eps,\delta]$ where $\eps$ is of cohomological degree $-1$ and $\delta$ is of cohomological degree $1$.  

Note that this is \emph{just a twist}: no deformation is needed when there is at least $N=2$ supersymmetry. 
\end{claim}
The twist of the $N=2$ gauge theory is that studied by Kapustin \cite{Kap06}. 

These results can be verified using the techniques of \cite{Cos11b}, which relies heavily on the twistor formulation of supersymmetric gauge theory developed in on \cite{BoeMasSki07}.   By ``claim'' I mean that I believe I know how to prove these results, but haven't worked out all details.  

The $S$-duality for the $N=4$ gauge theory \cite{KapWit06} suggest that one can hope for an $S$-duality for the Yangian based on $\g[\eps,\delta]$.   I hope to discuss this in future work.

There is a further generalization of this story to the $N=2$ $\sigma$-model.  Let $X$ be a hyper-K\"ahler manifold. Then, standard physics ideas say that one can define an $4$-dimensional $\sigma$-model with target $X$, which has $N=2$ supersymmetry.

Using the techniques of \cite{Cos11a}, one can encode $X$ by a curved $L_\infty$ algebra $\g_X$ over the base ring of forms on $X$.The fact that $X$ is holomorphic symplectic means that $\g_X$ has an invariant pairing of degree $-2$.  Thus, we can define the Yangian field theory based on $X$ by using the $L_\infty$ algebra $\g_X$. The fact that there is an invariant pairing of degree $-2$ means that the parameter $\hbar$ must be given cohomological degree $-2$.      

The results of this paper can be applied to construct this field theory at the quantum level.  Observables of this field theory form an $E_2$ algebra deforming functions on the space of maps from the formal disc to $X$. This $E_2$ algebra comes (as above) from a factorization algebra on $\C_z \times \C_w$ which is holomorphic in the $z$-direction and locally constant in the $w$-direction. 

If we fix a point in $X$, then we can turn the $E_2$ algebra into a Hopf algebra.  In this way, we find a family of integrable lattice models over $X$. It would be interesting to compute this object in examples. 
\begin{conjecture}
The Yangian field theory  based on $X$ is a twist of the $4$-dimensional $N=2$ $\sigma$-model with target $X$. 
\end{conjecture}
This construction further generalizes to the case when $X$ has a Hamiltonian action of a complex Lie group $G$. In this case, I believe we find a partial twist of the gauged $N=2$ $\sigma$-model.

\subsection{Mathematical results}
The results I have stated so far have mainly been applications of the representation theory of the Yangian to understanding supersymmetric gauge theory.  However, the flow of information goes both ways: thinking of the Yangian as arising from a $4$-dimensional field theory also gives new insight into the Yangian.   I will state some results and conjectures along these lines, even though many of the details will need to be worked out in subsequent publications. 

The first result is the following.  
\begin{claim}
The monoidal category $Y(\g)-\op{mod}$ of modules over the Yangian has the structure of a vertex algebra in monoidal categories. 
\end{claim}
This result is work-in-progress with Josh Shadlen. 

Concretely, this means that we have an operator product between objects of $Y(\g)-\op{mod}$ which behaves like a categorified form of the operator product from the theory of vertex algebras.  If $V,W \in Y(\g)-\op{mod}$, we can take the operator product between $V$ placed at $0$ and $W$ placed at a point $\lambda$ near $0$ in $\C$.  The operator product is the $Y(g)((\lambda))$-module 
$$
T_0(V) \otimes T_\lambda(W) 
$$
in the notation introduced earlier.  

The content of this claim is that this operator product satisfies a categorified analog of the axioms of a vertex algebra.   The associativity axiom in the theory of vertex algebras follows from the Yang-Baxter equation satisfied by the $R$-matrix.   This type of object has been considered before, by Gaitsgory \cite{Gai10} under the name ``chiral category''.  
 
 This result provides a categorical interpretation of solutions to the Yang-Baxter equation with a spectral parameter, in a way analogous to Drinfeld's categorical interpretation of solutions to the Yang-Baxter equation without spectral parameter.  Drinfeld showed that if $A$ is a Hopf algebra with an $R$-matrix satisfying the Yang-Baxter equation, then $A-\op{mod}$ has the structure of a braided monoidal category.  We show that if $A$ has an $R$-matrix with spectral parameter valued in punctured formal disc, then $A$ has the structure of vertex monoidal category. (Etingof-Kazhdan \cite{EtiKaz98} and Soibelman \cite{Soi97} also considered the categorical structure induced by $R$-matrices with spectral parameter, but the structure seems to be different).

A slightly different version of this statement is shown in section \ref{operator_product_z} of this paper: I show that there is a factorization algebra on $\C$, valued in monoidal categories, where the monoidal category assigned to a disc is a deformation of the category of modules of finite rank over $\g \otimes \op{Hol}(D)$. The operator product in this object is holomorphic, so that this is a holomorphic version of a vertex monoidal category.

\subsection{Acknowledgments}
Conversations with many people have contributed to this paper. I'm particularly grateful to Alexander Chervov, Davide Gaiotto, Owen Gwilliam, Si Li, Kolya Reshitikhin, Nick Rozenblyum, Josh Shadlen and Edward Witten.

\section{Physical $N=1$ super-symmetric gauge theory}
In this section, I will introduce the notation we will use for the $N=1$ supersymmetric gauge theory.   

\subsection{The $N=1$ supertranslation Lie algebra}
Recall that $\op{Spin}(4) = SU(2) \times SU(2)$.  Let $\mc S_+$ and $\mc S_-$ denote the defining complex representations of the two copies of $SU(2)$.  Thus, $\mc S_+$ and $\mc S_-$ are two complex dimensional, and equipped with $\C$-linear actions of $\Spin(4)$.   Let $V = \R^4$ denote the fundamental representation of $\Spin(4)$. Then, there is an isomorphism of complex $\Spin(4)$ representations
$$
\Gamma :  \mc S_+ \otimes \mc S_- \iso V_\C =  V \otimes_\R \C .
$$
\begin{definition}
The $N=1$ supertranslation Lie algebra $T$ is the complex super Lie algebra
$$
T = V_\C \oplus \Pi (\mc S_+ \oplus \mc S_-)
$$
with bracket defined as follows.  If $Q_{\pm} \in \mc S_{\pm}$, then
$$
[Q_+, Q_- ] = \Gamma ( Q_+ \otimes Q_- ).
$$
All other brackets are zero. 
\end{definition}
This super Lie algebra has an evident action of $\Spin(4)$. It is also acted on by $\C^\times$: the $\C^\times$ action on $V_\C$ is trivial, the action on $\mc S_+$ has weight $1$, and on $\mc S_-$ has weight $-1$. In physics terminology, $\C^\times$ is called the (complexified) $R$-symmetry group.

We are interested in field theories with $N=1$ supersymmetry.  Since I have discussed at great length elsewhere \cite{Cos11, CosGwi11} what I mean by a field theory, I won't repeat the definition here.  There is one small point worth mentioning, though: we will consider our field theories to be defined over $\C$, but we will always work complex-linearly. This means that our action functionals, etc. will always be holomorphic functions on the space of fields.  Functional integrals will be taken over a contour; but because we are working perturbatively, the choice of contour is irrelevant.  In the language of factorization algebras \cite{CosGwi11}, this means that our factorization algebras are defined over $\C$.   The foundational texts \cite{Cos11,CosGwi11} are set up to work equally well over $\C$ or $\R$.  

Since we are working over $\C$, it is legitimate to ask for an action of the complex Lie algebra $T$ on a theory. 

\subsection{}
By definition, a field theory on $\R^4$ has $N=1$ supersymmetry if
\begin{enumerate}
\item It is invariant under the Poincar\'e group $\op{Spin}(4) \ltimes \R^4$. 
\item The infinitesimal Poincar\'e symmetry, which is an action of the Lie algebra $\so(4) \ltimes \R^4$, is extended to a complex linear action of $\so(4,\C) \ltimes T$.  
\end{enumerate}
One can also ask that a theory with $N=1$ supersymmetry be equipped with an action of the $R$-symmetry group $\C^\times$, compatible with the action described above of this group on $T$. 

\subsection{}
We are interested in pure $N=1$ supersymmetric gauge theory (``pure'' means with no matter fields).  For the purposes of this paper, we do not need to know about the action of the full supersymmetry algebra, but only about the action of the supercharges in $\mc S_+$. 

We will start by writing down the theory in the most classical way  -- with a space of fields, an action functional and a gauge group -- and then rewrite the theory in the BV formalism.   We will use the notation $\mscr{S}_{\pm}$ to denote the spaces $\cinfty(\R^4) \otimes \mc{S}_{\pm}$ of sections of the spinor bundles on $\R^4$. 
\begin{definition}
The $N=1$ supersymmetric gauge theory, in the first order formalism, has space of fields
$$
\Omega^1 \otimes \g \oplus \Omega^2_+ \otimes \g \oplus \Pi ( \mscr{S}_+ \oplus \mscr S_-) \otimes \g. 
$$
We denote fields in these four summands by $A,B, \Psi_+, \Psi_-$.  The action functional is 
$$
S (A,B, \Psi_+, \Psi_-) = \int \ip{F(A)_+, B}_{\g}  + c  \int \ip{B,B}_{\g} + \int \ip{\Psi_+, \Dirac_A \Psi_- }_{\g}.
$$
Here $c$ is the coupling constant, and we are using the canonical symplectic pairing $\mscr{S}_+ \otimes \mscr{S}_+ \to \cinfty(\R^4)$ to define the action functional on the spinors $\Psi_+, \Psi_-$.  

The gauge Lie algebra is $\Omega^0 \otimes \g$.  The infinitesimal action of an element $X$ in the gauge Lie algebra on the space of fields is given by the formula
$$
(A, B, \Psi_+, \Psi_-)  \mapsto (A  , B , \Psi_+ , \Psi_- ) + \eps \left( \d X + [X,A], [X,B] , [X,\Psi_+] , [X,\Psi_-] \right)
$$
(here $\eps$ is an even parameter of square zero). 
\end{definition}

\subsection{}
Next, we will define an infinitesimal action of $\mc S_+$ on the space of fields.  We need some notation.   Let $\Gamma : \mc S_+ \otimes \mscr{S}_- \to \Omega^1 \otimes \C$ denote the natural map, arising from the identification of the vector representation of $\op{Spin}(4)$ with $\mc S_+ \otimes \mc S_-$.   (In what follows, all forms will be complexified). 

Note that there is a $\op{Spin}(4)$-equivariant identification of the fiber of the bundle $\wedge^2_+ T^\ast \R^4 \otimes \C$ with $\Sym^2 \mc S_+$.  Using the symplectic pairing on $\mc S_+$ we get a $\op{Spin}(4)$-equivariant map $\mc S_+ \otimes \Sym^2 \mc S_+ \to \mc S_+$.  We let
$$
\Upsilon: \mc S_+ \otimes \Omega^2_+ \to \mscr{S}_+
$$
be the map which arises by taking sections.   

We define the action of $\mc S_+$ on the space of fields of the $N=1$ pure gauge theory by saying that $Q \in \mc S_+$ sends $(A,B, \Psi_+, \Psi_-)$ to
$$
(A  , B , \Psi_+ , \Psi_- ) + \eps \left( \Gamma(Q \otimes \Psi_-), 0,  \Upsilon( Q \otimes B), 0   \right). 
$$
where $\eps$ is an odd parameter of square zero.
\begin{remark}
A physicist would write this action of supersymmetry by formulae like
\begin{align*}
\delta_{Q_\alpha} A_\mu^a = \sigma_{\mu \alpha \dot{\alpha}} \psi^{a \dot{\alpha}} \\
\delta_{Q_\alpha} \psi_{\beta}^a = \eps_{\alpha}^{\gamma} B^a_{\gamma \beta}. 
\end{align*}
where $\alpha,\dot{\alpha}$ are indices for $\mc{S}_+$ and $\mc{S}_-$, $a$ is the index for a basis of $\g$, $\mu$ is the index for $\C^4$,  $\sigma_{\mu \alpha \dot{\alpha}}$ are the Dirac matrices, and $\eps_{\alpha}^{\gamma}$ is an alternating tensor.  

The physicists notation can be understood as follows: $A_\mu^a$ is a function from the space of fields to the space $\cinfty(\C^2)$ of smooth functions on $\C^2$, which sends a connection $A$ to the $\mu,a$ component of $A$.  The same holds for quantities like $\psi^{a \dot{\alpha}}$, etc.  Then $\delta_{Q_\alpha} A_{\mu}^a$ indicates the action of the vector field $Q_\alpha$ on the function $A_{\mu}^a$.    \end{remark}
\begin{lemma}
The action of $\mc S_+$ commutes with the action of the gauge Lie algebra, and preserves the action functional on the space of fields.
\end{lemma}
\begin{proof}
It's obvious that $Q$ commutes with the action of $\Omega^0 \otimes \g$.   The $Q$-derivative of the action functional $S$ is expressed by the formula
$$
 \delta_Q S = \int\ip{ \d_A \Gamma(Q \otimes \Psi_-) , B   }_{\g} + \int \ip{\Upsilon (Q \otimes B), \Dirac_A \Psi_- }.
$$
This is zero, with the correct normalization of $\Upsilon$. 
\end{proof}

\subsection{The  $N=1$ theory in the BV formalism}

In the BV formalism, as explained in e.g. \cite{Cos11b}, field theories can be described by dg Lie algebras (or, more generally, $L_\infty$, algebras) with an invariant pairing of degree $-3$.  If $\L$ is such a dg Lie algebra, then the space of fields of the corresponding field theory is $\L[1]$, that is, $\L$ shifted in cohomological degree by $1$. The action functional is the Chern-Simons type functional
$$
S (\alpha) = \tfrac{1}{2}\ip{\alpha, \d \alpha} + \tfrac{1}{6} \ip{\alpha,[\alpha,\alpha] }
$$
 where $\alpha \in \L[1]$, and $\ip{-,-}$ is the invariant pairing on $\L$.   The classical master equation for this action functional is equivalent to the axioms satisfied by a dg Lie algebra with an invariant pairing.

The most familiar example is ordinary Chern-Simons theory on a three manifold $M$. In this case, the dg Lie algebra is $\Omega^\ast(M) \otimes \g$, with differential the de Rham differential, and invariant pairing giving by wedging and integrating (using the invariant pairing on $\g$). 

If one allows $L_\infty$ algebras instead of just dg Lie algebras, all perturbative classical field theories in the BV formalism can be written uniquely (up to equivalence) in this way.  In all cases we will consider in this paper, the field theory is given by a differential graded Lie algebra, as opposed to an $L_\infty$ algebra.  

\subsection{}
When working with theories involving fermions, there are two gradings one has to consider: a $\Z$ grading, called cohomological degree (or ``ghost number'') and a $\Z/2$ grading which I call the super degree.   Both gradings contribute to signs: that is, an element of bidegree $(n,m)$ commutes if $n+m$ is zero modulo $2$, and anticommutes if $n+m$ is $1$ modulo $2$.  The differential on our dg Lie algebra must be of cohomological degree $1$ and super degree $0$; the bracket must preserve both degrees; and the pairing must be of cohomological degree $-3$ and super degree $0$.  

We will sometimes refer to fields of overall even degree as bosonic, and overall odd degree as fermionic. 

The standard BV/BRST procedure of adding ghosts, anti-fields and anti-ghosts (explained in, for instance, \cite{CosGwi11} and \cite{Cos11}), when applied to the pure $N=1$ gauge theory, produces the dg Lie algebra $\mscr{L}^{N=1}$ described in the diagram
$$
\xymatrix{
 & 0 & 1 & 2 & 3 \\ 
\text{Super degree $0$} & \Omega^0 \ar[r]^{\d} & \Omega^1 \ar[r]^{\d_+} & \Omega^2_+ & \\
\text{Super degree $0$} & & \Omega^2_+ \ar[r]^{\d} \ar[ur]^{c \op{Id}} & \Omega^3 \ar[r]^{\d} & \Omega^4 \\
\text{Super degree $1$} & & \mscr{S}_+ \ar[r]^{\Dirac} & \mscr{S}_-' & \\
\text{Super degree $1$}& & \mscr{S}_- \ar[r]^{\Dirac} & \mscr{S}_+' &  }
$$
Here, $\mscr{S}'_{\pm}$ are additional copies of the spaces $\mscr{S}_{\pm}$ of spinors; the spinors $\mscr{S}'_{\pm}$ are the antifields to the original spinors $\mscr{S}_{\pm}$. Of course everything is with coefficients in $\g$.   The cohomological degree $1$ part is the original fields of the theory; in degree $0$, we have the ghosts for the Lie algebra of the gauge group; the antifields are in degree $2$, and the antifields for the ghosts are in degree $3$. 

The invariant pairing of degree $3$ on $\mscr{L}^{N=1}$ is the natural pairing between $\Omega^0$ and $\Omega^4$, $\Omega^1$ and $\Omega^3$, $\mscr{S}_{\pm}$ and $\mscr{S}_{\pm}'$, etc. 

The Lie bracket on $\mscr{L}^{N=1}$ is defined as follows.  The first row is a quotient of $\Omega^\ast \otimes \g$, and so has a natural Lie bracket.  The first row also has a natural action on every other row, which extends the action of $\Omega^0 \otimes \g$ by infinitesimal gauge transformations.   The bracket between two elements neither of which is in the first row is zero. 

This dg Lie algebra completely encodes the pure $N=1$ gauge theory in the BV formalism.   The full BV space of fields is $\L^{N=1}[1]$, and the full BV action functional is the Chern-Simons type functional
$$
S (\alpha) = \tfrac{1}{2}\ip{\alpha, \d \alpha} + \tfrac{1}{6} \ip{\alpha,[\alpha,\alpha] }.
$$
Here $\alpha \in \L^{N=1}[1] $ and $\ip{-,-}$ refers to the invariant pairing in $\L^{N=1}$.

\section{Twisting}
In this section I will explain what I mean by a twist of a field theory with $N=1$ supersymmetry.   In general, to twist a supersymmetric field theory we do the following. 
\begin{enumerate}
\item Choose an $S^1$ in the $R$-symmetry group $G_R$, which is used to change the cohomological degrees of the theory. 
\item Choose a supercharge $Q$ with $[Q,Q] = 0$, which is of weight $1$ under the chosen $S^1 \subset G_R$.  By adding $Q$ to the BRST differential of our theory, we find a new theory. 
\end{enumerate} As I mentioned in the introduction, twisting for me does \emph{not} involve changing the action of $\op{Spin}(4)$ on the theory.  Indeed, the twisted theories we are interested in are not $\op{Spin}(4)$-invariant; this symmetry is broken to one of the $SU(2)$'s in the decomposition $\op{Spin}(4) = SU(2) \times SU(2)$. 

\subsection{}
Let me explain in more detail how twisting works in  the example of pure $N=1$ gauge theory.  Choose a spinor $Q \in \mc S_+$; this induces a complex structure on $\R^4$, whose $(0,1)$-part is $\Gamma( Q \otimes \mc S_-) \subset \R^4 \otimes \C$.

Define the $R$-symmetry action of $\C^\times$ on our space of fields by giving $\mscr{S}_+$ weight  $1$ and $\mscr{S}_-$ weight $-1$.  

We will use the spinor $Q$ and the $R$-symmetry group to define the twisted theory.  The twisted theory will also be described by a dgla, which is obtained by
\begin{enumerate}
\item Changing the cohomological grading of $\mscr{L}^{N=1}$ using the action of the $R$-symmetry group.  
\item Adding $Q$ to the differential of $\mscr{L}^{N=1}$, thus giving a new dgla.
\end{enumerate}

Let $t$ be a parameter of cohomological degree $1$, odd super degree, and weight $-1$ under the $R$-symmetry group.  Note that $t$ is an even parameter.  The twisted theory is defined by the dgla
$$
\L^{Q} = \left( \mscr{L}^{N=1}((t)) , \d_{\mscr{L}} + t Q \right)^{\C^\times_R}.
$$
Adding on the parameter $t$, and taking invariants under the $R$-symmetry group $\C^\times_R$, has the effect of shifting the grading on the dgla $\L$.   

Concretely, the twisted theory is given by the dgla $\L^Q$ described in the diagram
$$
\xymatrix{
  0 & 1 & 2 & 3 \\ 
  \mscr{S}_- \ar[r]^{\Dirac} \ar[dr]^{Q} & \mscr{S}_+' \ar[dr]^{Q}& &   \\
  \Omega^0 \ar[r]^{\d} & \Omega^1 \ar[r]^{\d_+} & \Omega^2_+ & \\
  & \Omega^2_+ \ar[r]^{\d} \ar[dr]^{Q} \ar[ur]^{c \op{Id}} & \Omega^3 \ar[r]^{\d} \ar[dr]^{Q}& \Omega^4 \\
 & & \mscr{S}_+ \ar[r]^{\Dirac} & \mscr{S}_-'  
}
$$
Note that, after twisting, all fields are of super degree $0$, but the cohomological degrees have shifted.  The Lie bracket is as before. 

We will sometimes consider the theory where we only change the grading on $\L^{N=1}$, and do not add $Q$ to the differential.   We will refer to this grading as the twisted cohomological grading.  Equivalently, this is the theory obtained by twisting as above but with $Q = 0$. 

We can rewrite the twisted theory in more classical terms as the BV theory associated to a space of fields, an action functional and a gauge group.  In these terms, the fields are the elements of cohomological degree $1$ of the Lie algebra above: $A \in \Omega^1 \otimes \g$, $B \in \Omega^2_+ \otimes \g$, and $\Psi'_+ \in \mscr{S}'_+ \otimes \g$.   Note that all these fields are of super degree $0$ and cohomological degree $0$. This is forced on us by the twisting procedure.  Before we twisted, the commuting spinor $\Psi'_+$ was an antifield for the usual anticommuting spinor $\Psi_+$. The twisting procedure changes cohomological degree (``ghost number''), so that $\Psi_+'$ becomes a field of ghost number $0$.

The action functional is
\begin{equation*}
S(A,B, \Psi'_+ ) = \int \ip{F(A),B} + c \int \ip{B,B} + \int \ip{\Upsilon'(Q \otimes \Psi'_+), B} \tag{$\dagger$}
\end{equation*}
where 
$$\Upsilon' : \mc{S}_+ \otimes \mscr{S}'_+ \to \Omega^2_+$$ is the unique up to scale $\op{Spin}(4)$-equivariant, $\cinfty(\R^4)$-linear map (adjoint to the map $\Upsilon$ discussed earlier). 

The Lie algebra of the gauge group consists of $X \in \Omega^0 \otimes \g$ and $\Psi_- \in \mscr{S}_- \otimes \g$.  The infinitesimal action of the element $(X,\Psi_-)$ on a field $(A,B,\Psi'_+)$ is given by the formula
$$
(A,B,\Psi'_+) \mapsto (A + \eps \Gamma(Q \otimes \Psi_-) + \eps \d X + \eps [X, A] , B + \eps [X, B] , \Psi'_+ + \eps \Dirac_A \Psi_- + \eps [X,\Psi_+]). 
$$

\section{A holomorphic reformulation of the $N=1$ gauge theory}
\label{section_holomorphic_formulation}
In this section I will show how the $N=1$ gauge theory can be written naturally in terms of holomorphic geometry.  This will be useful when we analyze the twisted theory. 

Let us fix a spinor $Q \in \mc S_+$.  This defines a complex structure on $\R^4$, characterized by declaring that the image of the map $Q : \mscr{S}_- \to \Omega^1$ is $\Omega^{1,0}$.      The choice of $Q$ reduces the $\op{Spin}(4)$ symmetry to $SU(2)$.  

We have $SU(2)$-invariant isomorphisms
\begin{align*}
\mscr S_-  & = \Omega^{1,0} \\ 
\Omega^2_+ &= \Omega^{2,0} \oplus  \Omega^0  \cdot \omega \oplus \Omega^{0,2} \\
\mscr{S}_+ &=  \Omega^0 \cdot \omega \oplus \Omega^{0,2}.
\end{align*}
Here, $\Omega^0 \cdot \omega$ is the space of $(1,1)$-forms which are multiples of the K\"ahler form $\omega$.  Of course, the $SU(2)$ action on $\mscr{S}_+$ is trivial;  however, it is convenient to represent it this way. 

Recall that the antifields to $\mscr{S}_\pm$ is another copy of $\mscr{S}_\pm$, which we denote $\mscr{S}'_\pm$.   We can identify
\begin{align*}
\mscr S'_-  & = \Omega^{1,2} \\ 
\mscr{S}'_+ &= \Omega^{2,0} \oplus \Omega^0 \cdot \omega
\end{align*}
Then, the natural pairing between $\mscr{S}'_{\pm}$ and $\mscr{S}_{\pm}$ is given by wedging and integrating.  

With this notation, the Dirac operator 
$$\Dirac_A :\mscr{S}_- \to \mscr{S}'_+$$ is the composition
$$
\Dirac_A : \Omega^{1,0} \xto{\d_A} \Omega^2 \xto{\pi} \Omega^{2,0} \oplus \omega \cdot\Omega^{0,0}
$$
where $\d_A = \d_{dR} + [A,-]$ is the covariant derivative, and $\pi$ is the projection onto $\Omega^{2,0} \oplus \omega \cdot \Omega^0$. 

Similarly, the Dirac operator $\Dirac_A : \mscr{S}_+ \to \mscr{S}'_-$ is the composition of $\d_A$ with projection onto $\Omega^{1,2}$.  

The fields of the $N=1$ gauge theory,  in this holomorphic notation, consist of 
\begin{align*}
A & \in \Omega^1 \otimes \g \\
B & \in \left(  \Omega^{2,0} \oplus  \Omega^0  \cdot \omega \oplus \Omega^{0,2} \right) \otimes \g \\
\Psi_+ & \in \left( \Omega^0 \cdot \omega \oplus \Omega^{0,2} \right)\otimes \g\\
\Psi_- & \in \Omega^{1,0} \otimes \g. 
\end{align*}
The Lie algebra of the gauge group is $\Omega^0 \otimes \g$, acting as before.  The action functional can be written as 
$$
\int \ip{B, F(A)}  + \int\ip{\Psi_+ , \d_A \Psi_- } + c \int \ip{B,B},
$$
where in each term $\ip{-,-}$ indicates a combination of wedging of forms and the invariant pairing on the Lie algebra $\g$.  

The supercharge $Q \in \mc{S}_+$ (associated to the chosen complex structure) acts on the space of fields in a very simple way: it maps $\mscr{S}_- = \Omega^{1,0}$ into $\Omega^1$ by the natural inclusion, and $\Omega^2_+ = \Omega^{2,0} \oplus \Omega^0 \cdot \omega \oplus \Omega^{0,2}$ to $\mscr{S}_+ = \Omega^0 \cdot \omega \oplus \Omega^{0,2}$ by the natural projection.

\section{Holomorphic BF theory and the twisted $N=1$ gauge theory}
The twisted theory $\L^{Q}$ can be written in a very simple way in terms of holomorphic geometry.  In what follows we will always work with the complex structure on $\R^4$ corresponding to the chosen spinor $Q \in \mc S_+$. 
\begin{definition}
The holomorphic BF theory is the field theory on $\C^2$ whose fields are a $(0,1)$-form $A \in \Omega^{0,1}(\C^2,\g)$ and a $(2,0)$-form $B \in \Omega^{2,0}(\C^2,\g)$, with action
$$
S(A,B) = \int \ip{ B , \left( \dbar A + \tfrac{1}{2} [A ,  A] \right)  }_{\g}.
$$
Here, $\ip{-,-}_{\g}$ refers to a chosen invariant non-degenerate pairing on $\g$.    The Lie algebra of the group of gauge transformations is $\Omega^{0,0}(\C^2,\g)$. An element $X$ in this Lie algebra acts on the space of fields $(A,B)$ in the natural way: 
\begin{align*}
A & \mapsto A + \eps ( \dbar X + [X, A ] )\\
B & \mapsto B + \eps [X, B].
\end{align*}
\end{definition}
In the BV formalism, the holomorphic BF theory is described by the dgla 
$$
\mscr{L}^{BF} = \Omega^{0,\ast}(\C^2,\g) \oplus \Omega^{2,\ast}(\C^2,\g)[-1].
$$
In the language of \cite{Cos11b}, this theory is the cotangent theory to the derived moduli space of $G$-bundles on $\C^2$.  As, the dgla $\Omega^{0,\ast}(\C^2,\g)$ describes deformations of holomorphic $G$-bundles on $\C^2$.  Adding on the fields in  $\Omega^{2,\ast}(\C^2,\g)$ gives the Lie algebra describing a certain shifted cotangent bundle to this moduli space. 
\begin{theorem}
The twisted $N=1$ gauge theory on $\C^2$ is equivalent to the holomorphic BF theory.
\label{theorem_bf}  
\end{theorem}
\begin{remark}
A similar relationship between holomorphic BF theory and $N=1$ supersymmetric gauge theory was suggested in \cite{BauTan04}.
\end{remark}

\begin{proof}
A somewhat abstract proof of this (based on the twistor formulation of supersymmetric gauge theory developed by Mason et al. \cite{BoeMasSki07}) was presented in \cite{Cos11b}.  Here we'll give a more direct proof. 

We have seen how to formulate the $N=1$ gauge theory in terms of complex geometry.  We have also seen that the twisted theory can be described as the theory where fields are a connection $A \in \Omega^1\otimes \g$, a field $B \in \Omega^2_+ \otimes \g$, and a bosonic (commuting) spinor $\Psi'_+ \in \mscr{S}_+' \otimes \g$, where the prime in $\mscr{S}_+'$ indicates that these fields correspond, in  the untwisted theory, to the antifields to the spinors $\mscr{S}_+$.   The Lie algebra of the gauge group is $\Omega^0 \otimes \g \oplus \mscr{S}_- \otimes \g$. 

Combining these two descriptions, we see that the twisted theory can be written in terms of complex geometry.  The fields are 
\begin{align*}
A \in \Omega^1 \otimes \g &\\
B \in \Omega^2_+ \otimes \g &= \left(\Omega^{2,0} \oplus \omega \cdot \Omega^0 \oplus \Omega^{0,2}\right)\otimes \g \\
\Psi'_+ \in \mscr{S}'_+ &= \Omega^{2,0} \oplus \omega \cdot \Omega^{0,0}
\end{align*}
Let us decompose our fields $A,B,\Psi'_+$ into components $A^{1,0}$, $A^{0,1}$, $B^{2,0}$, $B^{1,1}$, $B^{0,2}$, $\Psi_+^{\prime 2,0}$, $\Psi_+^{\prime 1,1}$.  Each field is a form of the specified type, except that we only use $(1,1)$-forms which are a multiple of the K\"ahler form $\omega$.  

The Lie algebra of gauge symmetries is $\Omega^0 \otimes \g \oplus \Omega^{1,0} \otimes \g$ (where we have used the identification $\mscr{S}_- = \Omega^{1,0}$).   An element $(X,\Psi_-)$ in this Lie algebra acts on a field $(A,B,\Psi'_+)$ by 
$$
(A,B,\Psi'_+) \mapsto (A + \eps \d X + \eps [X,A] + \eps \Psi_- ,  B + \eps[X,B] , \Psi'_+ + \eps [X,\Psi'_+] + \eps \pi \circ \d_A \Psi_- ).
$$

The action functional of the twisted theory is
$$
\int \ip{F(A), B} + c \int \ip{B,B} + \int \ip{\Psi_+', B}
$$
In each term, the pairing $\ip{-,-}$ represents a combination of the wedge product of forms and the inner product on the Lie algebra $\g$. 

Let us perform the gauge-invariant change of coordinates
$$
\Psi_+'  \mapsto \Psi_+' - c B^{1,1} - 2 c B^{2,0}    -  \pi F(A) \\
$$
where $\pi$ is the projection from $\Omega^2$ to $\mscr{S}'_+ = \Omega^{2,0} \oplus \Omega^0 \cdot \omega$. 

Then, the action functional becomes simply
$$
\int \ip{ F(A)^{0,2} , B^{2,0}} + \int \ip{\Psi_+^{\prime 2,0}, B^{0,2}} +\int  \ip{\Psi_+^{\prime 1,1}, B^{1,1}}.
$$
Thus the fields $\Psi'_+, B^{0,2}, B^{1,1}$ do not propagate or interact with the other fields. We can thus integrate them out.  The remaining fields are $B^{2,0}$, $A^{1,0}$, $A^{0,1}$, with action functional $
\int \ip{ F^{0,2} (A) , B^{2,0}}. $ 

As well as ordinary gauge symmetry, we have a gauge symmetry by $\Psi_- \in \Omega^{1,0}$.  We can use $\Psi_-$ to set $A^{1,0} = 0$, so that the fields and the action functional are precisely those of holomorphic BF theory. 
\end{proof}
\begin{remark}
I have given the proof using physics terminology: it is not difficult to verify that this prove gives a quasi-isomorphism of dg Lie algebras 
$$\L^Q \simeq \Omega^{0,\ast}(\g) \oplus \Omega^{2,\ast}(\g)[-1]$$
between the dg Lie algebra describing the twisted $N=1$ theory and that describing holomorphic BF theory. This quasi-isomorphism is compatible with invariant pairings. 
\end{remark}

\subsection{}
We can define the twisted $N=1$ theory on any complex surface $X$, perturbing around any holomorphic $G$-bundle $P$ on $X$.  The Lie algebra describing this theory is $\Omega^{0,\ast}(X,\g_P) \oplus \Omega^{2,\ast}(X,\g_P)[-1]$. 
\begin{theorem}
The twisted $N=1$ theory admits a unique quantization, compatible with certain natural symmetries, on any Calabi-Yau surface $X$; where we perturb around any holomorphic $G$-bundle on $X$.   (For us, Calabi-Yau just means that $X$ is equipped with a holomorphic volume form).
\label{theorem_twisted_existence}
\end{theorem}
The proof is given in Appendix \ref{appendix_existence}, where I also explain precisely what I mean by ``natural symmetries''. 

\begin{remark}
It follows from this result that the well-known $R$-symmetry anomaly present in the $N=1$ gauge theory must be $Q$-exact.   If not, then the twisted theory could not be constructed at the quantum level as a $\Z$-graded theory, which would contradict this theorem.
\end{remark}

\section{A deformation}
The Yangian will arise from a twist of a deformation of the $N=1$ gauge theory.  In this section we will describe this deformation (which we call the Yangian deformation) and verify that it is invariant under our chosen supercharge $Q \in \mc S_+$.  

We will work with the holomorphic formulation of the $N=1$ gauge theory (described in section \ref{section_holomorphic_formulation}), using the complex structure on $\R^4$ induced by our chosen supercharge $Q \in \mc S_+$.   As before, we identify $\mscr{S}_- = \Omega^{1,0}$, $\mscr{S}_+ = \Omega^{0,2}\oplus \Omega^0 \cdot \omega$, and $\Omega^2_+ = \Omega^{2,0} \oplus \Omega^{0,2} \oplus \Omega^0 \cdot \omega$.   Let us expand $\Psi_- \in  \Omega^{1,0}$ as 
$$\Psi_- = \Psi_-^z \d z + \Psi_-^w \d w.$$  Similarly, we let $\Psi_+^{0,2}$ and $\Psi_+^{1,1}$ be the components of $\Psi_+$ in $\Omega^{0,2}$ and $\Omega^0 \cdot \omega$ respectively.  

Let 
$$
\op{CS}(A) = \tfrac{1}{2} \ip{A, \d A} + \tfrac{1}{6} \ip{A, F(A)}.
$$
be the Chern-Simons three-form associated to the connection $A$. 

Let us deform the $N=1$ gauge theory to a theory with the same fields, but where we add to the action functional a term 
$$
S' (A,B, \Psi_+, \Psi_- ) = \int \d z\op{CS}(A) +2  c \int \ip{\Psi_+^{0,2}, \Psi_-^w}_\g \d z \d w.
$$
\begin{lemma}
The action of the supercharge $Q \in \mc S_+$ on the fields of the $N=1$ gauge theory can be deformed so that it preserves the deformed action functional $ S + c' S'$.  \end{lemma} 
\begin{proof}
The full action functional of our theory is
\begin{multline*}
S^{deformed} = \int \ip{B, F(A)}  + \int\ip{\Psi_+ , \d_A \Psi_- } + c \int \ip{B,B} \\ 
+ c' \int \d z \op{CS}(A)  +2  c c' \int \ip{\Psi_+^{0,2}, \Psi_-^w}_\g \d z \d w.
\end{multline*}
The original action of $Q$ is defined by
$$
(A,B,\Psi_+, \Psi_- ) \mapsto (A + \eps \iota(\Psi_-) ,B,\Psi_+ + \eps \pi B , \Psi_- )
$$
where $\iota: \Omega^{1,0} \into \Omega^1$ is the natural inclusion, and 
$$\pi : \Omega^{2,0} \oplus \Omega^{0,2} \oplus \Omega^0 \cdot \omega \to \Omega^{0,2} \oplus \Omega^0 \cdot \omega$$
is the natural projection.  

Let us deform this to the action defined by the formula
$$
(A,B,\Psi_+, \Psi_- ) \mapsto (A + \eps \iota(\Psi_-)  ,B - c' \eps \Psi_-^w \d z \d w ,\Psi_+ + \eps \pi B , \Psi_- ) 
$$

Note that the $Q$-variation of $c' \int \d z \op{CS}(A)$ is 
$$
\delta_{Q} c' \int \d z \op{CS}(A) = c' \int \d z \d w \Psi_-^w F^{0,2}(A).
$$
This cancels with the $c'$-dependent term appearing in the $Q$-variation of $c \int \ip{B,F(A)}$. 

Similarly, the $Q$-variation of $2 c c' \int \ip{\Psi_+^{0,2}, \Psi_-^w}_\g \d z \d w $ is $2 c c' \int \ip{B^{0,2}, \Psi_-^w} \d z \d w$.  This cancels with the $Q$-variation of $\int \ip{B,B}$.  
\end{proof}

The main result of this paper concerns Wilson operators of the twist of this deformed theory.  These Wilson operators in the twisted theory arise from a supersymmetric Wilson operator in the deformed but not twisted theory. This is in contrast to the undeformed $N=1$ gauge theory, where there are no supersymmetric Wilson operators.

Let $A_w, A_{\br{w}}$ denote the coefficients in $A$ of $\d w$ and $\d \br{w}$. Let $B_{\d z \d w}$ denote the coefficient in $B$ of $\d z \d w$. 
\begin{lemma}
The connection in the $w$-plane by the restriction of 
$$A_w \d w + A_{\br{w}} \d \br{w} + \tfrac{1}{c'} B_{\d z \d w} \d w$$ is invariant under the action of our chosen supercharge $Q \in \mc S_+$.    
\end{lemma}
\begin{proof}
This is immediate from the definition of the deformed action of supersymmetry. 
\end{proof}

\subsection{}
This lemma shows that our deformation of the $N=1$ gauge theory has enough supersymmetry to be able to twist it.  The twisted theory will be a deformation of holomorphic BF theory.

\begin{definition}
Define the deformed holomorphic BF theory to be the theory where the fields are $A \in \Omega^{0,1}(\C^2) \otimes \g$ and $B \in \Omega^{2,0}(\C^2) \otimes \g$; the action functional is 
$$
S(A,B) = \int \ip{B, F^{0,2}(A) } + \lambda \tfrac{1}{2} \int \d z A \partial A
$$
where $\lambda$ is a coupling constant, and $\partial A \in \Omega^{1,1}$ is the $(1,1)$ part of $\d A$.  The Lie algebra of the gauge group is $\Omega^0 \otimes \g$, and the action of this Lie algebra on the fields is by
\begin{align*}
A & \mapsto A + \eps \dbar X + \eps [X,A]\\
B & \mapsto B + \lambda \eps \d z \wedge \partial X + \eps [X,B].
\end{align*}
\end{definition}
\begin{remark}
This theory is equivalent to the one where the fields are $A' \in (\Omega^1(\C^2) / \d z) \otimes \g$, and the action functioanl is $\lambda \int \d z CS(A)$. If we write $B = B_0 \d z \d w$, the equivalence sends the fields $(A,B)$ in the above formulation to
$$
A' = \lambda^{-1} B_0 \d w + A.  
$$
\end{remark}
The dg Lie algebra describing this theory is
$$
\Omega^{0,\ast} (\C^2,\g) \xto{\d z \partial} \Omega^{1,\ast}(\C^2,\g)[-1].
$$
This dg Lie algebra describes the moduli space of holomorphic $G$-bundles on $\C^2$ with a holomorphic connection in the $w$-direction.

\begin{proposition}
The $Q$-twist of the Yangian deformation of the $N=1$ gauge theory is equivalent to this deformed holomorphic BF theory, where the coupling constants $c'$ and $\lambda$ correspond. 
\end{proposition}

\begin{proof}
The proof parallels closely the proof of theorem \ref{theorem_bf}.   We can write the twisted theory as the BV theory associated to a collection of fields, an action functional, and a Lie algebra of gauge symmetries.  The fields are, as in theorem \ref{theorem_bf}, 
\begin{align*}
A & \in \Omega^1  \\
B & \in \Omega^2_+ = \Omega^{2,0} \oplus \Omega^0 \cdot \omega \oplus \Omega^{0,2} \\
\Psi_+' & \in \mscr{S}_+' = \Omega^{2,0} \oplus \Omega^0 \cdot \omega.
\end{align*}
The action functional is 
$$
S(A,B,\Psi'_+) = \int \ip{F(A), B} + c \int \ip{B,B} + c' \int \d z \op{CS}(A) + \int \ip{\Psi', B}. $$
The Lie algebra of the gauge group is $\Omega^0 \otimes \g \oplus \mscr{S}_- \otimes \g$.   Infinitesimal gauge transformations $X \in \Omega^0 \otimes \g$ act in the usual way: $A$ is a connection and everything else is a tensor.  An infinitesimal gauge transformation $\Psi_- \in \mscr{S}_- \otimes \g$ acts by
$$
(A,B,\Psi'_+) \to (A + \eps \iota \Psi_-, B - \eps c'\d z \d w \Psi_-^w, \Psi'_+ + \eps 2 c c' \Psi_-^w \d z \d w + \eps \pi \circ \d_A \Psi_- )  
$$
where, as before, $\iota : \Omega^{1,0} \to \Omega^1$ is the natural inclusion, and $\pi : \Omega^2 \to \Omega^{2,0} \oplus \Omega^0 \cdot \omega$ is the natural projection.   In this expression, we are expanding $\Psi_- \in \Omega^{1,0}$ as $\Psi_- = \Psi_-^z \d z + \Psi_-^w \d w$. 

As before, we let $\Psi_+^{\prime a,b}$, $A^{a,b}$ and $B^{a,b}$ indicate the components of the fields $A,B, \Psi'_+$ which are in $\Omega^{a,b}$.    Let us perform the gauge invariant change of coordinates on the space of fields  
$$
\Psi_+'  \mapsto \Psi_+' - c B^{1,1}   -  \pi F(A) \\
$$
where $\pi$ is the projection from $\Omega^2$ onto $\mscr{S}_+' = \Omega^{2,0} \oplus \Omega^0 \cdot \omega$. 

After this change of coordinates, the action functional becomes
$$
\int \ip{ F(A)^{0,2} , B^{2,0}} + 2c \int \ip{B^{0,2}, B^{2,0} }  + c' \int \d z \op{CS}(A) + \int  \ip{\Psi_+^{\prime 2,0}, B^{0,2}}  + \int  \ip{\Psi_+^{\prime 1,1}, B^{1,1}}.
$$
We can now integrate out the fields $B^{1,1}$ and $\Psi_+^{\prime 1,1}$, as they do not propagate or interact with other fields. 

Next, let us change coordinates again, and let 
$$
\Phi^{2,0} = \Psi_+^{\prime 2,0} + 2 c  B^{2,0}.
$$
We will use the field $\Phi^{2,0}$  instead of $\Psi^{\prime 2,0}$.  Note that $\Phi^{2,0}$ transforms by $\Phi^{2,0} \mapsto \Phi^{2,0} + \eps \left( \d_A \Psi_-\right)^{2,0}$ under the action of $\Psi_- \in \Omega^{1,0} \otimes \g$.  

In this new coordinate system, the action functional is simply
$$
\int \ip{ F(A)^{0,2} , B^{2,0}} +  \int \ip{B^{0,2}, \Phi^{2,0} }   + c' \int \d z \op{CS}(A) .
$$
Thus, we can integrate out the fields $\Phi^{2,0}$ and $B^{0,2}$.  

The next step is to fix the gauge symmetry by $\Psi_- \in \Omega^{1,0} \otimes \g$.   Since this acts on $A$ by sending $A \mapsto A + \eps \Psi_-$, we can use this gauge symmetry to set $A^{1,0} = 0$.  In other words, the locus where $A^{1,0} = 0$ is a slice to the action of $\Omega^{1,0} \otimes \g$.   

Once we do this, the space of fields consists of $A^{0,1}$ and $B^{2,0}$.  The action functional is 
$$
\int \ip{ F(A)^{0,2} , B^{2,0}}  + c' \tfrac{1}{2} \int \d z \ip{  A  , \partial A}.
$$
The remaining gauge symmetry is $X \in \Omega^0 \otimes \g$.  However, the action of this gauge symmetry is not the most obvious one.   As, the action of $X$ on $A^{0,1}$ involves the term $\partial X \in \Omega^{1,0}$.   This needs to be removed by an application of the other gauge symmetry $\Psi_- \in \Omega^{1,0} \otimes \g$.  Thus, we find that $X$ acts on the space of fields by
$$
(A^{0,1}, B^{2,0} ) \mapsto  (A^{0,1} , B^{2,0})  + \eps ([X,A^{0,1}],[ X,B^{2,0}] ) 
+ \eps \left( \dbar X, c' \dpa{w} X \d z \d w \right).
$$
This shows that our theory is equivalent to the desired deformation of holomorphic BF theory. 
\end{proof}

\subsection{}
\label{log_theory}
The deformed holomorphic BF theory, as a classical field theory, can be defined on any complex surface $X$ equipped with a principal $G$-bundle $P$, and a closed holomorphic $1$-form $\alpha$, where $P$ is equipped with a flat connection along the holomorphic foliation defined by the kernel of $\alpha$. 
\begin{theorem}
Let $X$ be a Calabi-Yau surface equipped with a closed holomorphic one-form $\alpha$.  Let $V$ be the vector field which contracts with the holomorphic volume form to give $\alpha$. Let $P$ be a flat $G$-bundle on $X$, equipped with a flat holomorphic connection in the direction spanned by $V$.   Then there exists a unique quantization of the deformed, twisted $N=1$ gauge theory, compatible with certain natural symmetries as before. 
\end{theorem}
The proof is given in Appendix \ref{appendix_existence}.

We will use a variant of this deformed holomorphic BF theory, which can be defined on any complex surface $X$ with the following data:
\begin{enumerate}
\item A reduced divisor $D \subset X$.
\item A meromorphic volume form $\omega \in K_X(2 D)$ which trivializes the bundle $K_X(2 D)$.
\item A holomorphic vector field $V$ which preserves both $\omega$ and $D$. Thus, $V$ is parallel to $D$.
\end{enumerate}
Let $P$ be a holomorphic principal bundle on $X$, trivialized on $D$.  Let $\nabla_V$ be a holomorphic connection on $P$ in the direction given by $P$.  In the case $G = GL_n$, and $A$ is the rank $n$ holomorphic vector bundle on $X$, then $\nabla_V$ is encoded by a map
$$
\nabla_V : \cinfty(X,A) \to \cinfty(X,A)
$$ 
on smooth sections of $A$, satisfying
\begin{align*}
\nabla_V (f s) &= f \nabla_V s + (V f) s \\
[\nabla_V, \dbar] = 0.
\end{align*}
Along the divisor $D$ on which $P$ is trivialized, we assume that $\nabla_V$ is the connection associated to the trivialization. 

Given $P$, we can define a perturbative field theory as follows.  Fields are pairs
\begin{align*}
\alpha & \in \Omega^{0,1}(X, \g_P(-D)) \\
\beta & \in \Omega^{0,0}(X, \g_P(-D)) .
\end{align*}
Here $\g_P(-D)$ refers to the holomorphic bundle obtained by tensoring the adjoint bundle $\g_P$ with $\Oo(-D)$. 

The action functional is
$$
S(\alpha,\beta) = \int_X \omega \ip{\alpha, \dbar \beta}_{\g} + \omega \tfrac{1}{2}\ip{\alpha, \nabla_V \alpha }_{\g} + \tfrac{1}{6} \ip{[\alpha,\alpha], \beta}_{\g}.
$$
Although the holomorphic volume form $\omega$ has quadratic poles along $D$, the smooth volume forms appearing in the integrand have no singularities, because $\alpha$ and $\beta$ have first-order zeroes along $D$.

The Lie algebra of the gauge group is $\Omega^{0,0}(X, \g_P(-D))$.  The infinitesimal action of an element $c \in \Omega^{0,0}(X,\g_P(-D))$ on $(\alpha,\beta)$ is by
\begin{align*}
\alpha & \mapsto \alpha + \dbar c + [c, \alpha] \\
\beta & \mapsto \beta + \nabla_V c + [c,\beta]. 
\end{align*}

This theory is associated as above to the dg Lie algebra 
$$
\mscr{L} = \Omega^{0,\ast}(X,\g_P[\eps] (-D)),
$$
where $\eps$ is an odd parameter.  The differential on $\mscr{L}$ is $\dbar + \eps \nabla_V$, and the pairing is
$$
\ip{\alpha, \eps \beta} = \int_X \omega \ip{\alpha,\beta}_{\g}
$$
where $\alpha,\beta \in \Omega^{0,\ast}(X, \g_P(-D))$. 

Note that this dg Lie algebra controls the space of deformations of $P$ as a holomorphic bundle trivialized on $D$ and with a partial connection $\nabla_V$.  Thus, the moduli of solutions of the equations of motion is the space of such bundles.  This derived moduli space has a symplectic form of degree $-1$. 

\begin{theorem}
This theory admits a unique perturbative quantization compatible with certain symmetries. 
\end{theorem}
Again, the proof is in the appendix.  Note that away from the divisor $D$, the quantum theory we construct is the same one we considered earlier. 

Of particular interest is the case when $X = \mbb{P}^1_z \times E_w$, the volume form is $\omega = \d z \d w $, and the vector field is $V = \dpa{w}$.  In this case, we have the following.
\begin{lemma}
Let $P$ be the trivial bundle on $\mbb{P}^1 \times E$ with the trivial partial connection $\nabla_V$, and let $\mscr{L}_P$ be the dg Lie algebra described above which controls deformations of $P$.  Then,
$$
H^\ast(\mscr{L}_P ) = 0.
$$
\label{lemma_isolated_solution}
\end{lemma}
\begin{proof}
We have
$$
H^\ast(\mscr{L}_P) = H^\ast_{\dbar}(\mbb{P}^1, \Oo(-\infty)) \otimes H^\ast_{dR} ( E ) \otimes \g = 0. 
$$
\end{proof}
In physics terminology, this says that there are no massless modes.  This means that we can compute expectation values in perturbation theory. Normally, the perturbative expectation value of an observable will be a volume form on the moduli of solutions to the equations of motion, which we need to integrate to get. In this case, when we have an isolated solution, the expectation value is simply a number. 

\section{Factorization algebras}
The main result of this paper is that the Yangian Hopf algebra is encoded in the factorization algebra of observables of the Yangian deformation of the twisted $N=1$ gauge theory.  To state this theorem precisely, I need to say briefly what a factorization algebra is.
\begin{definition}
Let $X$ be a manifold. A prefactorization algebra $\mc{F}$ on $X$ is  an assignment to every open subset $U \subset X$ a cochain complex $\mc{F}(U)$; and to every inclusion $U_1 \amalg \dots \amalg U_n \into V$ of disjoint open set $U_i$ into $V$, a cochain map
$$
\mc{F}(U_1) \otimes \dots \otimes  \mc{F}(U_n) \to \mc{F}(V),
$$ 
which is independent of the ordering chosen on the set $U_i$, and which satisfies the following associativity condition.  Let $U_i$ and $V_j$ be finite collections of open subsets of $X$, where $U_i \cap U_{i'} = \emptyset$ and $V_{j} \cap V_{j'} = \emptyset$. Let $W$ be another open subset of $X$. Suppose that
$$
\amalg U_i \subset \amalg V_j \subset W.
$$
Then the following diagram commutes:
$$
\xymatrix{
\otimes_{i} \mc{F}(U_i) \ar[r] \ar[dr] & \otimes_{j} \mc{F}(V_j) \ar[d] \\
& \mc{F}(W) ,
}
$$
where the arrows are the structure maps of the prefactorization algebra. 

A factorization algebra is a prefactorization algebra satisfying a certain ``descent'' condition, saying that $\mc{F}(V)$ for any open subset $V$ is determined from the value of $\mc{F}$ on a sufficiently fine open cover of $V$. (The descent condition will not be used in this paper).
\end{definition}
In \cite{CosGwi11}, it is shown that a quantum field theory on a manifold $X$ (in the sense of \cite{Cos11}) gives rise to a factorization algebra $\Obs^q$ of ``quantum observables'' on $X$.  This is a factorization algebra over $\C[[\hbar]]$, and quantizes a much simpler factorization algebra of ``classical observables''.  

If the field theory we're considering is expressed, in the BV formalism, by a sheaf of dg Lie algebras $\L$, then the factorization algebra of classical observables $\Obs^{cl}$ assigns to the open subset $U \subset X$ the Lie algebra cochain complex $C^\ast(\L(U))$.  In the case of the Yangian deformation of the twisted $N=1$ gauge theory on $\C^2$, the dg Lie algebra associated to an open subset $U \subset \C^2$ is
$$
\L(U) = \Omega^{0,\ast}(U, \g) \xto{\d z \partial }\Omega^{2,\ast}(U,\g)[-1]. 
$$
\section{Hopf algebras from factorization algebras}
\label{e_2_hopf}
A combination of some general results of Tamarkin \cite{Tam98, Tam03b, Tam07} and Lurie \cite{Lur12} asserts that there is a very close relationship between Hopf algebras and factorization algebras on $\R^2$.  In this section I will try to explain these results, which rely on some quite sophisticated abstract homotopical algebra.  By applying this construction to the factorization algebra associated to the Yangian deformation of the twisted $N=1$ gauge theory, we will find the Yangian Hopf algebra. 
  
There is a vast literature on this kind of homotopical algebra, so I will not give all details. 
\begin{definition}
A factorization algebra $\mc{F}$ on $\R^2$  is called \emph{locally constant} if, for every inclusion of discs $D \into D'$, the map $\mc{F}(D) \to \mc{F}(D')$ is a quasi-isomorphism. 
\end{definition}
We will use Tamarkin and Lurie's results to show that locally constant factorization algebras on $\R^2$ can be turned into Hopf algebras.   This construction relies on the concept of an $E_2$-algebra, which I will briefly recall. 

\begin{definition}
Let $E_2(n)$ denote the topological space parametrizing $n$ disjoint round discs in the unit disc in $\R^2$. 
\end{definition}
Thus, a point in $E_2(n)$ consists of $n$ points $p_1,\dots,p_n$ in the unit disc in $\R^2$, and $n$ radii $r_1$, \dots,$r_n$, such that the discs $D_{r_i}(p_i)$ of radius $r_i$ around $p_i$ are all disjoint and contained in the unit disc $D_{1}(0)$.  

The spaces $E_2(n)$ form what is called an ``operad'': there are gluing maps
$$
\circ_i : E_2(m) \times  E_2(n) \to E_2( n+m -1).
$$
We can think of a point in $E_2(m)$ as a cobordism from $m$ circles to $1$ circle; the composition $A \circ_i B$ is defined by gluing the output circle of the cobordism $A$ to the $i$'th input circle of the cobordism $B$.  In order to do this, we need to scale the cobordism $A$ by the radius of the $i$'th input circle in $B$. 
\begin{definition}
An $E_2$-algebra is a cochain complex $V$ equipped with cochain maps
$$
C_\ast(E_2(n)) \otimes V^{\otimes n} \to V
$$
in a way compatible with gluing. Thus, the maps 
$$\circ_i : C_\ast(E_2(m)) \otimes C_\ast(E_2(n)) \to C_\ast(E_2(n+m-1))$$
correspond to the composition of morphisms between tensor powers of $V$. 
\end{definition}
\begin{remark}
The cochain complex assigned to a circle in a two-dimensional topological field theory has the structure of an $E_2$ algebra; see \cite{Cos07a} for example.  
\end{remark}
By considering the space of $0$-chains $C_0(E_2(n))$, we see that an $E_2$ algebra $A$ has a product associated to every configuration of discs in the plane.  The maps associated to one-chains tell us that, as we vary this configuration, the product changes by chain homotopy.  Higher chains in $C_\ast(E_2(n))$ give higher homotopies. 
 
One can obviously consider $E_k$ algebras for any $k \ge 1$.  It is well-known that $E_1$ algebras provide a model for homotopy associative algebras (i.e.\ there are equivalences between the categories of $E_1$ algebras, $A_\infty$ algebras, and differential graded algebras). 

By considering configurations of discs with center on the horizontal line, we see that every $E_2$ algebra is, in particular, an $E_1$ algebra.  
\begin{theorem}[Lurie]
If $\mc{F}$ is a locally-constant factorization algebra on $\R^n$, then the cochain complex $\mc{F}(D)$ associated to a disc in $\R^n$ has a natural structure of $E_n$ algebra (up to coherent homotopy). 
\end{theorem}
For the proof, see Theorem 5.2.4.9 of \cite{Lur12}.   It is easy to see how the $0$-chains of $C_\ast(E(n))$ act: this is just part of the structure of a factorization algebra.  The hard part of Lurie's theorem is showing that the product associated to a configuration of discs varies in a homotopically trivial fashion over the moduli space of discs.  

\subsection{Hopf algebras}
Now we can turn to Tamarkin's theorem \cite{Tam98} relating $E_2$ algebras an Hopf algebras.  Suppose we have an $E_2$ algebra $A$. As we have seen, by considering configurations of discs centered on the horizontal axis, we can give $A$ the structure of $E_1$ algebra; in other words, of homotopy associative algebra.  

We can ask that $A$ be augmented as an $E_1$ algebra.  This means that $A$ is equipped with a homomorphism $A \to \C$ of $E_1$ algebras to the base ring.  

A well-known construction in homotopical algebra called \emph{Koszul duality} allows one to turn augmented $E_1$ algebras into coalgebras.  If $A$ is an augmented $E_1$ algebra, then one defines the bar dual $A^!$ by
$$
A^! = \C \otimes^{\mbb{L}}_A \C.
$$
If we assume that $A$ is an actual dg associative algebra, then the bar dual $A^!$ can be modeled as 
$$
A^! = \oplus_{n \ge 0} A^{\otimes n}[n]
$$
with differential defined by 
\begin{multline*}
\d (a_1 \otimes  \dots \otimes a_n) = \sum_{i = 1}^n  \pm a_1 \otimes \dots (\d a_i) \otimes \dots \otimes a_n \\
+ \sum_{i = 1}^{n-1} a_1 \otimes \dots \otimes (a_i a_{i+1}) \otimes \dots \otimes a_n.   
\end{multline*}

\begin{theorem}[Tamarkin]
If $A$ is an $E_2$ algebra which is augmented as an $E_1$ algebra, then $A^!$ has a natural structure of Hopf algebra.  
\end{theorem}
\begin{remark}
\begin{enumerate}
\item
In a given cochain model for $A^!$, the axioms of a Hopf algebra will only hold up to coherent homotopy. However, at the level of cohomology, $H^\ast(A^!)$ will have an honest Hopf algebra structure.
\item Under certain hypotheses, which will be satisfied in our main example, we can recover the $E_2$ algebra $A$ from the Hopf algebra $A^!$. 
\end{enumerate}
\end{remark}

I will give a proof of a version of this result suited to our applications in section \ref{Koszul_duality_categorical}. 

As an example, any commutative algebra has the structure of $E_n$ algebra for any $n$.  Thus, if $A$ is an augmented commutative dg algebra, we can construct a Koszul dual Hopf algebra $A^!$.  This will automatically be co-commutative, but will not in general be commutative.

\subsection{Koszul duals of cochains of a Lie algebra} Before I turn to the general discussion of Koszul duality, I will discuss a basic class of examples that will be important for the applications we will consider. 

Let $\g$ be a nilpotent Lie algebra, concentrated in cohomological degree $0$. (Nilpotent meaning all sufficiently large iterated brackets vanish).    Let $\exp \g$ be the algebraic group associated to the nilpotent Lie algebra $\g$. The underlying variety of $\exp \g$ is just $\g$, and the  product is given by the Campbell-Baker-Hausdorff formula.   Let $\Oo(G)$ denote the algebra of polynomial functions on $G$, which is a Hopf algebra in a natural way.  Note that $\Oo(G)$ is a dense subspace of the linear dual of the universal enveloping algebra $U(\g)$, and that the Hopf algebra structure on $\Oo(G)$ is dual to that on $U(\g)$. 

If $A = C^\ast(\g)$ is the Chevalley-Eilenberg cochain complex of $\g$, then a standard lemma asserts that 
$$
A^! \simeq \Oo(G)
$$ 
as a Hopf algebra. 

In our main example, we need a similar result but without the nilpotence hypothesis.   This needs a little more care, because it is not in general true that the quasi-isomorphism class of the commutative dga $C^\ast(\g)$ knows about the Lie algebra $\g$.   

There are two ways to solve this issue. One is to introduce a formal parameter $s$, and work with the Lie algebra $(\g, s[-,-])$ over the base ring $\C[[s]]$.  This has the effect of making $\g$ nilpotent. 

We will take a slightly different approach: we will work with complete filtered commutative dgas.   It is not at all necessary to read the discussion of complete filtered algebras in order to follow the essential arguments of this paper. 
\begin{definition}
\label{definition_filtered}
A complete filtered cochain complex is a cochain complex $V$ with a decreasing filtration $V = F^0 V \supset F^1 V \dots$ by sub-cochain complexes such that $V = \liminv V / F^i V$.    A map of completed filtered cochain complexes must preserve filtrations.  A map is said to be a \emph{quasi-isomorphism} if it induces an isomorphism on the cohomology groups of the associated graded.   

We let $\op{FVect}$ be the category of complete filtered vector spaces, and $\op{dgFVect}$ the category of complete filtered cochain complexes.

If $V, W$ are complete filtered cochain complexes, then $V \otimes W$ is defined to be $V \otimes W = \liminv (V / F^i V) \otimes (W / F^j W)$.  This has a natural filtration.  

If $V_i$ are a collection of complete filtered cochain complexes, the completed direct sum is defined by $\oplus V_i = \liminv_{k} \left( \oplus V_i / F^k \right)$.  This is easily seen to satisfy the usual universal property of a coproduct.
\end{definition}
Since the category of complete filtered cochain complexes is a dg symmetric monoidal category, it makes sense to talk about dg algebras in this dg category. Concretely, a complete filtered commutative dga is a completed filtered cochain complex $A$ with an associative product map $A \otimes A \to A$, compatible with filtrations.  This means that $F^i A \cdot F^j A \subset F^{i+j} A$. 

\begin{remark}
Note that every complete filtered cochain complex has a natural topology (that of the inverse limit), and the tensor product is the completed projective tensor product of topological cochain complexes. 
\end{remark}
If $\g$ is a dg Lie algebra, then $C^\ast(\g)$ as a completed filtered commutative dga, where $F^i C^\ast(\g)$ is $\Sym^{\ge i} (\g^\vee[-1])$.  

The usual constructions from homological algebra work with complete filtered objects. For instance, if $A$ is a complete filtered dga, with an augmentation $A \to \C$ (which is a map of filtered algebras), then we can define the completed Koszul dual by
$$A^! = \liminv \left( \C \otimes^{\mbb{L}}_{A / F^i A} \C\right).$$
This is modeled by the bar complex
$$
\liminv_{k} \left( \oplus_n (A/ F^k A)^{\otimes n}[n] \right).
$$
Thus, working with complete filtered objects has the effect of replacing the Koszul dual by a certain completion.  

Note that if $\g$ is any Lie algebra, then $U(\g)$ has an increasing filtration by saying that $F_k U(\g)$ is the subspace spanned by words of length $\le k$ in the generators $\g$.   Dually, this gives $U(\g)^\vee$ a decreasing filtration, by saying that $F^k \left( U(\g)^\vee \right) = \left( U(\g) / F_k U(\g) \right)^\vee$. This filtration makes $U(\g)^\vee$ into a complete filtered vector space.  The Hopf algebra structure on $U(\g)$ dualizes to one on $U(\g)^\vee$, but using the completed tensor product of filtered vector spaces.   Alternatively, we can think of $U(\g)^\vee$, with its natural topology, as a topological Hopf algebra, using the completed projective tensor product of topological vector spaces. 

In the case that $\g$ is finite dimensional, then $U(\g)^\vee$ is the Hopf algebra of functions on the formal group associated to $\g$.  
\begin{lemma}
Let $\g$ be any Lie algebra.  Then,
$$
C^\ast(\g)^! \simeq  U(\g)^\vee
$$
as Hopf algebras.
\end{lemma}
\begin{proof}
Let $A = C^\ast(\g)$. 
By definition, $A^!$ has a complete decreasing filtration. We can compute $A^!$ by a spectral sequence whose first page is the cohomology of the associated graded. It is clear that $\op{Gr} (A^!)$ is given by the bar complex for $\op{Gr} A$, so that $\op{Gr}(A^!) = \left(\op{Gr}(A)\right)^!$. 

Now, $\op{Gr}(A)$ is the exterior algebra on $\g$, so that (by a standard result) $\op{Gr}(A)^!$ is the completed symmetric algebra on $\g^\vee$, which is isomorphic to $U(\g)^\vee$ as an algebra. There are no further differentials in the spectral sequence, for degree reasons. Thus, we have an isomorphism of commutative algebras
$$
U(\g)^\vee \iso H^\ast( A^!)
$$

It remains to calculate the coproduct; this is a simple exercise. 
\end{proof} 

\section{Koszul duality and equivalences of categories}
\label{Koszul_duality_categorical}
After these preliminary examples, I will turn to developing the general theory we need. This section contains a number of technical results in homotopical algebra for which I couldn't find a convenient reference (although \cite{Pos09} comes closest).  It is not necessary to understand the proofs in order to read the rest of the paper.

We will show that, under certain hypothesis, if $A$ is a differential graded algebra in the category of complete filtered cochain complexes, then there is an equivalence of dg categories between the category of $A$-modules and $A^!$-comodules.  When $A$ is an $E_2$ algebra, we will show that $A^!$ is a Hopf algebra, and this equivalence lifts to an equivalence of monoidal dg categories.

Let $\op{FVect}$ denote the category of complete filtered vector space. Let $\op{FVect}_\hbar$ denote the category of complete filtered modules over $\C[[\hbar]]$ which are free, i.e.\ which are of the form $V \otimes \C[[\hbar]]$ for some $V \in \op{FVect}$.  Morphisms in $\op{FVect}_{\hbar}$ are $\hbar$-linear and filtration preserving; as explained earlier, the tensor product is completed.   Let $\mc{C}$ denote the dg category of cochain complexes of objects of $\op{FVect}_{\hbar}$.   We will let $\mc{C}^b \subset \mc{C}$ be the full sub dg category consisting of those objects with the property that each graded piece is bounded above and below as a cochain complex. 

The dg category $\mc{C}$ is enriched in itself: if $V,W$ are objects of $\mc{C}$, we let $F^k \Hom (V,W) \subset \Hom(V,W)$ consist of those maps which send $F^i V$ to $F^{i+k} W$.  Then, $\hbar F^k \Hom(V,W) \subset F^{k+1} \Hom(V,W)$ and 
$$
\Hom(V,W) = \liminv \Hom(V,W) / F^k \Hom (V,W). 
$$
Similarly, the dg category $\mc{C}^b$ is enriched in $\mc{C}$.

We are interested in taking homotopy limits and colimits in $\mc{C}$.  To define homotopy limits and colimits, one needs to know how to define direct sums and products, and realizations and totalizations of simplicial and cosimplicial objects.  We have already seen how to define direct sums and products in $\mc{C}$; the standard constructions of realizations and totalizations of simplicial and cosimplicial cochain complexes work equally well in $\mc{C}$. 

We will sometimes take strict limits and colimits in situations when they coincide with homotopy limits and colimits.  This happens, for example, with sequential limits (colimits) of objects of $\mc{C}^b$ whose constituent maps are injective (respectively, surjective). 

We will develop some homological algebra for algebras and coalgebras in $\mc{C}^b$.  We require that our algebras $A$ and coalgebras $C$ have the property that $A / F^1 A = \C$ and $C / F^1 C = \C$, spanned by the unit and counit.  Under this restriction, we can construct projective resolutions of $A$-modules (or dually, injective resolutions of $C$-comodules) while remaining in the dg category $\mc{C}^b$.  

The main results of this section are as follows.  
\begin{enumerate}
\item If $A$ is an algebra in $\mc{C}^b$ (or $C$ is a coalgebra), we construct dg categories $A-\op{mod}$ and $C-\op{comod}$ if quasi-free (respectively, quasi-cofree) modules. 
\item  If $A$ is an augmented algebra in $\mc{C}^b$, we construct a Koszul dual coalgebra 
$$A^! = \C[[\hbar]] \otimes^{\mbb{L}}_A \C[[\hbar]]$$
and a quasi-equivalence of dg categories
$$
A-\op{mod} \simeq A^!-\op{comod}.
$$
\item Let $A$ be as above, $\op{Perf}(A)$ be the full subcategory of $A-\op{mod}$ consisting of perfect $A$-modules (as defined below).  Let $\op{Fin}(A^!-\op{comod})$ be the full subcategory of $A^!$-comodules $M$ with the property that $H^\ast (\op{Gr}(M) \op{mod} \hbar) $ is of finite total dimension (equivalently, $H^\ast (\op{Gr} M)$ is of finite total rank as a $\C[[\hbar]]$-module).  The previous quasi-equivalence restricts to a quasi-equivalence
$$
\op{Perf}(A) \simeq \op{Fin}(A^!-\op{comod}).
$$
\item Suppose that $A$ is an augmented $E_2$ algebra in $\mc{C}^b$ such that  $H^i (A^! ) = 0$ if $i \neq 0$, and $H^0 (A^!)$ is flat as a $\C[[\hbar]]$-module.  Then, we construct a Hopf algebra structure on $H^0(A^!)$ and a quasi-equivalence of monoidal dg categories
$$
A-\op{mod} \simeq H^0 (A^!)-\op{comod}.
$$
As above, this restricts to a quasi-equivalence
$$
\op{Perf}(A) \simeq \op{Fin}(H^0(A^!)-\op{comod}).
$$
\item If $C$ is a coalgebra in $\mc{C}^b$, then the linear dual $C^\vee$ is a topological algebra. We show that there's an equivalence of dg categories between a certain dg category of $C^\vee$-modules with finite rank cohomology and $\op{Fin}(C-\op{comod})$.  
\end{enumerate}

Before we state the theorems, we need to introduce some notation.  
\begin{definition}
Let $A$ be an algebra in $\mc{C}^b$.  As always, we assume that $A/ F^1 A = \C$. We say an $A$-module $M$ is free if it is of the form $M = A \otimes V$ for an object $V$ of $\mc{C}^b$.  We say that $M$ is quasi-free if $M$ admits an increasing filtration $G^0 M = 0 \subset G^1 M \subset \dots$ by sub $A$-modules such that $M = \colim G^i M$ and such that each $G^i M /  G^{i-1} M$ is free.   We let $A-\op{mod}$ denote the dg category of quasi-free $A$-modules.
  
Let $C$ be a coalgebra in $\mc{C}^b$, which, as always, is such that $C / F^1 C = \C$, spanned by the counit. We say a $C$-comodule $M$ is cofree if it is of the form $M = C \otimes V$ for an object $V$ of $\mc{C}^b$.  We say that $M$ is quasi-cofree if $M$ admits a decreasing filtration $G^0 M = M \supset G^1 M \dots $ by sub $C$-comodules such that $M = \lim M / G^i M$ and such that each $G^i M /  G^{i+1} M$ is cofree.    We let $C-\op{comod}$ denote the dg category of quasi-cofree $C$-comodules.

If $A$ is an algebra in $\mc{C}^b$, the notion of $A_\infty$ module over $A$ makes sense: an $A_\infty$ $A$-module is an object $V$ of $\mc{C}^b$, with a collection of maps $\mu_n : A^{\otimes n} \otimes V \to V[1-n]$ for $n \ge 2$ satisfying the standard $A_\infty$ identity. (Note that the $\mu_n$ are not cochain maps).  We let $A_\infty A-\op{mod}$ denote the dg category of $A_\infty$ $A$-modules, where the morphisms are $A_\infty$ module maps. 
\end{definition}
The reason for using quasi-(co)-free modules is explained in the following proposition.
\begin{proposition}
Let $A$ be an algebra ($C$ a coalgebra) in $\mc{C}^b$.  Then, every quasi-isomorphism $f : M \to N$ of quasi-free $A$-modules (respectively, quasi-cofree $C$-comodules) is a homotopy equivalence.
\end{proposition}
\begin{proof}
First we need a lemma. 
\begin{lemma}
Let $A$ be an algebra in $\mc{C}^b$, and let $M_1,M_2,N$ be quasi-free $A$-modules.  Let $f : M_1 \to M_2$ be a quasi-isomorphism. Then, the map
$$
\Hom_A( N, M_1) \to \Hom_A(N,M_2)
$$
is a quasi-isomorphism.

Dually, let $C$ be a coalgebra in $\mc{C}^b$, and let $M_1,M_2,N$ be quasi-cofree $C$-modules.  Let $f : M_1 \to M_2$ be a quasi-isomorphism. Then, the map
$$
\Hom_C( M_2,N) \to \Hom_C(M_1,N)
$$
is a quasi-isomorphism.
\end{lemma}
\begin{proof}
Let us first verify the algebra case. By assumption, $N$ has an increasing filtration where the associated graded modules are free.  We let $G^i N$ be this filtration.  Then,
$$
\Hom_A(N,M ) = \lim \Hom_A(G^i N, M). 
$$
The map $\Hom_A(N,M_1) \to \Hom_A(N,M_2)$ arises from a map of inverse systems
$$
\Hom_A(G^i N,M_1) \to \Hom_A(G^i,M_2). 
$$
The map $\Hom_A(G^i N,M) \to \Hom_A(G^{i-1} N,M)$ is a surjective map of cochain complexes.     Therefore, by a spectral sequence argument, it suffices to show that the map
$$
\Hom_A( G^i N / G^{i-1}N, M_1) \to \Hom_A(G^i N / G^{i-1} N, M_2)
$$
is an equivalence.  But this is immediate because $G^i N / G^{i-1} N$ is free.

The argument for coalgebras is similar.  Let $N,M_1,M_2$ be quasi-cofree $C$-comodules, and let $f : M_1 \to M_2$ be a quasi-isomorphism.  Then, by assumption, $N$ has a decreasing filtration $G^i N$ such that  $N = \liminv N / G^i N$.  Then, 
$$
\Hom_C(M,N) = \liminv \Hom_C(M, N / G^i N).
$$
As above, it suffices to show that the map
$$
\Hom_C(M_2, G^i N/ G^{i+1} N) \to \Hom_C(M_1,G^i N/ G^{i+1} N)
$$
is a quasi-isomorphism. This is immediate because $G^i N/ G^{i+1} N$ is cofree. 
\end{proof}
Now let us return to the proof of the proposition.  Suppose that $f : M \to N$ is a quasi-isomorphism of $A$-modules. Then, the map
$$
\Hom_A(N,M) \xto{f \circ - } \Hom_A(N,N)
$$ 
is a quasi-isomorphism. It follows that there is a $g : N \to M$, unique up to homotopy, such that $f \circ g$ is homotopic to the identity on $N$.   This constructs a right homotopy inverse to $f$.  Now, $g : N \to M$ is also a quasi-isomorphism, so by the same argument $g$ admits a right homotopy inverse, which must be homotopic to $f$.  

The same argument applies in the comodule case. 
\end{proof}
\begin{definition}
If $\mc{D}$ is a dg category, we let $H^0 ( \mc{D})$ be the linear category whose objects are those of $\mc{D}$ but where the morphisms are $H^0$ of the morphism complexes of $\mc{D}$. 

Let $\mc{D}_1, \mc{D}_2$ are dg categories, we say that a dg functor $\Phi : \mc{D}_1 \to \mc{D}_2$ is a quasi-equivalence if
\begin{enumerate}
\item $H^0 \Phi : H^0 (\mc{D}_1) \to H^0 (\mc{D}_2)$ is an equivalence of categories. 
\item For all pairs of objects $d,d'$ of $\mc{D}_1$, the map
$$
\Phi ( d,d') : \mc{D}_1(d,d') \to \mc{D}_2 (\Phi(d), \Phi(d'))
$$
is a quasi-isomorphism.  	
\end{enumerate}
\end{definition}
The main theorem of this section is the following.
\begin{theorem}
\label{theorem_equivalence_categories}
Let $A$ be an augmented algebra in $\mc{C}^b$, which as always has the property that $A / F^1 A = \C$.  (Recall that $A$ is a complete filtered algebra over $\C[[\hbar]]$; augmented means that $A$ is equipped with a map to $\C[[\hbar]]$, and $F^1 A$ is part of the given filtration on $A$). 

Then there are quasi-equivalences of dg categories
$$
A-\op{mod} \simeq A_\infty \ A-\op{mod} \simeq A^! - \op{comod}.
$$
\end{theorem}
\begin{proof}
We will let $1_{\mc{C}^b} = \C[[\hbar]]$ denote the unit object of $\mc{C}^b$.
Let us write $A = 1_{\mc{C}^b} \oplus \br{A}$, where $\br{A} \subset A$ is the augmentation ideal. Thus, $\br{A}$ is a non-unital algebra.  By using the reduced bar complex, we can use the model
$$
A^!  = \oplus \br{A}^{\otimes n}[n] 
$$
with the differential, as before, coming from the product on $\br{A}$.   Note that, although this complex is unbounded below, because $F^1 \br{A} = \br{A}$ each graded piece of $A^!$ is bounded, so that $A^!$ is an object of $\mc{C}^b$. 

Again using the fact that $F^1 \br{A} = \br{A}$, and the fact that our direct sums are completed direct sum of filtered cochain complexes, we see that
$$
A^! = \prod \br{A} ^{\otimes n} [n].
$$

We will first construct the quasi-equivalence
$$
A_\infty \ A-\op{mod} \simeq A^! -\op{comod}.
$$
If $N$ is an $A_\infty$ $A$-module, we let 
$$
N^! = A^! \otimes N = \oplus_n \br{A}^{\otimes n}[n] \otimes N. 
$$
So far, $N^!$ is a cofree $BA$ comodule.  It is standard that the $A_\infty$ module structure on $N^!$ gives rise to a differential $\d_{N^!}$ on $N^!$ making it into a dg comodule for $A^!$.  This differential is characterized by the property that the composition
$$
\br{A}^{\otimes n}  \otimes N [n] \into N^! \xto{\d_{N^!}} \to N 
$$
(where the second map arises from the counit) is the $A_\infty$ multiplication $\mu_n$ on $N$. 

We need to verify that this is co-free. Let $F^i N$ denote the defining filtration on $N$ (recall that all objects of $\mc{C}^b$ are equipped with a filtration).  Let us define a filtration on $N^!$ by
$$
G^i (N^!) = A^! \otimes F^i N.
$$
This filtration is by sub-$A^!$-comodules, and our definitions guarantee that
$$
N^! = \liminv N^! / G^i N^!.
$$ 
The fact that $A^! = \oplus \br{A}^{\otimes n}$ and $F^1 \br{A} = \br{A}$ guarantees that $G^i N^! / G^{i+1} N^!$ is cofree as an $A^!$-comodule.   Thus, we have shown that $N^!$ is quasi-cofree. 

It is standard that if $M, N$ are $A_\infty$ $A$-modules, then $A_\infty$-maps $M \to N$ are the same as maps of $A^!$-comodules $M^! \to N^!$.   

Finally, we need to verify that every object of $A^!-\op{comod}$ is homotopy equivalent to one coming from an $A_\infty$ $A$-module.  We will show something stronger: every object of $A^!-\op{comod}$ is actually \emph{isomorphic} to one coming from an $A_\infty$ $A$-module. 

 Let $M$ be a quasi-co-free $A^!$-comodule.  Then, there exists an object $V$ of $\mc{C}^b$ and a differential $\d_M$ on $A^! \otimes V$ with the following properties:
\begin{enumerate}
\item $\d_M$ is $A^!$-colinear.
\item There is an isomorphism of $A^!$-comodules
$$
(A^! \otimes V, \d_M) \iso M.
$$
\item The coaugmentation map $V \to A^! \otimes V$ is a cochain map. 
\end{enumerate}
The differential $\d_M$ on $A^! \otimes V$ can be decomposed into terms
$$
\d_N^{(k)} : \br{A}^{\otimes n} \otimes V \to V,
$$
where $\d_N^{(0)}$ is the given differential on $V$.    It is standard that the maps $\d_N^{(k)}$ endow $V$ with the structure of $A_\infty$-module over $A$.

Next, we need to construct the quasi-equivalence between the dg category $A-\op{mod}$ of quasi-free $A$-modules, and the dg category of $A_\infty$ $A$-modules.  The first thing to verify is that, if $M,N$ are quasi-free $A$-modules, that the natural map
$$
\Hom_{A-\op{mod}} (M,N) \to \Hom_{A_\infty\ A-\op{mod}} (M,N)
$$
is a quasi-isomorphism.  If $M$ is a quasi-free $A$-module, let $BM$ be the standard bar resolution of $M$, which looks like
$$
BM = \oplus_{n \ge 0} A \otimes \br{A}^{\otimes n} \otimes M [n]
$$
(in other words, we take the standard bar complex model for the derived tensor product $A \otimes_{A}^{\mbb{L}} M$).  

Then, $A$-module maps $B M \to N$ are the same as $A_\infty$ $A$-module maps from $M$ to $N$.  Further, since $M$ is quasi-free, the natural quasi-isomorphism $B M \to M$ lifts to give a homotopy equivalence between $M$ and $BM$.  It follows that there is a homotopy equivalence
$$
\Hom_{A-\op{mod}} ( B M, N) \simeq \Hom_{A-\op{mod}} (M, N). 
$$

The final thing we need to check is that every $A_\infty$ $A$-module is homotopy equivalent to a strict $A$-module.  To see this, if $M$ is an $A_\infty$ $A$-module, then we can define (as above) the bar resolution
$$
BM = \oplus_{n \ge 0} A \otimes \br{A}^{\otimes n} \otimes M[n] 
$$
where the differential is defined using the higher product maps $\mu_n : \br{A}^{\otimes n} \otimes M \to M$.    Then, $BM$ is a quasi-free $A$-module, and there is a homotopy equivalence of $A_\infty$ $A$-modules between $BM$ and $M$.  

\end{proof}

\begin{lemma}
Let $A$ be as in the previous theorem.  Then, there is a natural quasi-isomorphism between the Hochschild chain complex of $A$ and the co-Hochschild chain complex of $A^!$. 
\label{lemma_coHochschild}
\end{lemma}
\begin{remark}
\begin{enumerate}
\item A much more general version of this statement -- involving factorization homology of $E_n$ algebras  -- has been proved by John Francis and (independently) by Takuo Matsuoka.   For completeness, however, I will give a concrete proof of this special case. 
\item The co-Hochschild chain complex of a coalgebra $C$ is defined in the same way as the Hochschild complex: it is the derived cotensor product
$$
C \odot^{\mbb{R}}_{C \otimes C} C. 
$$
\end{enumerate}
\end{remark}
\begin{proof}
We let $C = A^!$, and we let $1$ denote the coaugmentation comodule for $C$ (or rather, a quasi-cofree resolution for this).  Note that, under the quasi-equivalence of dg categories between $A-\op{mod}$ and $A^!-\op{comod}$, the $A$-module $A$ maps to $1$.    It follows that $A$ is quasi-isomorphic to $\Hom_C(1,1)$. We can compute $\Hom_C(1,1)$ using an injective cobar resolution for the trivial comodule: the result is 
$$
\Hom_C(1,1) = \prod_n \br{C} [-n]
$$
with a differential arising from the coproduct on $C$.  Note that because $F^1 \br{C} = \br{C}$, we can rewrite this as 
$$
\Hom_C(1,1) = \oplus_n \br{C} [-n].
$$
As an algebra, $\Hom_C(1,1)$ is the free algebra on $\br{C}[-1]$.   We will denote this algebra by $C^!$. 

We can write down an explicit quasi-free resolution for $C^!$ as a $C^!$-bimodule, as follows.  Multiplication $m : C^! \otimes C^! \to C^!$ is a bimodule map.  The kernel $I \subset C^! \otimes C^!$ is the two-sided ideal generated by the subspace $\alpha \otimes 1 - 1 \otimes \alpha$ for $\alpha \in \br{C}[-1]$; one can check easily that $I$ is freely generated by this subspace, so that 
$$
M = \{ I[1] \to C^! \otimes C^!\}
$$
is a quasi-free resolution for the diagonal bimodule.  

We tensor this over $C^! \otimes C^!$ with $C^!$, to get the Hochschild complex.  We find that 
$$
M \otimes_{C^! \otimes C^!} C^! = \{ I[1] \otimes_{C^! \otimes C^!} C^! \to C^!\}.
$$
Note that the underlying graded vector space of $I[1] \otimes_{C^! \otimes C^!} C^!$ is $\br{C} \otimes C^!$, which we can expand as 
$$
I[1] \otimes_{C^! \otimes C^!} C^! = \prod_n \br{C} \otimes \br{C}^{\otimes n}[-n].
$$
If we analyze the differential, we find that this complex is the standard cobar complex model for the co-Hochschild complex of $C$ with coefficients in $\br{C}$, the coaugmentation coideal.  Including $C^!$, which is the other part of $M$, gives us the standard cobar model for the co-Hochschild complex of $C$ with coefficients in $C$.  
\end{proof}

\subsection{}
Let $V \in \mc{C}^b$.  We say that $V$ is \emph{finite} if $H^\ast ( \op{Gr} V \op{mod} \hbar)$ is of finite total dimension.  
 
Let $M = V \otimes_{\C[[\hbar]]} A$ be a free module, where $V \in \mc{C}^b$.  We say that $M$ is of finite rank if $V$ is finite. 

We say an $A$-module $M$ is \emph{perfect} if the $\op{Gr} A$-module $\op{Gr} M$ is quasi-isomorphic to one obtained from free $\op{Gr} A$ modules of finite rank by a finite number of iterated cones.   We let $\op{Perf}(A) \subset A-\op{mod}$ be the subcategory of perfect modules. 

If $N$ is an $A^!$-comodule, we say that $N$ is finite if it is finite when considered as an object of $\mc{C}^b$.   We let $\op{Fin}(A^!) \subset A^!-\op{comod}$ be the full sub-dg category of finite $A^!$-comodules.
\begin{proposition}
Under the hypotheses of the theorem, the quasi-equivalence 
$$
A-\op{mod} \to A^!-\op{comod}
$$
restricts to a quasi-equivalence
$$
\op{Perf}(A) \to \op{Fin}(A^!).
$$
\label{proposition_perfect_finite}
\end{proposition}
\begin{proof}
Let $M \in \op{Perf}(A)$. We need to show that 
$1 \otimes^{\mbb{L}}_A M$ is finite.  By assumption, $\op{Gr} M$ is quasi-isomorphic to  a finite extension of modules of the form $V \otimes \op{Gr} A$, where $V$ is finite.  Note that
$$
\op{Gr} \left(1 \otimes_A^{\mbb{L}} M  \right) = 1 \otimes_{\op{Gr}A}^{\mbb{L}} \op{Gr} M.
$$
The functor $1 \otimes_{\op{Gr} A}^{\mbb{L}} - $ sends cones to cones. Thus, it suffices to check that for every free $\op{Gr} A$-module $N$ of finite rank,  $1 \otimes_{\op{Gr}A}^{\mbb{L}} N$ is finite. This is clear.  

Conversely, suppose that $M$ has the property that $1 \otimes^{\mbb{L}}_A M$ is finite.  We need to verify that $M$ is perfect.   Let us choose a free resolution of $\op{Gr} M$, so that we have a quasi-isomorphism
$$
\op{Gr} M \iso V \otimes_{\C} \op{Gr} A
$$
where $V$ is a graded cochain complex over $\C$. Note that
$$
\op{Gr} \left(1 \otimes_A^{\mbb{L}} M  \right) \op{mod} \hbar = V.
$$
Thus, by assumption, each graded piece of $V$ has finite cohomology. 

Let $V^i$ denote the $i$'th graded piece of $V$. Note that $V^i$ is zero if $i < 0$.  Then, the differential on $V \otimes \op{Gr} A$ maps
$$
V^i \otimes \mapsto V^i \oplus_{0 \le j < i} V^j \otimes \op{Gr}^{i - j} A.
$$
Let 
$$F^i (V \otimes \op{Gr} A) = \oplus_{j \le i} V^j\otimes \op{Gr} A.$$
This is a filtration by dg $\op{Gr}(A)$-modules.  Since $H^\ast(V^i) = 0$ for $i$ sufficiently large, $V \otimes \op{Gr} A$ is quasi-isomorphic to $F^i(V \otimes \op{Gr} A)$ for some $i$ sufficiently large. 

Thus, we have shown that $\op{Gr} M$ has a finite filtration whose associated graded modules are free of finite rank.  This tells us that $\op{Gr} (M)$ can be expressed in terms of free modules of finite rank by a finite number of iterated cones, as desired. 
\end{proof}

We will leverage proposition \ref{proposition_perfect_finite} to prove a result about Hochschild homology.  Let $A$ be an algebra in the dg category $\mc{C}^b$, and let us assume (as in the statement of theorem \ref{theorem_equivalence_categories}) that $A / F^1 A = \C$.  
\begin{proposition}
There is a natural isomorphism
$$
HC( \op{Perf}(A) ) \iso HC( A) 
$$
between the Hochschild homology groups of the dg category $\op{Perf}(A)$ and those of $A$. 
\label{prop:HH_iso_perfect}
\end{proposition}
\begin{remark}
We use here the fact that $A-\op{mod}$ is enriched in $\mc{C}$, and the Hochschild chain complexes of $\op{Perf}(A)$ and of $A$ is defined using colimits in $\mc{C}$.  This means we use completed direct sums. 
\end{remark}
\begin{proof}
Although this is standard, our definition of perfect may appear slightly unorthodox, so I will sketch a proof.  Since $A$ is a perfect module over itself, there's a natural map $HC(A) \to HC(\op{Perf}(A))$. We want to show that this map is a quasi-isomorphism.  Everything has a complete decreasing filtration, so it suffices to show that this map is a quasi-isomorphism on the associated graded.  Thus, we need to verify that the map
$$
HC(\op{Gr} A) \to HC( \op{Perf} \op{Gr} A))
$$
is a quasi-isomorphism. We defined a perfect $A$-module to be one whose associated graded is a finite iterated extension of free $\op{Gr} A$-modules of finite rank.   

The dg category $\op{Perf}(A)$ is enriched in $\mc{C}^b$; thus, the $\Hom$-complexes are filtered.  We can therefore consider $\op{Gr} \op{Perf}(A)$, the associated graded dg category, obtained by replacing the $\Hom$-complexes by their associated graded.  

We can identify $\op{Gr} \op{Perf}(A)$ with the category of graded $\op{Gr} A$-modules which are finite iterated extensions of free modules. Now standard results apply. 
\end{proof}
\begin{corollary}
There is a natural quasi-isomorphism
$$
HC( A) \iso HC( \op{Fin}(A^!)) 
$$
between the Hochschild chain groups of $A$ and that of the category of finite-rank $A^!$ comodules. 
\end{corollary}

\begin{proposition}
Suppose that $C$ is a coalgebra in $\mc{C}^b$, that $H^i(C) = 0$ for $i \neq 0$, and that $H^0 (C)$ is a flat $\C[[\hbar]]$-module. Then there is an equivalence of dg categories
$$
C-\op{comod} \simeq H^0(C)-\op{comod}.
$$
\end{proposition}
\begin{proof}
We will first produce chain of quasi-isomorphisms connecting the coalgebras $C$ and $H^0(C)$.  Define a coalgebra $D$ by saying that in cohomological degree $i > 0$, $D^i = C^i$, in degree $i  < 0$, $D^i = 0$, and $D^0 = C^0 / \op{Im} \d$.  Then, $D$ has a unique coalgebra structure such that the map $C \to D$ is a coalgebra map.  The fact that $H^0(C)$ is flat over $\C[[\hbar]]$ implies that $D$ is also, so that is a coalgebra in the dg category $\mc{C}^b$.  There are quasi-isomorphisms $H^0(C) \to D \from C$.  

Next, we need to verify that if $D,C$ are coalgebras in $\mc{C}^b$, that a quasi-isomorphism $\Phi : C \to D$ of coalgebras induces a quasi-equivalence $C-\op{comod} \simeq D-\op{comod}$.  We assume, as always, that $C / F^1 C = \C$ and $D / F^1 D = \C$.  We right $\br{D}$ and $\br{C}$ for the kernel of the counit map $C \to \C[[\hbar]]$, $D \to \C[[\hbar]]$. 

We define functors 
\begin{align*}
\Phi_\ast : C-\op{comod} &\to  D-\op{comod}\\ 
\Phi^\ast : D-\op{comod} &\to C-\op{comod}
\end{align*}
as follows. If $M$ is a $C$-comodule, then we can consider it as a $D$-comodule; however, it may not be quasi-free as a $D$-comodule.  We let $\Phi_\ast(M)$ be the canonical cobar resolution of $M$ as a $D$-comodule, thus, $\Phi_\ast(M)$ looks like
$$
\Phi_\ast(M) = M \odot_D^{\mbb{R}} D = \prod (M \otimes \br{D}^{\otimes n + 1} ) [-n].
$$
Dually, if $N$ is a $D$-comodule, we let 
$$
\Phi^\ast (N) = N \odot_D^{\mbb{R}} C.
$$
Again, we can model the right-derived cotensor product as a bar complex.  

We will construct natural quasi-isomorphisms of functors $\Phi^\ast \Phi_\ast \simeq \op{Id}_{C-\op{comod}}$ and $\Phi_\ast \Phi^\ast \simeq \op{Id}_{D-\op{comod}}$. 

If $M$ is a $C$-comodule, there is a natural quasi-isomorphism
$$
M \odot^{\mbb{R}}_D C \to \Phi^\ast \Phi_\ast (M).
$$
Note that $M \odot^{\mbb{R}}_D C$ is the totalization of a cosimplicial $C$-comodule whose $n$-simplices are $M \otimes \br{D}^{\otimes n} \otimes C$. Call this cosimplicial object $M \odot_D^{\tr} C$.  Similarly, the canonical cobar resolution $M \odot_{C}^{\mbb{R}} C$ is the totalization of a cosimplicial object $M \odot_C^{\tr} C$ whose $n$-simplices are $M \otimes \br{C}^{\otimes n} \otimes C$. There is a map of cosimplicial $C$-comodules
$$
M \odot_C^{\tr} C \to M \odot_D^{\tr} C
$$
arising from the coalgebra map $\Phi : C \to D$. Since $\Phi$ is a quasi-isomorphism, this map of cosimplicial objects is a levelwise quasi-isomorphism, and so becomes a quasi-isomorphism upon totalization. 

Similarly, if $N$ is a $D$-comodule, then 
$$
\Phi_\ast \Phi^\ast N \simeq N \odot^{\mbb{R}}_D C.
$$
The map $C \to D$ induces a map
$$
N \odot^{\tr}_{D} C \to N \odot^{\tr}_D D
$$
of cosimplicial $D$-comodules, which is again a levelwise equivalence, so that we have a natural quasi-isomorphism $\Phi_\ast \Phi^\ast N \simeq N$. 
\end{proof}

\begin{corollary}
Suppose that $A$ is an algebra in $\mc{C}^b$ satisfying the hypothesis of theorem \ref{theorem_equivalence_categories}.  Suppose that $H^i (A^!) = 0$ for $i \neq 0$ and that $H^0 (A^!)$ is flat over $\C[[\hbar]]$.  Then, there is a quasi-equivalence of dg categories
$$
\op{Perf}(A) \simeq \op{Fin} (H^0(A^!)),
$$
and a quasi-isomorphism of Hochschild chain complexes
$$
HC( A) \iso HC( \op{Fin} (H^0(A^!)) .
$$
\end{corollary}

\subsection{The dual of a coalgebra is an algebra}
\label{sub:dual_coalgebra}
Before we move on to a discussion of the Koszul dual of $E_2$ algebras, we need one more lemma, which relates modules for a coalgebra $C$ to modules for a linear dual associative algebra $C^\vee$.  First, we need to construct the linear dual of an object of $\mc{C}^b$.
\begin{definition}
If $V$ is a vector space over $\C$, let $V^\vee$ denote the topological vector space equipped with the topology defined by the inverse limit over all maps $V^\vee \to W^\vee$, where $W \subset V$ is a finite-dimensional subspace. 
\end{definition}
\begin{definition}
Let $V$ be a complete filtered $\C$-vector space.  Let $V^\vee$ be the topological vector space with an increasing filtration defined by 
\begin{align*}
F^i V^\vee &= (V / F^i V)\\
V^\vee &= \colim_i F^i V^\vee.
\end{align*}
\end{definition}
\begin{definition}
Let $V$ be a free complete filtered module over $\C[[\hbar]]$.  Let $V^\vee$ denote the topological module over $\C[[\hbar]]$, equipped with an increasing filtration, defined as follows.  

Let $V_n$ denote $V$ modulo $\hbar^n$.  For $i \in \Z$, let $F^i V_n^\vee$ be the space of $\C[\hbar]/\hbar^n$-linear maps
$$
\{ f : V \to \C[\hbar]/\hbar^n : f (F^j(V_n)) \subset \hbar^{i+j} \C[\hbar] / \hbar^n \}. 
$$
Note that $F^i V^\vee_n$ has a topology, as it is a closed subspace of the tensor product of the $\C$-linear dual of $V$ with $\C[\hbar]/\hbar^n$. 

Let 
$$
V^\vee_n = \colim_i F^i V_n^\vee.
$$
Let
$$
V^\vee = \liminv V_n^\vee
$$
where the inverse limit is taken in the category of topological vector spaces.  

Note that $V^\vee$ has a filtration by closed subspaces, indexed by $\Z$, defined by
$$
F^i V^\vee = \liminv F^i V^\vee.
$$
Further, $V^\vee$ has the following completeness property:
$$
V^\vee =  \liminv_i \limdir_{j\le i} F^j (V^\vee) /F^i (V^\vee) 
$$
\end{definition}
\begin{remark}
\begin{enumerate}
\item 
If we choose an isomorphism $V = W[[\hbar]]$, where $W$ is a complete filtered vector space, then we get an isomorphism of topological vector spaces
$$
V^\vee \iso W^\vee \what{\otimes}_{\pi} \C[[\hbar]] = \liminv_n W^\vee \otimes \C[\hbar] / \hbar^n,
$$
where $\what{\otimes}_{\pi}$ is the completed projective tensor product.
\item
If every graded piece of $V$ is of countable dimension, then $V^\vee$ is a nuclear space. The category of nuclear spaces is particularly well behaved with respect to the completed projective tensor product. 
\end{enumerate}
\end{remark}
\begin{lemma}
Let $C$ be a coalgebra in $\mc{C}^b$. Let $C^\vee$ be defined as in the previous definition.  Then, $C^\vee$ has the structure of a filtered topological dg algebra: there is a separately continuous $\C[[\hbar]]$-linear multiplication map
$$
C^\vee \times C^\vee \to C^\vee
$$
dual to the comultiplication on $C$, satisfying the axioms of an associative algebra. 
\end{lemma}
\begin{remark}
Separately continuous means that it is continuous if one of the variables is held fixed. 
\end{remark}
\begin{proof}
This is immediate. 
\end{proof}

We will show that there is an equivalence of categories between a certain category of ind-$C^\vee$-modules and the category of $C$-comodules.  To start with, we will show an equivalence of ordinary (non-dg) categories. By taking complexes, we will find an equivalence of dg categories. 

We say a (non-dg) comodule $N$ for a coalgebra $C$ is \emph{filtered-finite} if the cohomology of each $N / F^i N$ is finite dimensional as a $\C$-vector space.  We let $F \op{Fin}(C)$ denote the ordinary (not dg) category of filtered finite comodules. 

We need a similar definition for $C^\vee$. We will let $F \op{Fin}(C^\vee)$ be the category of filtered finite $\C[[\hbar]]$-modules $M$ together with a $\C[[\hbar]]$-linear continuous action
$$
C^\vee \times M \to M
$$
of $\C^\vee$ on $M$, preserving filtrations on both sides.  

\begin{lemma}
There is an equivalence of categories
$$
F \op{Fin}(C) \simeq F \op{Fin}(C^\vee). 
$$
\end{lemma}
\begin{proof}

We need to verify that if $M$ is a filtered-finite $\C[[\hbar]]$-module, then linear maps 
$$
M \to M \otimes C
$$
are the same thing as continuous linear maps
$$
C^\vee \otimes M \to M.
$$
This is a question of functional analysis.  $U,V,W$ be complete filtered vector spaces, where each filtered piece of $U$ and $V$ is of finite dimension, and each filtered piece of $W$ is of countable dimension.  We need to construct a natural isomorphism
$$
\Hom(U[[\hbar]], (V\otimes W)[[\hbar]] ) = \Hom ((W^\ast \otimes U)[[\hbar]], V[[\hbar]]) 
$$
where maps are $\C[[\hbar]]$-linear and filtration preserving; on the right hand side we give $W^\vee$ a topology as above and consider continuous linear maps.  

We can address this question by working modulo $\hbar$.  Let us further choose splittings on the filtrations of $U,V,W$, so that we write $U = \prod_{i \ge 0} U_i$ where $F^n U = \prod_{i \ge n} U_i$, and similarly for $U$ and $V$. 

Then, the space of filtration-preserving maps from $U$ to $V \otimes W$ is 
$$
\prod_{i \le j + k} \Hom (U_i, V_j \otimes W_k). 
$$
Similarly, the space of filtration-preserving continuous maps from $U \otimes W^\vee$ to $V$ is
$$
\prod_{i \le j + k} \Hom (U_i \otimes W_k^\vee, V_j),
$$
where $\Hom$ denotes continuous linear maps. We thus need to identify
$$
\Hom (U_i \otimes W_k^\vee, V_j) = \Hom (U_i, V_j \otimes W_k).
$$
Since $U_i$ and $V_j$ are finite dimensional, it suffices to show that
$$
\Hom(W_k^\vee, \C) = W_k.
$$
It is a standard fact that the continuous linear double dual of a countable dimensional vector space $W_k$ is $W_k$.
\end{proof}
\begin{lemma}
Every (non-dg) $C$-comodule is a colimit of filtered-finite comodules.
\end{lemma}
\begin{proof}
Let $M$ be a $C$-comodule. Let $M_n = M / F^n M$, and $C_n = C / F^n C$.  Let $m \in M_n$.  Then, the coaction $\tr(m)$ of $m$ is an element of the algebraic tensor product $C_n \otimes M_n$, and so can be written as a finite sum
$$
\tr(m) = \sum_{finite} c_i \otimes m_i
$$
for some $c_i \in C_n$, $m_i \in M_n$.   The $m_i$ appearing in this sum span a finite-dimensional subspace $V$ of $M_n$.  The coassociativity of the coaction of $C_n$ on $M_n$ implies that this subspace $V$ is a sub-comodule.  The counit axiom implies that $m \in V$.
\end{proof}
Let $\op{Ind} F \op{Fin}(C^\vee)$ denote the category of ind-objects of $F \op{Fin}(C^\vee)$.
\begin{corollary}
There is an equivalence of categories between $\op{Ind} F \op{Fin}(C^\vee)$ and the category of non-dg $C$-comodules. 
\end{corollary}
Recall that the dg category $C-\op{comod}$ has, as objects, complexes of cofree $C$-comodules which at each filtered level are bounded above and below. 

\begin{definition}
Let $IFF (C^\vee)$ be the category of complexes of injective objects of $\op{Ind} F \op{Fin}(C^\vee)$, which at each filtered level are bounded above and below. 
\end{definition}
\begin{proposition}
Then, there is a quasi-equivalence of dg categories 
$$
IFF (C^\vee) \simeq C-\op{comod}.
$$
\end{proposition}
\begin{proof}
Note that a cofree $C^\vee$-comodule is automatically injective. 

The only thing that we need to check is that a complex of injective $C$-comodules is homotopy equivalent to a complex of cofree ones.  This follows from the fact that any injective comodule is homotopy equivalent to its bar coresolution. 
\end{proof}

As a corollary, we find the following.  Let $\op{Fin}(C^\vee)$ denote the dg category of objects of $IFF(C^\vee)$ with finite-rank cohomology.  
\begin{corollary}
\label{corollary_dual_algebra_finite_equivalence}
There is an equivalence of dg categories
$$
\op{Fin}(C^\vee) \simeq \op{Fin}(C). 
$$
\end{corollary}

\subsection{Koszul duality between $E_2$ and Hopf algebras}
\label{subsection_Koszul_E2}
Let $A$ be an $E_2$ algebra in $\mc{C}^b$, which as always satisfies $A / F^1 A = \C$.  Then, we can view as an $E_1$ algebra, and so an associative dg algebra; as above, we can construct a category $A-\op{mod}$ of left $A$-modules.  The fact that $A$ is $E_2$ means that $A-\op{mod}$ is a monoidal category in a homotopical sense.  This monoidal structure is implemented by multi-modules: for example, the product on $A-\op{mod}$ arises from an $A\otimes A-A$-bimodule structure on $A$. Higher coherences are given by quasi-isomorphisms, homotopies between quasi-isomorphisms, etc. between $A^{\otimes n} - A$ bimodules.  

One can also describe this structure in the language of $(\infty,1)$-categories.  As explained in \cite{Lur12}, every dg category gives rise to an $(\infty,1)$-category, by a version of the nerve construction.  This construction takes quasi-equivalences of dg categories to equivalences of $(\infty,1)$-categories.   Then, $A-\op{mod}$ is an $E_1$ $(\infty,1)$-category.    

Suppose that $A$ is augmented as an $E_1$ algebra. Then, because of the quasi-equivalence of dg categories $A-\op{mod} \simeq A^!-\op{comod}$, we see that $A^!-\op{comod}$ acquires the structure of a homotopy monoidal category in the sense described above.

If we suppose that $H^i(A^!) = 0$ for $i \neq 0$ and $H^0(A^!)$ is flat as a $\C[[\hbar]]$-module, then we have a further equivalence
$$
A-\op{mod} \simeq H^0(A^!)-\op{comod}
$$
so that $H^0(A^!)-\op{comod}$ acquires a monoidal structure. 

We will prove a version of Tamarkin's theorem, showing that $H^0(A^!)$ has the structure of a bialgebra which induces this monoidal structure. 

The result we need is the following.
\begin{proposition}
Let $C$ be a coalgebra in $\mc{C}^b$, concentrated in cohomological degree $0$, and suppose that $C-\op{comod}$ is equipped with a monoidal structure satisfying the following properties.
\begin{enumerate}
\item The forgetful functor $C-\op{comod} \to \mc{C}^b$ is monoidal.
\item The monoidal structure on $C-\op{comod}$ commutes with those homotopy limits and colimits which exist in $C-\op{comod}$. 
\end{enumerate}
Then, $C$ has a unique bialgebra structure inducing this monoidal structure on $C-\op{comod}$.
\end{proposition}
\begin{remark}
Note that the monoidal structure on $\mc{C}^b$ commutes with products and direct sums, and so with those homotopy limits and colimits that exist.  This is a consequence of the fact that we use a completed tensor product in $\mc{C}^b$, and all filtered cochain complexes in $\mc{C}^b$ are bounded on each graded piece.
\end{remark}
\begin{proof}
Let $C-\op{comod}^0$ refer to the ordinary category of free non-dg $C$-comodules.  The first thing we will show is that our monoidal structure is determined by its restriction to $C-\op{comod}^0$. The point is that the whole dg category $C-\op{comod}$ is the category of complexes of objects of $C-\op{comod}^{0}$, which are at each graded level bounded above and below. 

Any functor preserving limits and colimits commutes with the formation of complexes, so that the monoidal structure is determined uniquely by its restriction to $C-\op{comod}^0$.  

$C-\op{comod}^0$ is a monoidal category in the classical sense: there is no room for higher coherences beyond the classical MacLane coherence.  The functor $C-\op{comod}^0 \to \mc{C}^0$ is monoidal.  It follows easily from standard Tannaka-Krein type results (see \cite{EtiGelNikOst09}) that $C$ has a bialgebra structure inducing this monoidal structure.   

\end{proof}

Now, suppose that $A$ is an augmented $E_2$ algebra in $\mc{C}$, that $H^0(A^!)$ is flat over $\C[[\hbar]]$, and that $H^i (A^!) = 0$ if $i \neq 0$.  Suppose that $\op{Gr} A$ is commutative. Then we have the following
\begin{corollary}
$H^0(A^!)$ has a Hopf algebra structure such that the equivalence of categories
$$
A^!-\op{mod} \simeq H^0(A^!)-\op{comod}
$$
is monoidal.
\end{corollary}
\begin{proof}
We need to verify that monoidal structure on $A-\op{mod}$ distributes over limits and colimits, and that the functor to $\mc{C}^b$ is monoidal. To check the first two points, it suffices to check it on $\op{Gr} A$.   On $\op{Gr} A$, the monoidal structure is simply tensoring over $A$, which (under our assumptions) commutes with limits and colimits. 

The fact that the functor $A-\op{mod} \to \mc{C}^b$ is monoidal follows from the fact that the augmentation map $A \to 1_{\mc{C}}$ is a map of $E_2$ algebras. 
\end{proof} 

\subsection{}
\label{subsection_laxHopf}
In this subsection, we will consider a technical generalization of our Koszul duality results which we will need when considering the operator product expansion in the field theory.

If $V$ is an object of $\mc{C}^b$, we let 
$$
V((\lambda))  = V \what{\otimes} \C((\lambda)) := \liminv_i \colim_k \liminv_j \lambda^{-k} (V/ F^i V)[\lambda] / \lambda^j.
$$
We let $V \br{\otimes} \C((\lambda))$ be the algebraic tensor product, $$
V \br{\otimes} \C((\lambda)) = \liminv_i V/F^i V \otimes \C((\lambda))
$$
where the tensor product on the right hand side is simply the algebraic one, which allows only finite sums.

If $A$ is an algebra in $\mc{C}^b$, then $A((\lambda))$ is a $\C((\lambda))$-linear algebra in $\mc{C}^b$. Similarly, if $A$ is an $E_2$ algebra in $\mc{C}^b$, then $A((\lambda))$ is a $\C((\lambda))$-linear $E_2$ algebra in $\mc{C}^b$. 

If $C$ is a coalgebra in $\mc{C}^b$, then $C((\lambda))$ is a lax $\C((\lambda))$-linear coalgebra, whose $n$-ary space is $\left(C^{\otimes n}\right) ((\lambda))$.  For instance, the coproduct is a $\C((\lambda))$-linear map
$$
C((\lambda)) \to (C^{\otimes 2})((\lambda)).
$$

Our definition of free and quasi-free $A$-modules from before extends without any changes to a definition of free and quasi-free $A((\lambda))$-modules: we don't take the $\lambda$-adic topology on $A((\lambda))$ into account when considering such modules.  For instance, a free module over $A((\lambda))$ is one of the form 
$$
V \br{\otimes} \left( A((\lambda)) \right) 
$$
for an object $V \in \mc{C}^b$. A quasi-free module is a module $M$  which admits a countable increasing filtration $G^0 M \subset G^1 M \subset \dots$ where each $G^i M / G^{i-1} M$ is free and $M = \colim G^i M$. 

We let $A((\lambda))-\op{mod}$ be the dg category of quasi-free $A((\lambda))$-modules.

If $C$ is a coalgebra in $\mc{C}^b$, defining the category of $\C((\lambda))$-linear $C((\lambda))$-comodules is a little more tricky.  The problem is that cofree comodules don't exist, because the coaction of $C((\lambda))$ on itself lands in $(C^{\otimes 2})((\lambda))$, which does not coincide with $\C((\lambda))$-linear algebraic tensor square of $C((\lambda))$.  

To get around this fact, we will use lax $A_\infty$ $C((\lambda))$-linear comodules.
\begin{definition}
Let $V$ be a $\C((\lambda))$-module in $\mc{C}^b$ (all such are of the form $V = W \br{\otimes} \C((\lambda))$ for some $W \in \mc{C}^b$). 

A lax $A_\infty$ coaction of $C((\lambda))$ on $V$ is a sequence of linear maps
$$
\tr_n : V \mapsto V \br{\otimes}_{\C((\lambda))} \left\{ C^{\otimes n} ((\lambda))  \right\}  = W \br{\otimes} \left\{ C^{\otimes n} ((\lambda))  \right\}										
$$
satisfying the standard $A_\infty$ identities. 

If $V,V'$ are lax $A_\infty$ $C((\lambda))$ coalgebras, a map $V \to V'$ is a sequence of maps
$$
V \to V' \br{\otimes} \left\{ C^{\otimes n}((\lambda)) \right\}
$$
for each $n \ge 1$, satisfying the standard $A_\infty$ identities. 

We let $C((\lambda))-\op{comod}$ denote the dg category of lax $A_\infty$ $C((\lambda))$-comodules. 
\end{definition}

\begin{proposition}
Let $A$ be an $E_2$ algebra, whose Koszul dual $A^!$ is a Hopf algebra.  (For simplicity, let us assume that we choose a model of $A^!$ which is concentrated in cohomological degree $0$). 

Then, there is an equivalence of monoidal categories
$$
A((\lambda))-\op{mod} \simeq A^!((\lambda))-\op{comod}.
$$
\label{proposition_equivalence_lambda}
\end{proposition}
\begin{proof}
As a consequence of the equivalence of monoidal categories 
$$A-\op{mod} \simeq A^!-\op{comod}$$
we see that there is an isomorphism of $E_2$ algebras
$$
A \simeq \RHom_{A^!-\op{comod}} (1,1)
$$
where $1 = \C[[\hbar]]$ denotes the trivial comodule.  

We can model $\RHom_{A^!}(1,1)$ by using the cobar resolution of $1$, which looks like
$$
1 \simeq \prod_{n \ge 0} \br{A^!}^{\otimes n}[-n] \otimes A^!,
$$
with a differential coming from the comultiplication on $A^!$, and a certain $E_2$ structure which we don't make explicit but which arises from the product on $A^!$. 

Thus, 
$$
\RHom_{A^!}(1,1) = \prod_{n \ge 0} \br{A^!}^{\otimes n}[-n] = \oplus_n \br{A^!}^{\otimes n}[-n] = (A^!)^!.
$$
The $E_1$ product on $(A^!)^!$ in this model makes it a free algebra generated by $\br{A}^![-1]$, with a differential coming from the coproduct on $A^!$. 

We will show that, for all coalgebras $C$, there is an equivalence
$$
C((\lambda))-\op{comod} \simeq C^!((\lambda))-\op{mod}.
$$
We're interested in co-$A_\infty$ $C((\lambda))$-comodules whose underlying $\C((\lambda))$-module is of the form $V \br{\otimes} \C((\lambda))$, for some $V \in \mc{C}^b$.  The $A_\infty$-coproducts are a sequence of maps sequence of $\C((\lambda))$-linear maps
$$
\tr_n : V\br{\otimes} \C((\lambda)) \to V \br{\otimes} \left\{C^{\otimes n}((\lambda)) \right\}.   
$$
Thus, giving an $A_\infty$ comodule structure on $V$ is the same thing as giving a $\C((\lambda))$-linear differential
$$
\d : V \br{\otimes} \left\{ C^!((\lambda)) \right\} \to V  \br{\otimes} \left\{ C^!((\lambda)) \right\}$$
which is also linear over $C^!((\lambda))$. Thus, $V \br{\otimes} \left\{C^!((\lambda))\right\}$ is a quasi-free $C^!((\lambda))$-module.  This constructs the functor 
$$
C((\lambda))-\op{comod} \to C^!((\lambda))-\op{mod}.
$$
This functor is full and faithful (at the cochain level). We need to verify that it is essentially surjective.  This follows from the fact that every quasi-free $C^!((\lambda))$-module is $M$ is of the form
$$
M \iso V \br{\otimes} \left\{C^!((\lambda))\right\}
$$ 
with a $C^!((\lambda))$-linear differential.  The differential gives $V$ the structure of $A_\infty$ $C((\lambda))$-comodule. 

\end{proof}

\section{The Yangian}
Let us return to field theory, and consider the twist of the Yangian deformation of the $N=1$ pure gauge theory on $\C^2$.  We have seen that we can quantize this theory, thus giving us a factorization algebra which we denote $\F$ on $\C^2$.  As before, we use coordinates $z$ and $w$ on $\C^2$, and our field theory is holomorphic in the $z$-direction and topological in the $w$-direction. 

Let $\pi : \C^2 \to \C_w$ be the projection onto the $w$-plane.   Factorization algebras, like sheaves, can be pushed forward: the factorization algebra $\pi_\ast \F$ on $\C_w$ assigns to an open subset $V \subset \C_w$ the complex $\F(\C_z \times V)$.  
 
We can restrict $\F$ to a factorization algebra any open subset of $\C^2$.  Thus, if $U_z \subset \C_z$ is an open subset in the $z$-plane, we have a factorization algebra $\pi_\ast (\F \mid_{U_z \times \C_w})$ on $\C_w$.  This assigns to an open subset $V \subset \C_w$ the cochain complex $\F(U_z \times V)$.   
\begin{lemma}
For all open subsets $U_z \subset \C_z$, $\pi_\ast (\F \mid_{U_z \times \C_w})$ is a locally constant factorization algebra on $\C_w$. 
\end{lemma}
\begin{proof}
Recall that the sheaf on $\C^2$ of dg Lie algebras describing our field theory is 
$$
\L(U) = \Omega^{0,\ast}(U, \g) \xto{\dpa{w}} \Omega^{0,\ast}(U,\g)[-1]. 
$$ 
Modulo $\hbar$,  we have
$$
\pi_\ast (\F \mid_{U_z \times \C_w}) (D(w_0,s)) \text{ mod } \hbar = C^\ast (\L(U_z \times D(w_0,s)))
$$
where $D(w_0,s)$ is the disc of radius $s$ centered at $w_0$ in the plane $\C_w$. 

Note also that, if $\op{Hol}(U_z)$ refers to the algebra of holomorphic functions on $U_z$, we have a quasi-isomorphism
$$
\L(U_z \times D(w_0,s)) \simeq \op{Hol}(U_z) \otimes \g. 
$$
It follows (with a small amount of functional analysis) that we have a quasi-isomorphism of commutative algebras
$$
\pi_\ast (\F\mid_{U_z\times \C_w})(D(w_0,s)) \text{ mod } \hbar \simeq C^\ast(\op{Hol}(U_z) \otimes \g)
$$
(where we of course use the appropriate completed cochains on the right hand side).  

It is clear from this expression that $\pi_\ast (\F \mid_{U_z \times \C_w})$ is locally constant modulo $\hbar$.  The fact that it is locally constant without reducing modulo $\hbar$ follows by a spectral sequence. 
\end{proof}

By the general constructions explained in subsection \ref{subsection_Koszul_E2}, we find that for any open subset $U_z \subset \C_z$, the factorization algebra $\pi_\ast (\F \mid_{U_z \times \C_w})$ on $\C_w$ gives us an $E_2$ algebra.  We will denote this $E_2$ algebra by $\F_{U_z}$.   The underlying cochain complex of this $E_2$ algebra is $\F(U_z \times D_w)$, where $D_w$ is any disc in the $w$-plane.   As the open subsets $U_z$ varies, we find that $\F_{U_z}$ defines a factorization algebra on $\C_z$ valued in $E_2$ algebras. 
 
\subsection{}
The main result of this paper is, roughly, that the Koszul dual Hopf algebra to the $E_2$ algebra $\F_{D_z}$ is the linear dual to the Yangian.  Recall \cite{ChaPre95} that the Yangian is a Hopf algebra deforming the universal enveloping algebra $U (\g[z])$.  We will use a completed version of the Yangian, which quantizes $U(\g[[z]])$.   

We see an immediate problem, however:  the Yangian we use is built from power series in $z$, whereas our $E_2$ algebra $\F_{D_z}$ is built from holomorphic functions in $z$, as we see from the fact that
$$\F_{D_z} \text{ mod } \hbar \simeq C^\ast( \op{Hol}(D_z) \otimes \g ).$$

To remedy this problem, we need to select the sub-$E_2$ algebra of $\F_{D_z}$ associated to the formal disc inside the disc $D_z$.  If $z_0 \in \C_z$, let $D(z_0,r)$ be the disc of radius $r$ around $z_0$.   We would like to define the $E_2$ algebra $\F_{z_0}$ associated to a formal disc to be something like $\lim_{r \to 0} \F_{D(z_0,r)}$.  This approach has some small technical problems, so we will select our subalgebra $\F_{z_0} \subset \F_{D(z_0,r)}$ in a different way.

Recall that, modulo $\hbar$, $\F_{D(z_0,r)}$ is quasi-isomorphic to $C^\ast(\op{Hol}(D(z_0,r) \otimes \g)$.  The $E_2$ algebra we associate to the formal disc must have the property that, modulo $\hbar$, it is $C^\ast(\g[[z-z_0]])$, where $\C[[z-z_0]]$ is the algebra of formal power series around $z_0 \in \C_z$.   Note that there's an $S^1$ action on $C^\ast(\op{Hol}(D(z_0,r) \otimes \g)$, coming from rotation on the disc $D(z_0,r)$; and  $C^\ast(\g[[z-z_0]])$ is the subcomplex of $C^\ast(\op{Hol}(D(z_0,r) \otimes \g)$ of those elements which are a finite sum of eigenvectors for this $S^1$ action.  

\begin{lemma}
There is an $S^1$ action on the factorization algebra $\F$ on $\C_z \times \C_w$, covering the $S^1$ action on $\C_z \times \C_w$ which rotates $z$.  Under this $S^1$ action, the parameter $\hbar$ has weight $1$. 
\end{lemma}
\begin{proof}
There's a natural action of $S^1$ on the dg Lie algebra 
$$
\L = \Omega^{0,\ast}(\C^2,\g) \xto{\dpa{w}} \Omega^{0,\ast}(\C^2,\g)[-1]
$$
describing our field theory.   On each copy of $\Omega^{0,\ast}(\C^2)$, this is just the natural $S^1$ action by pull-back on the Dolbeault complex. 

Under this action, the pairing on $\L$ has weight $1$.  The construction of the quantization presented in the appendix shows is compatible with this $S^1$ action; the fact that the pairing has weight $1$ means we must also give $\hbar$ weight $1$. 
\end{proof}

It follows we have an $S^1$ action on the $E_2$ algebra $\F_{D(z_0,r)}$,  arising from the $S^1$ action on $\C_z \times \C_w$ which rotates around $z_0$. 

We let 
$$\F^{k}_{D(z_0,r)} \subset \F_{D(z_0,r)}$$
be the $k$-eigenspace for the $S^1$ action, in other words, the set of elements where $\lambda \in S^1$ acts by multiplication by $\lambda^k$.   

Similarly, we let 
$$\F^{classical,k}_{D(z_0,r)} \subset \F_{D(z_0,r)} \op{mod} \hbar$$
be the $k$-eigenspace in the classical observables.

The following technical lemma is a straightforward consequence of the definition of observable presented in \cite{CosGwi11}. 
\begin{lemma}
The map $\F^{k}_{D(z_0,r)} \to \F^{classical, k}_{D(z_0,r)}$ is surjective. 
\end{lemma}

It follows from this lemma that, as a graded $\C[[\hbar]]$-module without differential, $\F^{k}_{D(z_0,r)}$ is isomorphic to $\F^{classical,k}_{D(z_0,r)}[[\hbar]]$.   If we didn't have this lemma we wouldn't know that $\F^{k}_{D(z_0,r)}$ was even non-empty. 

\begin{definition}
Let
$$
\F_{z_0} = \oplus_{k \in \Z} \F_{D(z_0,r)}^{k}.
$$
Note that $\F_{z_0}$ is a graded $E_2$ algebra over the graded ring $\C[[\hbar]]$, where $\hbar$ has weight $1$.
\end{definition}
\begin{remark} 
Recall that $\F_{D(z_0,r)}$ is a completed filtered $E_2$ algebra. The direct sum in this expression is taken in the sense of completed filtered cochain complexes, as explained in definition \ref{definition_filtered}.  We could equivalently write
$$
\F_{z_0} = \liminv_{n} \left( \oplus_{k \in \Z} \left( \F_{D(z_0,r)}^{k} / \hbar^{n} \F_{D(z_0,r)}^{k-n} \right)  \right).
$$
\end{remark}

Note that there is a quasi-isomorphisms of filtered commutative algebras
$$
\F_{z_0} \text{ mod } \hbar \simeq C^\ast (\g[[z-z_0]]) .
$$
Thus, $\F_{z_0}$ is an $E_2$ algebra quantizing $C^\ast(\g[[z-z_0]])$ as desired. 

\subsection{Augmenting $\F_{z_0}$}.
In order to apply Koszul duality, we need an augmented $E_2$ algebra.  Here we will construct the augmentation. 

Recall from section \ref{log_theory} that our theory can be defined on complex surfaces with a divisor, where the volume form has a quadratic pole along the divisor.  Consider our theory on $\mbb{P}^1_z \times \C_w$, with volume form $\d z \d w$ which has a quadratic pole along the divisor $\infty \times \C_w$.   

We have a factorization algebra $\F_{\mbb{P}^1 \times \C_w}$ on $\mbb{P}^1 \times \C_w$.  Pushing forward to $\C_w$ yields a locally constant factorization algebra on $\C_w$, which can be viewed as an $E_2$ algebra.  Let us call this $E_2$ algebra $\F_{\mbb{P}^1}$.
\begin{lemma}
The cohomology of $\F_{\mbb{P}^1}$ is $\C[[\hbar]]$.
\end{lemma}
\begin{proof}
Modulo $\hbar$, $\F_{\mbb{P}^1}$ is cochains of the Lie algebra $\mscr{L}$ which describes the theory. This Lie algebra is 
$$
\mscr{L} = \Omega^{0,\ast}(\mbb{P}^1, \Oo(-\infty)) \otimes \Omega^{\ast}(\C_w)) \otimes \g
$$
which has no cohomology. 
\end{proof}
Thus, there is a quasi-isomorphism of $E_2$ algebras $\F_{\mbb{P}^1} \simeq \C[[\hbar]]$.  The factorization structure gives a map of $E_2$ algebras
$$
\F_{z_0} \to \F_{\mbb{P}^1} \simeq \C[[\hbar]].
$$
This is the desired augmentation. 

\subsection{The Yangian}
Before I state the main theorem, I should first recall the definition of the Yangian. We use a completed version of the Yangian.
\begin{definition}
Let $\g$ be a simple Lie algebra, and let $\ip{-,-}_{\g}$ denote the invariant pairing on $\g$.   Make $\g[[z]]$ into a graded Lie algebra by giving $X z^k$ degree $k$, for $X \in \g$.  

The Yangian $Y(\g)$ is the unique topological Hopf algebra flat over the ring $\C[[\hbar]]$ with the following properties. 

\begin{enumerate}
\item $Y(\g)$ has a grading where the central parameter $\hbar$ has weight $1$.
\item There's an isomorphism of graded Hopf algebras
$$
Y(\g) \otimes_{\C[[\hbar]]} \C_{\hbar = 0} \iso U(g[[z]]). 
$$  
\item As is standard in the theory of quantum groups \cite{ChaPre95, Dri87}, any deformation of the Hopf algebra $U(\g[[z]])$ gives a Lie bialgebra structure on $\g[[z]]$.   The Lie bialgebra structure on $\g[[z]]$ which quantizes to the Yangian  has for coproduct the map
\begin{align*}
\g[[z]] & \to \g[[z_1]] \otimes \g[[z_2]] \\
X   & \mapsto \frac{ [X, c] }{z_1 - z_2 },
\end{align*}
where 
$$c \in \g \otimes \g \subset \g[[z_1]] \otimes \g[[z_2]]  $$ is the quadratic Casimir, dual to the chosen invariant pairing, and $[X,-]$ refers to the action in the tensor square of the adjoint representation. 
\end{enumerate}
\end{definition}
\begin{remark}
\begin{enumerate}
\item People normally consider the Yangian as a quantization of $U(\g[z])$.  This is a dense subalgebra of the Yangian we consider.   The grading makes it easy to pass back and forth between the completed Yangian we use and the uncompleted version; one is the sum of its graded pieces, and the other is the product. 
\item The fact that the Yangian is the unique Hopf algebra satisfying these conditions is proved in \cite{Dri87, ChaPre95}.
\item Lie bialgebras correspond to Manin triples, which are Lie algebras with an invariant pairing and a decomposition as a direct sum of isotropic sub-Lie algebras. The Lie bialgebra structure given above on $\g[[z]]$ corresponds to the Manin triple $(\g((z)), \g[z], z^{-1} \g[z^{-1}])$. with the pairing 
$$\ip{a f(z), b g(z) } = \ip{a,b}_{\g} \op{Res} f(z) g(z) \d z.$$
\end{enumerate}
\end{remark}

\subsection{}
The main theorem says that the Koszul dual Hopf algebra $\F^!_{D_z}$ is the $\C[[\hbar]]$-linear dual of the Yangian.  I need to say in what precise sense I mean linear dual.
\begin{definition}
The dual Yangian $Y^\ast(\g)$ is the topological Hopf algebra defined by
$$
Y^\ast(\g) = \Hom_{\C[[\hbar]]} ( Y(\g), \C[[\hbar]] ) .
$$
As a vector space, $Y^\ast(\g)$ is isomorphic to $\what{\op{Sym}}^\ast(\g^\vee[\partial_z] ) [[\hbar]]$.  

Let $\otimes_{\C[[\hbar]]}$ denote the completed projective tensor product of topological $\C[[\hbar]]$-modules.  Then, the product and coproduct for the dual Yangian
\begin{align*}
m : Y^\ast(\g) \otimes_{\C[[\hbar]]} Y^\ast(\g) & \mapsto Y^\ast(\g) \\
c : Y^\ast(\g) \to Y^\ast(\g) \otimes_{\C[[\hbar]]} Y^\ast(\g)
\end{align*}
are continuous $\C[[\hbar]]$-linear maps, and are dual to the product and coproduct of the Yangian. 
\end{definition}
\begin{lemma}
The dual Yangian is the unique graded topological Hopf algebra quantizing the Lie bialgebra $\g[[z]]^\vee = \g^\vee[\partial_z]$, with Lie bialgebra structure dual to that on $\g[[z]]$ discussed earlier.
\end{lemma}
\begin{proof}
To show this, we need to show how to recover the Yangian $Y(\g)$ from the dual Yangian $Y^\ast(\g)$.  $Y(\g)$ is simply the continuous $\C[[\hbar]]$- linear dual of $Y^\ast(\g)$.  Similarly, $Y(\g) \otimes_{\C[[\hbar]]} Y(\g)$ is the continuous $\C[[\hbar]]$-linear dual to $Y^\ast(\g) \otimes_{\C[[\hbar]]} Y^\ast(\g)$.  Thus, the Hopf algebra structure on $Y^\ast(\g)$ dualizes to one on $Y(\g)$, which as we have seen is unique. 
\end{proof}

\subsection{Proof of the main theorem}
Now I can state the main theorem.
\begin{theorem}
Let $\g$ be a simple Lie algebra.  Then, the cohomology of the Koszul dual Hopf algebra to $\F_{z_0}$ is the dual Yangian Hopf algebra $Y^\ast(\g)$. That is, there's an isomorphism of Hopf algebras
$$
Y^\ast(\g) \iso H^\ast( \F_{z_0}^{!} ) = H^\ast\left(  \C[[\hbar]] \otimes^{\mbb{L}}_{\F_{z_0} }\C[[\hbar]] \right).
$$ 
\label{thm:main}
\end{theorem}
\begin{remark}
\begin{enumerate}
\item Of course, the dual Yangian is concentrated in cohomological degree $0$; part of the statement is that the right hand side is also concentrated in degree $0$.  
\item The Yangian is the continuous linear dual to $Y^\ast(\g)$, and so the continuous linear dual to $H^\ast(\F^!_{z_0})$.   This continuous linear dual can be described as $\RHom_{\F}  \left( \C[[\hbar]], \C[[\hbar]] \right)$, as long as one takes care with things like filtrations when defining $\RHom$.  
\end{enumerate}
\end{remark}
\begin{proof}
To prove the theorem, we need to verify that the Hopf algebra $\F_{z_0}^{!}$ satisfies the conditions characterizing the dual Yangian. Since $\F_{z_0}$ is a graded $E_2$ algebra over the graded ring $\C[\hbar]$, it is clear that the Koszul dual $\F^{!}_{z_0}$ is a graded Hopf algebra over the same graded ring.  This gives us the first condition.  

Next, we need to check that there's a quasi-isomorphism 
$$
\F_{z_0}^{!} \otimes_{\C[\hbar]} \C \iso U(\g[[z]])^\vee. 
$$

We have a quasi-isomorphism of graded $E_2$ algebras
$$
\F_{z_0} \otimes_{\C[\hbar]} \C \iso C^\ast (\g[[z]]).  
$$
Koszul duality as explained in section \ref{e_2_hopf} now tells us that
$$
\F^!_{z_0} \text{ mod } \hbar = U(\g[[z]])^\vee
$$
as desired. 

The hard step in the proof is to verify that, to first order, our quantization corresponds to the desired Lie bialgebra structure.   This is proved in the following proposition.

\end{proof}  
\begin{proposition}
\label{proposition_bialgebra}
Let $Y'(\g)$ be the Hopf algebra which is the linear dual to the Koszul dual of Koszul dual to the augmented $E_2$ algebra $\Obs_{z_0}$. Thus, modulo $\hbar$, $Y'(\g)$ is $U(\g[[z]])$ with its commutative coalgebra structure.  Modulo $\hbar^2$, the cocommutator on $Y'(\g)$ gives a co-Lie bracket
$$
\delta : \g[[z]] \to \g[[z_1]] \otimes \g[[z_2]].
$$
Up to a non-zero constant, this is the co-Lie bracket defining the semi-classical structure of the Yangian.  
\end{proposition}
The proof of this will take the rest of this subsection.
\begin{remark}
Another proof of this proposition is presented in proposition \ref{proposition_bialgebra_alternative} at the end of the paper. 
\end{remark}

Let 
$$\Obs^{cl}_{z_0} = \F_{z_0} \op{mod} \hbar.$$ 
Note that $\Obs^{cl}_{z_0} = C^\ast(\g[[z]])$, so the Koszul dual of $\Obs^{cl}_{z_0}$ is $U(\g[[z]])^\vee$.   The general Koszul duality framework we have developed tells us that we have an isomorphism of filtered associative algebras
$$
HC_\ast(\Obs^{cl}_{z_0}) \simeq HC_\ast (U(\g[[z]]))^\vee,
$$
where $HC_\ast$ indicates Hochschild chains.  The product on the right hand side comes from the coproduct on $U(\g[[z]])$. 

Passing to the cohomology of the associated graded, we find an isomorphism
$$
HH_\ast ( \op{Gr} \Obs^{cl}_{z_0}) \simeq HH_\ast ( \Sym (\g[[z]]) )^\vee. 
$$
Both sides of this equation are endowed with a Poisson bracket describing the behavior modulo $\hbar^2$.  On the right hand side, the Poisson bracket comes from the coproduct on the Yangian modulo $\hbar^2$.  On the left hand side, the Poisson bracket comes from the quantization of our field theory. 

Note that non-triviality of this Poisson bracket does not contradict the fact that the product on 
$$HH_\ast (\Obs_{z_0}) = H^\ast (\Obs(z_0 \times \C^\times_w))$$
is commutative. We have a non-trivial Poisson bracket on the homology of the associated graded, but it disappears when we pass to the next term of the spectral sequence computing $HH_\ast(\Obs^{cl}_{z_0})$. Even so, the non-triviality of this Poisson bracket is enough to verify that the Lie co-bracket on $\g[[z]]$ coming from our quantization is non-trivial. 

We will show that these two Poisson brackets coincide, up to a non-zero constant which will be normalized later.  This will be enough to prove that the Koszul dual of $\Obs_{z_0}$ gives the correct Lie bialgebra structure, because, as I will now explain, the Lie cobracket on $\g[[z]]$ is determined by the Poisson bracket on $HH_\ast(\Sym \g[[z]])^\vee$. 

$Y(\g)$ is a filtered Hopf algebra over $\C[[\hbar]]$, with a filtration indexed by $\Z$. The cocommutator on $Y(\g)$, modulo $\hbar^2$, gives a co-Poisson bracket on the associated graded,
$$
\delta:  \Sym (\g[[z]]) \to \Sym(\g[[z]]) \otimes \Sym(\g[[z]]).
$$
Thus, if $c$ denotes the trivial coproduct on $\Sym \g[[z]]$, which makes $\g[[z]]$ primitive, $c + \hbar \delta$ is a coalgebra homomorphism modulo $\hbar^2$. 

The co-Poisson bracket $\delta$ is determined by its behavior on $\g[[z]]$, and there it is given by the Lie coalgebra structure for the Yangian. 

Since $\Sym (\g[[z]])$ is a commutative algebra, $HH_0(\Sym(\g[[z]])) = \Sym(\g[[z]])$. Thus, applying the functor $HH_0$ to the algebra homomorphism $c + \hbar \delta$ gives us back the same homomorphism $c + \hbar \delta$.   This shows that the map
$$
HH_0( \Sym (\g[[z]])) \to HH_0(\Sym (\g[[z]])) \otimes HH_0(\Sym(\g[[z]]))
$$
encodes the Lie coalgebra structure which is the semi-classical limit of the Yangian. 

Now, Koszul duality tells us that the $E_2$ algebra $\Obs_{z_0}$ is Koszul dual to some quantization of $U (\g[[z]])$.  The Lie cobracket describing the semi-classical behavior of this quantization is encoded in the Poisson bracket
$$
HH_0(\op{Gr} \Obs^{cl}_{z_0}) \to HH_0(\op{Gr} \Obs^{cl}_{z_0}) \otimes HH_0(\op{Gr} \Obs^{cl}_{z_0} ).
$$
We thus need to calculate this Poisson bracket.  

The Hochschild-Kostant-Rosenberg theorem tells us that 
$$
HH_\ast ( \op{Gr} C^\ast (\g[[z]])  ) = \what{\Sym} ( \g[[z,\delta]]^\vee[-1] )
$$
where $\delta$ is a parameter of cohomological degree $1$.  Thus, 
$$
HH_0 ( \op{Gr} C^\ast (\g[[z]] ) )  = \Sym^\ast  ( \g[[z]]^\vee ) ,
$$
as an algebra.   We are interested in the Poisson bracket on $HH_0$. 

Now, we can also identify
$$
HH_\ast ( \op{Gr} C^\ast (\g[[z]])  ) = H^\ast( \op{Gr} \Obs^{cl} (z_0 \times \C^\times_w  ) ) .
$$
At the cochain level, we have
$$
\Obs^{cl}(z_0 \times \C^\times_w) = C^\ast ( \Omega^\ast(\C^\times_w) [[z]]  \otimes \g ) .
$$
Thus,
$$
H^0 ( \op{Gr} \Obs^{cl}(z_0 \times \C^\times_w) ) = \Sym^\ast ( \g[[z]]^\vee \otimes H^1 (S^1) )  
$$
where $H^1(S^1)$ is situated in cohomological degree $0$.  

The space
$$\Obs(z_0 \times \C^\times_w )$$
of observables supported on $\C^\times_w$ (or, equivalently, on an annulus) has a homotopy associative product, coming from the operator product.  We will show how this induces a Poisson bracket on the cohomology of $\op{Gr} \Obs^{cl } (z_0 \times \C^\times_w)$. 

Let $A$ be any filtered associative dg algebra $A$ over $\C[\hbar]/\hbar^2$, which is commutative modulo $\hbar$, and where (as usual) $\hbar F^i A \subset F^{i+2} A$.  Let $A^{cl} = A \op{mod} \hbar$. Then, the cohomology of $\op{Gr} A^{cl}$ has a sequence of Poisson brackets, where the $n$'th bracket is defined if the  $n-1$'th bracket is zero.  The definition is as follows. 

If $\alpha \in H^\ast \op{Gr}^i A^{cl}$, $\beta \in H^\ast \op{Gr}^j A^{cl}$, we let $\til{\alpha}, \til{\beta}$ be lifts of $\alpha,\beta$ to elements of $F^i A$ and $F^j A$ respectively.  Note that $\til{\alpha}, \til{\beta}$ are not necessarily closed, they are only closed in $F^i A^{cl} / F^{i+1} A^{cl}$ and $F^j A^{cl} / F^{j+1} A^{cl}$.  

The commutator $[\til{\alpha}, \til{\beta} ]$ is an element of $F^{i+j} A$, which is closed modulo $F^{i+j+1} A$.  Since the cohomology of $A$ is commutative modulo $\hbar$, we can lift this commutator to an element of $\hbar F^{i+j-2} A$, which is closed modulo $\hbar F^{i+j-1}$.  The cohomology class of this commutator is the Poisson bracket
$$
\{\alpha,\beta\}_{-2} \in H^\ast \op{Gr}^{i+j-2} A^{cl}. 
$$
If this Poisson bracket vanishes, then (by adding on the appropriate homotopy) we can lift the commutator to a closed element of $\hbar F^{i+j-1}$, and so get a secondary bracket
$$
\{\alpha,\beta\}_{-1} \in H^\ast \op{Gr}^{i+j-1} A^{cl}. 
$$
And so on. 

We will apply this construction to the example of $\Obs(z_0 \times \C^\times_w)$. We have a commutator map of filtered complexes
$$
\Obs (z_0 \times \C^\times_w) \times \Obs (z_0 \times \C^\times_w) \times \to \Obs (z_0 \times \C^\times_w).
$$
Modulo $\hbar$, this is homotopically trivial, this giving a Poisson bracket map
$$
\Obs^{cl} (z_0 \times \C^\times_w) \times \Obs^{cl} (z_0 \times \C^\times_w) \times \to \Obs^{cl} (z_0 \times \C^\times_w).
$$
which is a map of filtered complexes with a shift: it maps $F^i$ on the left hand side to $F^{i-2}$ on the right hand side.  At the level of the associated graded, if this map is homotopically trivial, we get a secondary bracket as above, which maps $F^i$ to $F^{i-1}$. We need to compute this secondary bracket. 

To explain the answer, we need some notation for elements of $\Obs^{cl}(z_0 \times \C^\times_w)$. If $X \in \g$, define an observable 
$$X[k](r) \in \Obs^{cl} ( z_0 \times C^\times_w ) $$
to be the linear map
\begin{align*}
X[k](r) : \Omega^1 (\C^\times_w) \otimes \C[[z]] \otimes \g &\to \C\\
X[k](r) ( \alpha \otimes f(z) \otimes Y ) = \ip{X,Y}( \delta_z^k f )\mid_{z = 0}\int_{\abs{w} = r} \alpha.
\end{align*}
Note that $X[k](r)$ is a closed element of cohomological degree $0$ of $\op{Gr} \Obs^{cl } (z_0 \times \C^\times_w)$.  The cohomology class $[X[k](r)]$ of $X[k](r)$ is independent of $r$: we can thus call the cohomology class $[X[k]]$. 

\begin{proposition}
The Poisson bracket map
$$
\op{Gr}^i \Obs^{cl} (z_0 \times \C^\times_w) \times \op{Gr}^j \Obs^{cl} (z_0 \times \C^\times_w) \times \to \op{Gr}^{i+j-2}\Obs^{cl} (z_0 \times \C^\times_w).
$$
is homotopically trivial.  This allows us to define a secondary Poisson bracket
$$
\op{Gr}^i \Obs^{cl} (z_0 \times \C^\times_w) \times \op{Gr}^j \Obs^{cl} (z_0 \times \C^\times_w) \times \to \op{Gr}^{i+j-1}\Obs^{cl} (z_0 \times \C^\times_w).
$$
At the level of cohomology, this secondary Poisson bracket is given by the formula
$$
[ X[k], Y[l] ] = \alpha \frac{k! l!}{(k+l+1)!} [X,Y][k+l+1]
$$
where $\alpha$ is a certain non-zero constant. 
\end{proposition}
\begin{remark}
We won't calculate the precise value of $\alpha$ here.  We will show in lemma \ref{lemma_Rmatrix_expectation_value} that if we normalize the action functional of our theory to $\frac{1}{2 \pi i} \d z CS(A)$ we find the usual normalization of the Yangian.  This lemma will also allow us to give another proof, in proposition \ref{proposition_bialgebra_alternative}, that the Lie bialgebra structure we find is that corresponding to the Yangian.   
\end{remark}
\begin{remark}
We identify the dual of $\g[[z]]$ with $z^{-1} \g[z^{-1}]$ using the residue pairing and the invariant pairing on $\g$.  If we do this, then $X[k] \in \g[[z]]^\vee$ corresponds to $k^! X z^{-k-1}$. Thus, the Lie coalgebra structure on $\g[[z]]$ we are computing is, up to a constant, the natural bracket on $z^{-1} \g[z^{-1}]$.
\end{remark}
\begin{proof}
The existence of the augmentation map for the $E_2$ algebra of observables $\Obs_{z_0}$ implies that the bracket $\{-,-\}_{-w}$ vanishes.   The secondary bracket gives us a Lie bialgebra structure on $\g[[z]]$, compatible with the given Lie algebra structure.  $\C^\times$-invariance of our construction implies that the cobracket
maps 
$$
z^k \g \to \sum_{r + s = k - 1} z^r \g \otimes z^s \g. 
$$
A simple algebraic lemma shows that if $\g$ is simple, then there is exactly one Lie bialgebra structure on $\g[[z]]$ with these properties, where the cobracket
$$
z \g \to \g \otimes \g
$$
is dual to the bracket on $\g$, under the invariant pairing. 

To fix our Lie coalgebra structure, we need to calculate the Poisson bracket $\{X[0], Y[0]\}$.   (Note that the result will hold for an arbitrary $\g$, not just simple  $\g$; the reason is that the Poisson bracket of $\{X[k], Y[l]\}$ is expressed in terms of $[X,Y][k+l+1]$ by some universal function of $k$ and $l$, which we can determine by considering the simple case). Since we are only treating the case when $k = l = 0$, we will use the notation $X,Y$ instead of $X[0]$, $Y[0]$ for these observables. 

Let us now turn to the explicit calculation of $\{X, Y\}$. 
This proof is the only part of this paper where we rely on the details of the constructions in \cite{CosGwi11}: the calculation is in terms of Feynman diagrams. 

Let
$$
\E = \Omega^{0,\ast}(\C_z) \what{\otimes} \Omega^\ast(\C^\times_w) \otimes \g[1]
$$
be the space of fields of our theory.  In order to do Feynman diagram calculations, we need to choose a gauge fixing condition: we will use $\dbar^\ast_z + \d^\ast_w.$ The choose of gauge fixing condition allows us to define a propagator
$$
P \in \br{\E} \what{\otimes} \br{\E}
$$
(where $\br{\E}$ indicates the distributional completion: the propagator has singularities along the diagonal). 
The propagator can be expressed in terms of the heat kernel $K_t$.  

We let 
$$
K_t^{scalar} \in \cinfty(\C_z \times \C^\times_w) \what{\otimes} \cinfty(\C \times \C^\times_w)
$$
be the usual scalar heat kernel for the Laplacian.  Then, the heat kernel of interest is
$$
K_t = K_t^{scalar} \d (\zbar_1 - \zbar_2) \d (\bar{w}_1 - \bar{w}_2) \d (w_1 - w_2) \in \E \otimes \E. 
$$
The propagator can be defined in terms of the heat kernel by
$$
P = \int_{t = 0}^{\infty}  ( (\dbar_z ^\ast + \d^\ast_w) \otimes 1) K_t \d t. 
$$
We can write
$$
P = c \otimes P_0
$$
where $c \in \g \otimes \g$ is the quadratic Casimir (i.e.\ the inverse to the invariant pairing on $\g$), and
$$
P_0 \in \Omega^{0,\ast}(\C_z) \what{\otimes} \Omega^\ast(\C^\times_w)[1] .
$$

We are also interested in the propagator with an infrared cutoff,
$$
P(0,L) = \int_{t = 0}^{L}  ( (\dbar_z ^\ast + \d^\ast_w) \otimes 1) K_t \d t.
$$
 
We need to recall some details of the construction of factorization algebras in \cite{CosGwi11}.  There, we define an observable $O$ to be a family of functionals
$$
O[L] \in \Oo ( \E) [[\hbar]]
$$
on the space $\E$ of fields, depending on a parameter $L \in \R_{> 0}$.  There is a differential 
$$
\d_L = Q + \{I[L], - \}_L + \hbar \tr_L
$$
on the space $\Oo(\E)[[\hbar]]$, depending on $L$.  Here, $I[L]$ is the scale-$L$ effective interaction, $\tr_L$ is the scale-$L$ BV operator, $\{-,-\}_L$ is the scale-$L$ BV bracket, and $Q$ is simply the linear differential on the space $\E$.   

The functionals $O[L]$ and $O[L']$ must be related by the renormalization group flow, which is a linear isomorphism of $\Oo(\E)[[\hbar]]$ intertwining the differentials $\d_L$ and $\d_{L'}$. The differentials $\d_L$ make the space of observables $\Obs(\C^2)$ into a cochain complex.  If $U \subset \C^2$, we say that $O$ is in $\Obs(U)$ if $O[L]$ is supported on $U$ in a limit as $L \to 0$. (Details are presented in \cite{CosGwi11}; I'm being vague about the definition of $\Obs(U)$, as we will not need the detailed definition in our calculation).

Now let's start the computation of the Poisson bracket. Let us change notation a little bit, and coordinatize $\C^\times_w$ as $S^1 \times \R$. We let $S^1_\eps$ be the circle in $\C_z \times S^1 \times \R$ given by $0 \times S^1 \times \eps$.  For $\eps \in \R$, and $X \in \g$, let $X(\eps)$ be the classical observable which is the linear function on the space of fields sending a field $\phi$ to
$$
X(\eps)(\phi) = \oint_{S^1_\eps} \ip{X, \phi} .
$$
Evidently, $X$ only extracts the part of $\phi$ living in $\Omega^0(\C)\what{\otimes} \Omega^1(S^1) \what{\otimes} \Omega^0(\R) \otimes \g$. 

Let $\til{X}(\eps)$ be a lift of $X(\eps)$ to an observable defined modulo $\hbar^2$, in $\Obs(S^1_\eps)$.  Consider the commutator
\begin{equation*}
[\til{X}, \til{Y}]_{\eps,\delta} =  (\til{X}(\eps) \til{Y}(\delta) - \til{X}(\delta) \til{Y}(\eps) )  
\end{equation*}
for $\eps < \delta \in \R$.  Here juxtaposition indicates the product in the factorization algebra.
  
We need to compute this commutator observable explicitly. An observable is a family of functionals, one for each $L > 0$; we will fix some $L$, and write down a formula for the relevant terms on $[\til{X}, \til{Y}]_{r,s}[L]$. 

The only term that will be of interest to use will be the term which occurs at one loop, and which is a linear functional on the space of fields.  There are also classical terms, with $2$ or more inputs, but these will not be important: they will disappear as $r \to s$, and we will take this limit shortly. The constant one-loop term is zero, which guarantees that the bracket $\{-,-\}_{-2}$ described above vanishes.

We thus let $[\til{X}, \til{Y}]^{1,1}_{\eps,\delta}[L]$ be the term which occurs at one loop (i.e.\ the coefficient of $\hbar$) and which is linear as a function on the space of fields. We want to compute this linear function on the field $z Z$, where $Z \in \g$.  The following formula is easy calculation in the theory developed in \cite{CosGwi11}:
\begin{multline*}
[\til{X}, \til{Y}]^{1,1}_{r,s}[L] (z Z) \\
= \ip{[X,Y],Z}_{\g} \left\{ \oint_{\abs{w_0} = r, z_0 = 0} \oint_{\abs{w_2} = s, z_2 = 0} -  \oint_{\abs{w_0} = s, z_0 = 0} \oint_{\abs{w_2} = r, z_2 = 0}   \right\}   \int_{z_1,w_1 \in \C_z \times \C^\times_w} \\  P_0(0,L)((z_0,w_0), (z_1,w_1) ) \d z_1 z_1  P_0(0,L)((z_1,w_1), (z_2,w_2) ) .
\end{multline*}
This integral corresponds to the Feynman diagram with three vertices, trivalent and two univalent; two internal lines, and one external line.  The two univalent vertices are labelled by the operators $X(\eps)$ and $Y(0)$, supported at the circles at $\eps$ and $0$. These are both linear functions on the space of fields. The trivalent vertex is labelled by the Chern-Simons interaction term $\int \d z\tfrac{1}{6} \ip{A,[A,A]}$. The external line is labelled by the field $z Z$.  Hopefully this expression is familiar to physicists: we're basically computing an OPE in the standard way.

The cohomology class of the Poisson bracket we are computing is independent of $\eps$ and $\delta$.  Let us set $\delta = 0$.  We can thus compute that the cohomology class of this commutator can be described as follows.  Consider the operator product $X(\eps) Y(0)$.  This defines a family of observables depending on $\eps \in \R \setminus \{0\}$.  As we vary $\eps$ without crossing $0$, the cohomology class of this observable does not change.  Thus, the commutator can be computed as the ``jump'' of this observable $X(\eps) Y(0)$ as $\eps$ crosses $0$. 

Classical terms do not have a jump, but the one-loop quantum term described above does.    

In the $\eps \to 0$ limit, sending $L \to \infty$ does not change the answer.  As, the jump we see when $\eps$ crosses zero only depends on the singularities in the integrand.   Sending $L \to \infty$ does not change the singularities. 

Next, observe that the integral above is expressed using the propagator for the theory on $\C_z \times \C^\times_w$.  Here, the space of fields is 
$$\Omega^{0,\ast}(\C_z) \what{\otimes} \Omega^\ast(\R) \what{\otimes} \Omega^\ast(S^1)  \otimes \g[1].$$  
Let us write 
$$
\Omega^\ast(S^1) = \mc{H}(S^1) \oplus \mc{H}(S^1)^{\perp}
$$
where $\mc{H}(S^1)$ is the space of harmonic forms.  If we replace the propagator $P_0(0,\infty)$ in the above expression to the projection of the propagator onto 
$$\left\{\Omega^{0,\ast}(\C_z) \what{\otimes} \Omega^\ast(\R) \otimes \mc{H}(S^1)\right\}^{\otimes 2}$$ 
we get the same answer.  This is because the observables $X(\eps)$ and $Y(0)$ we put at $\eps$ and $0$ take value zero on fields in $\mc{H}(S^1)^{\perp}$.  

This projection of the propagator is the propagator for a three-dimensional theory, on $\C_z \times \R_w$. The fields of this $3$-dimensional theory, in the $BV$ formalism, are $\Omega^{0,\ast}(\C) \what{\otimes} \Omega^\ast(\R) \otimes \g[\delta][1]$ where $\delta$ corresponds to the class $\ d \theta / 2 \pi$ in $H^1(S^1)$. The action is the usual Chern-Simons type action. 

In traditional terms, the fields of this $3$-dimensional theory are a partial connection
$$
A = A_w \d w + A_z \d \zbar
$$
and 
$$
B \in \eps \Omega^0(\C_z \times \R_w). 
$$
The action functional is
$$
\int \d z \ip{B, F(A)}_{\g}
$$
and the Lie algebra of the gauge group is, as usual, $\Omega^0(\C_z \times \R_w) \otimes \g$. 

For $\eps \in \R$, let $X(\eps)$ be the observable for this field theory which sends the field $B$ to 
$$X(\eps)(B) = \ip{X,B}_{g}(z = 0, w = \eps).$$
Let $P$ denotes the propagator of this $3$-dimensional theory. Let us write $P = c \otimes P_0$ where $c \in \g \otimes \g$ is the inverse to the invariant pairing on $\g$. We are computing
$$
\int_{z,w}   P_0 (0,(z_1,w_1)) z P_0((z_1,w_1), (0,\eps) ).
$$
We want to show that, as a function of $\eps$, this is a Heaviside step function: it takes value $1$ if $\eps > 0$, and value $0$ if $\eps < 0$.

The propagator for the $3$-dimensional theory is defined using a gauge fixing condition
$$\GF=2 \dbar^\ast_z + \d^\ast_w.$$ 
(The factor of $2$ is a convenient normalization, as $[\dbar, 2 \dbar^\ast]$ is the usual Laplacian). 

Let
$$
\tr = - \dpa{z}\dpa{\zbar} - \dpa{w}^2 
$$
be the Laplacian.  Then, the propagator above is the kernel for the operator
$$
\GF \tr^{-1}
$$
which acts on the space $\Omega^{0,\ast}(\C)\what{\otimes} \Omega^\ast(\R)$. The integral above we need to compute is the kernel for the composition
$$
\GF \tr^{-1} z \GF \tr^{-1}
$$
Note that 
$$
[z,\GF] = 2 [z, \dbar^\ast_z] = -  \dpa{\d \zbar}.
$$
Since $(\GF)^2 = 0$, we have
$$
\GF \tr^{-1} z \GF \tr^{-1} = -  \GF \dpa{\d \zbar} \tr^{-2}.
$$
Next,
$$
-\GF \dpa{\d \zbar} = \dpa{\d \zbar} \d^\ast_w. 
$$

Further, the kernel for the operator $\tr^{-2}$ is 
$$\alpha r \d \zbar \d w,$$
where $r^2 = \abs{z}^2 + w^2$ and $\alpha$ is a certain non-zero constant.

It follows that the kernel for $\GF \tr^{-1} z \GF \tr^{-1}$ is 
$$
\alpha \frac{\d}{\d w} r = \alpha r^{-1} w.
$$

Ultimately, we are setting $z = 0$, so we find a Heaviside step function with value $\alpha^{-1}$ when $w > 0$ and $-\alpha^{-1}$ when $w < 0$.  Since the commutator we are computing is expressed in terms of the jump of this expression as we cross the origin, we are done. 
\end{proof}

\subsection{}

We need a categorical corollary of this theorem (which follow from the results of section \ref{Koszul_duality_categorical}). 
\begin{corollary}
\label{cor:Yangian_equiv_categories}
There are quasi-equivalences of dg monoidal categories
$$
\op{Perf}(\F_{z_0}) \simeq \op{Fin}(Y^\ast(\g)-\op{comod}) \simeq \op{Fin}(Y(\g)-\op{mod}).
$$
Here, $\op{Perf}$ refers to the category of perfect modules, and $\op{Fin}$ refers to categories of modules of finite rank.  (For precise definitions, see section \ref{Koszul_duality_categorical}: because $Y(\g)$ is a topological algebra, some care is needed). 
\end{corollary}
\begin{proof}
This follows immediately from the results proved in section \ref{Koszul_duality_categorical}. 
\end{proof}

\section{Wilson loops}
In this section, I will show how finite-dimensional representations of the Yangian give us Wilson loop operators for our twisted deformed supersymmetric gauge theory. 

We will consider the theory compactified on a circle.   We have seen that our theory makes sense on any Calabi-Yau surface $X$ equipped with a holomorphic volume form and a holomorphic $1$-form.  We will work with $\C_z \times \C^\times_w$, with volume form $w^{-1} \d z \d w$ and holomorphic one-form $\d z$.  With these choices, the theory is described by the dg Lie algebra
$$
\Omega^{0,\ast}(\C_z \times \C^\times_w, \g) \xto{w \partial_w} \Omega^{0,\ast}(\C_z \times \C^\times_w, \g) .
$$
As before, solutions to the equations of motion are holomorphic bundles on $\C_z \times \C^\times_w$ with a holomorphic (and therefore flat) connection along the fibers of the projection $\C_z \times \C^\times_w \to \C_z$.

The observables of this theory are given by a factorization algebra $\F$ on $\C_z \times \C^\times_w$.  We are interested in Wilson loop operators supported on the circles $z_0 \times S^1_R$, where $S^1_R = \{\abs{w} = R\}$.  Such operators will be elements of $\F(U)$ for any $U \subset \C_z \times \C^\times_w$ containing $z_0 \times R$. 

Let $D(z_0,r)$ be a disc around $z_0 \in \C_z$, and $A(r_1,r_2)$ be the annulus $r_1 < \abs{w} < \abs{r_2}$.   The cochain complex $\F(D(z_0,r)  \times A(r_1,r_2))$ has an $S^1$ action given, as before, by rotating in the $z$-plane with center $z_0$.  We let $\F^{k}((D(z_0,r) \times A(r_1,r_2 ))$ denote the $k$-eigenspace for this $S^1$ action.  Let
$$
\F(z_0 \times A(r_1,r_2) ) = \oplus_k  \F^{k}((D(z_0,r) \times A(r_1,r_2 )).
$$
As before, we should think of $\F(z_0 \times A(r_1,r_2))$ as being the $r \to 0$ limit of $\F((D(z_0,r) \times A(r_1,r_2))$.  

Note that the fact our factorization algebra $\F$ is locally constant in the $w$-direction means that, if $A(r_1,r_2) \subset A(r_1', r_2')$, the map
$$
\F(z_0 \times A(r_1,r_2) )  \into \F(z_0 \times A(r'_1,r'_2) ) 
$$
is a quasi-isomorphism.  Thus, we will just use the notation $\F(z_0 \times S^1)$ for this cochain complex. If we want to emphasize the radius of the circle $S^1_R \subset \C^\times_w$, we will use the notation $\F(z_0 \times S^1_R)$  to denote any $\F(z_0 \times A(r_1,r_2))$ where $r_1 < R < r_2$. 
\begin{lemma}
The cochain complex $\F(z_0 \times S^1)$ has a natural structure of $E_1$ (or homotopy associative) algebra, arising from the operator product of loop operators in the $w$-plane.  
\end{lemma}
\begin{proof}
If $r_1 < r_2 < s_1 < s_2$, the structure maps of a factorization algebra give us maps
$$
\F(z_0 \times A(r_1,r_2)) \otimes \F(z_0 \times A(s_1,s_2)) \to \F(z_0 \times A(r_1, s_2)). 
$$
All of the complexes appearing here are quasi-isomorphic to $\F(z_0 \times S^1)$, so we get a map defined up to homotopy
$$
\F(z_0 \times S^1 ) \otimes \F(z_0 \times S^1) \to \F(z_0 \times S^1). 
$$
It follows from Lurie's results (which we already used in the $E_2$ case) that this product is associative up to coherent homotopy.  
\end{proof}

\begin{proposition}
There is an isomorphism of associative algebras
$$
H^\ast ( \F(z_0 \times S^1) ) \iso HH_\ast( \F_{z_0})
$$
between the associative algebra of loop operators and the Hochschild homology of the $E_2$ algebra $\F_{z_0}$. 
\label{proposition_hochschild_factorization}
\end{proposition}
\begin{remark}
I should explain why Hochschild homology of an $E_2$ algebra has the structure of associative algebra.  One way to see this is that, by a classic theorem of Dunn, $E_2$ algebra is an $E_1$ algebra in the category of $E_1$ algebras.  Hochschild homology is a symmetric monoidal functor from $E_1$ algebras to cochain complexes. Thus, it takes $E_1$ algebras in $E_1$ algebras to $E_1$ algebras in cochain complexes.  
\end{remark}
\begin{proof}
First, we need to write down a map
$$
HH_\ast( \F_{z_0} ) \to H^\ast (\F(z_0 \times S^1)). 
$$
The map we wants is constructed using a result of Lurie.   He shows that the Hochschild chain complex has a natural universal property expressed in terms of factorization algebras.   Let $\mc{G}$ be a locally-constant prefactorization algebra on $S^1$.  By restricting $\mc{G}$ to an interval around $1 \in S^1$, we see that $\mc{G}$ determines an associative algebra $A_{\mc{G}}$.  Monodromy around $S^1$ is an automorphism of this associative algebra; let us assume that this automorphism is trivial.  This happens in our example. 

In this situation, Lurie shows that there is a natural map from the Hochschild chain complex $HC_\ast(A_{\mc{G}})$ to $\mc{G}(S^1)$. The complex $\F(z_0 \times S^1)$ is the value on the annulus of a locally-constant prefactorization algebra on $\C^\times$. By restricting this to a locally-constant prefactorization algebra on $S^1$, we see that there's a natural map 
$$
HC(\F_{z_0} ) \to\F(z_0 \times S^1).
$$
We need to verify this is a quasi-isomorphism.  It's possible to show that this follows from the ``descent'' or ``gluing'' axiom of \cite{CosGwi11} which distinguishes a factorization algebra from a prefactorization algebra.   We will take a different approach, and observe that it suffices to check that this map is a quasi-isomorphism modulo $\hbar$.

Recall that $\F_{z_0}$ is quasi-isomorphic to the commutative dga $C^\ast(\g[[z]])$. The Hochschild-Kostant-Rosenberg theorem tells us that 
$$
HC_\ast( C^\ast(\g[[z]] ) \simeq C^\ast(\g[[z,\eps]])
$$
where $\eps$ is a parameter of cohomological degree $1$.  (Recall that these constructions are taking place in the category of cochain complexes with a complete decreasing filtration, so that the Hochschild chain complex is also completed.   This quasi-isomorphism does not hold otherwise). 

Now, modulo $\hbar$, $\F(z_0 \times S^1)$ is the the direct sum of the $S^1$-eigenspaces of $C^\ast(\L(D(z_0,r) \times A_w)$, where for $U \subset \C_z \times \C^\times_w$, $\L(U)$ is the dg Lie algebra
$$
\L(U) = \Omega^{0,\ast} (U, \g) \xto{w \partial_w} \Omega^{0,\ast}(U,\g)[-1]. 
$$
Note that $\L(D(z_0,r) \times A_w)$ is quasi-isomorphic to $\g \otimes \op{Hol} ( D(z_0,r)) [\eps]$. Thus, $\F(z_0 \times S^1)$ is also quasi-isomorphic to $C^\ast(\g[[z,\eps]])$; it is easy to verify that the map $HC_\ast(C^\ast(\g[[z]] )\to \F(z_0 \times S^1)$ corresponds to the identity on $C^\ast(\g[[z,\eps]])$. 
\end{proof}

\subsection{}
Now, we can state a precise relationship between the associative algebra $\F(z_0 \times S^1)$ of loop operators, and the Yangian $Y(\g)$.  The Yangian is a Hopf algebra; thus, the coproduct $Y(\g) \to Y(\g) \otimes Y(\g)$ is a map of associative algebras.  It follows from the functoriality of Hochschild homology that $HH_\ast(Y(\g))$ is a co-associative co-algebra.  we can thus form the linear dual $HH_\ast(Y(\g))^\vee$, and get an associative algebra. 
\begin{theorem}
There is an isomorphism of associative algebras
$$
H^\ast (\F(z_0 \times S^1)) \iso HH_\ast(Y(\g))^\vee. 
$$
\label{theorem_hochschild_yangian}
\end{theorem}
\begin{proof}
This follows by putting together some facts we've already proved with general results about Koszul duality.   

We need the following general result, which is lemma \ref{lemma_coHochschild}. Let $A$ be an augmented associative algebra, and let $A^!$ denote the Koszul dual coassociative algebra to $A$.  We let $\op{Co}HH(A^!)$ denote the co-Hochschild homology of $A^!$: defined just like Hochschild homology of associative algebras, but with all arrows reversed.  Then, there is a natural isomorphism
$$
\op{Co}HH(A^!) \simeq HH(A). 
$$

This implies that the co-Hochschild homology of the dual Yangian $Y^\ast(\g)$ is isomorphic to the Hochschild homology of the $E_2$ algebra $\F_{z_{0}}$, which in turn is isomorphic to the cohomology of the algebra $\F(z_0 \times S^1)$.  The co-Hochschild homology of $Y^\ast(\g)$ is the dual to the Hochschild homology of $Y(\g)$, giving the desired result. 
\end{proof}
\begin{lemma}
There is an isomorphism of associative algebras
$$
H^\ast (\F(z_0 \times S^1)) \iso HH_\ast(\op{Fin}(Y(\g))
$$
where $\op{Fin}(Y(\g))$ is the monoidal dg category of finite-rank $A_\infty$- modules over the topological Hopf algebra $Y(\g)$. 
\end{lemma}
\begin{proof}
There is quasi-equivalence of monoidal dg categories
$$
\op{Fin}(Y(\g)) \simeq \op{Fin}(Y^\ast(\g)-\op{comod})
$$
(this is corollary \ref{corollary_dual_algebra_finite_equivalence}).  Also, there is a quasi-equivalence of monoidal dg categories
$$
\op{Perf}(\op{Fin}(\F_{z_0})  \simeq \op{Fin}(\F^!_{z_0}-\op{comod})
$$
(which is theorem \ref{theorem_equivalence_categories}).  Finally, for any $E_2$ algebra $A$, there is a natural isomorphism of $E_1$ algebras 
$$
HC_\ast(A) \simeq HC_\ast(\op{Perf}(A)) 
$$
(this follows from proposition \ref{prop:HH_iso_perfect}). 
\end{proof}
\subsection{}
This result may seem a little abstract, but it has a concrete corollary: it allows us to associate Wilson loop operators to finite-dimensional representations of the Yangian.

Let $A$ be any algebra, and let $V$ be a finite-dimensional representation of $A$.   Then, taking trace in $V$ defines a map $\op{Tr}_V  : A \to \C$.  Evidently, this kills the commutator $[A,A]$, and so factors through a map $A / [A,A] \to \C$.   Now, $A/[A,A]$ is $HH_0(A)$; thus, every finite-dimensional representation $V$ of $A$ gives an element of the dual to the Hochschild homology of $A$.  This is the character $\chi(V) \in HH_0(A)^\vee$ of $V$.  We can also think of the character $\chi(V) : HH_0(A) \to \C$ as arising by applying the Hochschild homology functor to the algebra homomorphism $A \to \op{End}(V)$, and observing that $HH(\op{End}(V))$ is canonically isomorphic to $\C$. 

We can apply this to our situation as follows.
\begin{definition}
Let $V$ be a finite-dimensional representation of the Yangian $Y(\g)$.  Define the Wilson operator associated to $V$ by
$$
\chi(V) \in HH_0(Y(\g))^{\vee} = H^0 (\F(z_0 \times S^1))
$$
to be the character of $V$. 
\end{definition}
The operator product makes $H^0(\F(z_0 \times S^1))$ into an associative algebra. 
\begin{lemma}
Let $V, W$ be finite-dimensional representations of $Y(\g)$.  Since $Y(\g)$ is a Hopf algebra, we can form the tensor product $V \otimes W$.  Then,
$$
\chi(V \otimes W) = \chi(V) \chi(W)
$$
where on the right we are using the operator product on $H^0 ( F(z_0 \times S^1))$.  
\end{lemma}
\begin{proof}
The tensor product on $Y(\g)$-modules arises from the coproduct 
$$
c : Y(\g) \to Y(\g) \otimes Y(\g).
$$
The coproduct $c$ is an algebra homomorphism, and so induces a map
$$
HH_0(Y(\g)) \to HH_0(Y(\g)) \otimes HH_0(Y(\g)).
$$
This is the map dual to the associative product on $H^0(F(z_0 \times S^1)) = HH_0(Y(\g))^\vee$.  

The definition of the tensor product of $Y(\g)$-modules means that the following diagram commutes:
$$
\xymatrix{
Y(\g) \ar[r] \ar[d] &  Y(\g) \otimes Y(\g) \ar[d] \\
\op{End}(V\otimes W ) \ar[r] & \op{End}(V) \otimes \op{End}(W)
}
$$
Here, the left hand vertical arrow is the map defining the action of  $Y(\g)$ on $V \otimes W$.  

Applying the Hochschild homology functor to this diagram of algebras we see that the following diagram also commutes:
$$
\xymatrix{
HH(Y(\g)) \ar[r] \ar[d]^{\chi(V \otimes W)} & HH( Y(\g)) \otimes HH( Y(\g) ) \ar[d]^{\chi(V) \otimes \chi(W)} \\
\C = HH(\op{End}(V \otimes W)) \ar[r] & \C = HH(\op{End}(V)) \otimes HH(\op{End}(W)) 
}
$$
Commutativity of this diagram expresses the fact that $\chi(V \otimes W) = \chi(V) \chi(W)$, as desired. 
\end{proof}
\begin{remark}
We can also view the character $\chi(V)$ as follows.  We have the following chain of quasi-equivalences of monoidal dg categories
$$
\op{Perf}(\F_{z_0}) \simeq \op{Fin}(\F_{z_0}^!-\op{comod}) = \op{Fin}(Y^\ast(\g)-\op{comod}) \simeq \op{Fin}(Y(\g)-\op{mod})).
$$
Applying Hochschild homology, we find an isomorphism of associative algebras
$$
HH_\ast(\op{Perf}(\F_{z_0}) \iso HH_\ast(\op{Fin}(Y(\g)-\op{mod}) ).
$$
We also have an isomorphism of associative algebras
$$
HH_\ast(\F_{z_0}) \iso HH_\ast(\op{Perf}(\F_{z_0}) ).
$$

If $V \in \op{Fin}(Y(\g)-\op{mod})$, then $\chi(V)$ is the element of $HH_0(\op{Fin}(Y(\g)-\op{mod})$ corresponding to the identity morphism of $V$.  It is immediate from this definition that $\chi(V \otimes W) =\chi(V)\chi(W)$.  This definition of the character is equivalent to the one described above. 
\end{remark}

\section{Compactification on a torus and integrability}

Let $E$ be an elliptic curve, and let us consider our twisted deformed $N=1$ gauge theory on $\C_z \times E$, where we use the holomorphic $1$-form $\d z$ to give the deformation.  

We can consider the cochain complex $\Obs(D(z_0,r) \times E)$ of observables on a neighborhood of the torus $z_0 \times E$.  This has an $S^1$ action, given by rotating in the $z$-plane with center $z_0$.  We let
$$
\Obs(z_0 \times E) = \oplus_{k} \Obs^{k} (D(z_0,r) \times E) 
$$
be the direct sum of the $S^1$-eigenspaces for this $S^1$ action.     We can compute the cohomology of $\Obs(z_0 \times E)$ in terms of the Yangian $Y(\g)$.   

We have seen (corollary \ref{cor:Yangian_equiv_categories}) that there is a quasi-equivalence of monoidal dg categories
$$
\op{Perf}(\F_{z_0}) \simeq \op{Fin}(Y(\g))
$$
(where $\op{Perf}$ is the category of perfect modules, and $\op{Fin}(Y(\g))$ is a certain category of continuous modules over the topological algebra $Y(\g)$. For details, see section \ref{Koszul_duality_categorical}). 

For any monoidal dg category $\mc{D}$, the Hochschild chain complex $HC_\ast(\mc{D})$ is a homotopy associative algebra. Thus, we can form the iterated Hochschild homology
$$
HH_\ast^{(2)} (\mc{D}) = HH_\ast (HC_\ast( \mc{D})).
$$
\begin{theorem}
There is an isomorphism of graded $\C[[\hbar]]$-modules 
$$
H^\ast ( \Obs(z_0 \times E)) \iso HH^{(2)}_\ast ( \op{Fin}(Y(\g)).
$$
\end{theorem}
\begin{proof}
 First, following the reasoning in proposition \ref{proposition_hochschild_factorization}, we observe that we can write 
$$
H^\ast (\Obs (z_0 \times E)) = HH_\ast ( HC_\ast ( \Obs_{z_0} ) ) 
$$
That is, $H^\ast (\Obs (z_0 \times E))$ is the twice-iterated Hochschild homology of the $E_2$ algebra $\Obs_{z_0}$.  Then, we use proposition \ref{prop:HH_iso_perfect}, which states that the natural map
$$
HC_\ast(\F_{z_0}) \to HC_\ast(\op{Perf}(\F_{z_0})
$$
is a quasi-isomorphism. Since this map is a map of homotopy associative algebras, it follows that we have an isomorphism
$$
HH_\ast(HC_\ast(\F_{z_0})) \simeq HH^{(2)} (\op{Perf}(\F_{z_0}).
$$
Applying the quasi-equivalence of monoidal dg categories
$$
\op{Perf}(\F_{z_0}) \simeq \op{Fin}(Y(\g))
$$
gives the desired result. 
\end{proof}
\subsection{}
Thus, we have an implicit description of $H^\ast (\Obs (z_0 \times E))$ in terms of the Yangian.   The operator product in the $z$-direction gives these complexes a structure very like that of a vertex algebra.   (Later we will see how this operator product is encoded in the $R$-matrix). 

The aim of this section is to construct a maximal commutative subalgebra of $H^0 (\Obs(z_0 \times E))$. 

Let 
$$
C(\g) = H^0 ( \Obs(z_0 \times S^1)). 
$$
Let us choose generators for $\pi_1(E)$, and let $S^1_a \subset E$ be the circle corresponding to one of the generators.  Then, we get a natural map
$$
C(\g) = H^0 ( \Obs(z_0 \times S^1)) \to H^0 (\Obs( z_0 \times E)).
$$
The main theorem of this section is the following.
\begin{theorem}
The operator product between any elements $\alpha,\beta \in C(\g) \subset H^0 ( \Obs(z_0 \times E))$ is non-singular.  Thus, $C(\g)$ is a commutative subalgebra of the holomorphic vertex algebra $H^0 ( \Obs(z_0 \times E))$.

Further, if we work in type $A$ (so $\g = \sl_n$), then $C(\sl_n)$ is a maximal commutative subalgebra. 
\end{theorem}
\begin{remark}
Probably an expert on the Yangian will see how to generalize the proof given below to any semi-simple $\g$.  The main obstruction is showing that $C(\g)$ is ``large enough''; that is, that any element of the classical analog of $C(\g)$ quantizes.  This would follow from the statement that the natural map $U(\g) \to Y(\g)$ induces an injective map $HH_0(U(\g)) \to HH_0(Y(\g))$.   For type $A$, we know this is true because of the existence of the evaluation homomorphism.
\end{remark}

The proof will be in several steps.  The first is the following easy lemma.
\begin{lemma}
$C(\g)$ is a commutative subalgebra of the holomorphic vertex algebra $H^0 ( \Obs(z_0 \times E))$.  
\end{lemma}
\begin{proof}
The language of factorization algebras allows one to make the proof sketched in the introduction rigorous.    Let $\alpha,\beta \in C(\g)$.   Let $\alpha(z_0),\beta(z_0) \in H^0 (\Obs(z_0 \times E))$ refer to the observables obtained by placing $\alpha$ and $\beta$ at $z_0$.  Then, because we have a holomorphic vertex algebra (i.e.\ a holomorphically translation invariant factorization algebra on $\C$), the operator product
$$
\alpha(z_0) \cdot \beta(0) \in H^0 (\Obs ( \C \times E)) 
$$
is a holomorphic function of $z_0 \in \C^\times$.   We need to verify that it extends over the locus where $z_0 = 0$.  

Let $S^1_{a'}$ be a circle obtained by displacing $S^1_a$ a small amount; so $S^1_{a}$ and $S^1_a$ don't intersect.  Since our theory is topological in the $w$-direction, the space of operators on $z_0 \times S^1_{a}$ and on $z_0 \times S^1_{a'}$ are the same, so we may as well view $\beta$ as an operator on $0 \times S^1_{a'}$. 

Then, the observables $\alpha(z_0)$ and $\beta(0)$ have disjoint support for all values of $z_0$.  The axioms of a factorization algebra on $\C \times E$ imply that the operator product $\alpha (z_0) \cdot \beta(0)$ is well-defined for all values of $z_0 \in \C$, as desired. 
\end{proof}

\subsection{}
The remainder of the proof involves verifying that $C(\sl_n)$ is maximal.  Along the way, we will give an explicit description of $C(\sl_n)$.

Let us start by computing the classical analog of $C(\g)$ for any $\g$.  Let
$$
C^{cl}(\g) = H^0 ( \Obs^{cl} (z_0\times S^1))
$$
where $\Obs^{cl}$ refers to the factorization algebra obtained by reducing $\Obs$ modulo $\hbar$.  Note that $C^{cl}(\g)$ is the Hochschild homology of the (filtered) algebra $C^\ast(\g[[z]])$.    It follows easily that
$$
C^{cl}(\g_n) = \Oo(\g[[z]]) ^{\g[[z]]} 
$$
is the algebra of adjoint-invariant formal power series on $\g[[z]]$. 

We want to describe explicitly $C^{cl}(\g)$ as a subalgebra of $\Oo(\g[[z]])$. 

The following basic technical result is proved in \cite{BeiDriHitchin}, page 23. 
\begin{proposition}
Let $G$ be a semi-simple algebraic group. Then, 
the algebra of $G[[z]]$-invariant polynomials on $\g[[z]]$ is naturally isomorphic to the algebra of polynomial functions on the pro-scheme $
\left( \mf{h} / W \right) [[z]] $ of maps from the formal disc to $\mf{h}/W$.  
\end{proposition}
\begin{remark}
$\mf{h}/W$ refers to the scheme quotient, not the stack quotient; thus $\mf{h}/W$ is the spectrum of the ring of invariant polynomials on $\g$. 
\end{remark}

Next, let us consider this case of two commuting elements of $\g$.
\begin{proposition}
Let $X \subset \g \oplus \g$ be the scheme of pairs of commuting elements.  Let $Y = (\mf{h} \oplus \mf{h}) / W$.  Then, the natural $G[[z]]$-equivariant map $X[[z]] \to Y[[z]]$ induces an isomorphism
$$
\Oo ( X[[z]])^{G[[z]]} \iso \Oo (Y[[z]] ).
$$
\end{proposition}
\begin{proof}
Let $X^{rss} \subset X$ be the locus where \emph{either} $A$ or $B$ is a regular semi-simple element.  Note that if $A$ is regular semi-simple, then $B$ is necessarily semi-simple.  Let $Y^{rss} \subset Y$ be the image of $X^{rss}$.  Thus, $Y^{rss}$ is the image of open subscheme 
$$\mf{h}^{rss} \times \mf{h} \cup \mf{h} \times \mf{h}^{rss} \subset \mf{h} \oplus \mf{h}$$
under the quotient map.

Note that the complement of $Y^{rss} \subset Y$ is of codimension $\ge 2$. Thus,
$$
\Oo ( Y^{rss}[[z]] ) = \Oo ( Y[[z]] ).
$$
According to \cite{Ric79}, $X$ is irreducible; therefore, the map
$$
\Oo (X[[z]]) \to \Oo( X^{rss}[[z]])
$$
is injective.   

Further, the map $X^{rss} \to Y^{rss}$ is smooth, and $G$ acts transitively on the fibers.  It follows that the map
$$
\Oo(Y^{rss}[[z]] ) \to \Oo (X^{rss} [[z]] )^{G[[z]]}
$$
is an isomorphism.   Therefore,
$$
\Oo (X^{rss} [[z]] )^{G[[z]]} = \Oo ( Y[[z]] ) .
$$
Since $\Oo (X[[z]])^{G[[z]]}$ is a subalgebra of the left hand side, and $\Oo(X[[z]])^{G[[z]]}$ contains $\Oo(Y[[z]])$ (as there is a natural $G[[z]]$-equivariant map $X[[z]] \to Y[[z]]$), the result follows. 

\end{proof}

\subsection{}
So far, we have shown that 
$$
H^0 (\Obs^{cl}( 0 \times E)) = \Oo ( ((\mf{h} \oplus \mf{h}) / W) [[z]] ) ;
$$
and that
$$
C^{cl}(\g) = H^0  ( \Obs^{cl} (0 \times E)) = \Oo ((\mf{h}/W)[[z]]).
$$
Next, we will analyze the quantum commutative algebra 
$$C(\g)= H^0 (\Obs ( 0 \times E)) $$ in the case $\g = \sl_n$. Note that $C(\g)$ is a commutative algebra equipped with a derivation, corresponding to infinitesimal translation in the $z$-direction.
\begin{theorem}
There is an isomorphism of commutative algebras
$$
C(\sl_n) = \Oo( (\mf{h}/W) [[z]] ) [[\hbar]]
$$
under which the derivation on $C(\sl_n)$ corresponds to the derivation $\partial_z$ on $\Oo( (\mf{h}/W) [[z]] ) [[\hbar]].$
\end{theorem}
\begin{remark}
Concretely, this means the following.  Let $P_2,\dots,P_n$ be generators for the ring of functions on $\mf{h}/W$.   Let $P_i^k$ be the function on $(\mf{h}/W)[[z]]$ which sends $A$ to
$$
\partial_z^k P(A(z)) \mid_{z = 0}.
$$
Then, $\Oo ( (\mf{h}/W)[[z]])$ is the power series algebra $\C[[P_i^k \mid i = 2 \dots n, k = 0 \dots \infty]]$. 
\end{remark}
\begin{proof}
We use the existence of the evaluation homomorphism $Y(\mf{sl}_n) \to U(\mf{sl}_n) [[\hbar]]$.  Every linear map
$$
l : U(\mf{sl}_n) [[\hbar]] \to \C[[\hbar]]
$$
with $l([a,b]) = 0$ thus gives an element of $C(\sl_n)$.  The space of such linear maps is $\Oo ( \mf{h}/W)[[\hbar]]$.  Thus, the invariant polynomials $P_i \in \Oo(\mf{h}/W)$ all quantize in a natural way.  Since $C^{cl}(\sl_n)$ is generated by the $P_i$ and their $z$-derivatives, and since the derivation $\partial_z$ lifts to the quantum level, we see that every element of $C^{cl}(sl_n)$ quantizes, as desired.  
\end{proof}

The final result in this section is the following.
\begin{theorem}
The commutative holomorphic vertex algebra $C(\sl_n) \subset H^0  (\Obs ( 0 \times E))$ is maximal. 
\end{theorem}

\begin{proof}

The strategy of the proof is one which is standard in the literature on quantum-mechanical integrable systems.  If one has a Poisson algebra $A$ quantizing to an associative algebra $A^q$, and if $B \subset A$ is a Poisson commutative subalgebra quantizing to a commutative subalgebra $B^q \subset A^q$, then if $B$ is a maximal Poisson commuting subalgebra then $B^q$ is a maximal commutative subalgebra. 

Morally speaking, this type of argument applies directly to our situation.  The analog of $B$ in our situation is $\Oo(\mf{h}/W[[z]])$, and the analog of $A$ is $\Oo ( (\mf{h} \oplus \mf{h} )/ W[[z]] )$.  Since $\mf{h} / W$ is half the size of  $(\mf{h} \oplus \mf{h} )/ W$ -- or more precisely, since $(\mf{h} \oplus \mf{h} )/ W  \to \mf{h} / W$ has Lagrangian fibers -- the result should be immediate.  

However, in order to make this kind of argument precise, we need to use a Poisson-type structure on the commutative algebra of classical observables, which I call a chiral Poisson structure ((in \cite{BeiDri04} this structure is called a coisson algebra; I have just decompressed their notation).    Unfortunately, lack of space means I will only sketch the existence and formal properties of this Poisson structure.  However, this structure will be very familiar to experts on vertex and chiral algebras . 

The structure is the following. For every $f \in \C[z]$, we have a bracket
$$
\pi_f : H^0  (\Obs^{cl} ( 0 \times E)) \otimes H^0  (\Obs^{cl} ( 0 \times E)) \to H^0  (\Obs^{cl} ( 0 \times E))
$$
satisfying the following properties. (Not all properties are listed: in particular, I will not detail the analog of the Jacobi identity here). 
\begin{enumerate}
\item If $\alpha,\beta \in H^0  (\Obs^{cl} ( 0 \times E))$ are observables that quantize to $\til{\alpha}, \til{\beta}$ modulo $\hbar^2$, then
$$
\hbar \pi_f ( \alpha,\beta) = \frac{1}{2 \pi i}\oint_{\abs{z_0} = \eps}  \til{\alpha}(0) \til{\beta}(z_0) f(z_0) \d z_0.
$$
Here, $\til{\alpha}(0)$ and $\til{\beta}(z)$ refer to the observables $\til{\alpha}$, $\til{\beta}$ placed at $0 \times E$ and $z_0 \times E$; we then take their operator product, which is an observable on the disc $\abs{z} < 2 \eps$ depending holomorphically on the position of $z_0$; then we take the contour integral.    One can take this property as defining the bracket $\pi_f$ on the subalgebra of those elements of $H^0 ( \Obs^{cl}(0 \times E))$ which quantize modulo $\hbar^2$.  It is easy to check that the bracket is independent of the choice of the lifts $\til{\alpha}$ and $\til{\beta}$.  A little work is required to check that the bracket has a natural definition on all of $H^0 ( \Obs^{cl} ( 0\times E))$. 
\item $\pi_f$ is a derivation in each factor with respect to the commutative product on $H^0 ( \Obs^{cl} ( 0 \times E))$. 
\item $\pi_{f}(\alpha,\beta)$ is linear in $f$ (as well as in $\alpha$ and $\beta$). 
\item Let $\del$ denote the derivation of $H^0 ( \Obs^{cl} ( 0 \times E))$ corresponding to infinitesimal translation in the $z$-direction. Then, 
$$
\del \pi_f ( \alpha,\beta ) =  \pi_f (\del \alpha, \beta ) + \pi_f ( \alpha,\del \beta). 
$$ 
This follows from the contour integral expression above and the fact that $\del$ is a derivation for the operator product:
$$
\del ( \til{\alpha}(0) \til{\beta}(z_0) ) = \del ( \til{\alpha}(0))  \til{\beta}(z_0)  +  \til{\alpha}(0)  \del \til{\beta}(z_0).
$$
\item In addition, the Leibniz rule 
$$
\pi_f ( \alpha, \del \beta ) + \pi_{\dpa{z} f} (\alpha,\beta) = 0
$$
holds.  This follows from the fact that $\del \beta(z) = \frac{\d}{\d z} \beta(z)$, and integration by parts in the contour integral expression above. 

\item Recall that
$$
H^0 ( \Obs^{cl } (0 \times E)) = \Oo ( Y[[z]] ) 
$$
where $Y = (\mf{h} \oplus \mf{h} ) / W$, and $\Oo$ indicates formal power series around the origin in $Y$.   Note that $Y$ is a Poisson scheme in a natural way: indeed, we can describe $Y$ as the symplectic reduction of $\g \oplus \g = T^\ast \g$ under the adjoint $G$-action. 

There is an evaluation map $\op{ev} : Y[[z]] \to Y$.    If $\alpha,\beta \in \Oo(Y)$, then for all $f \in \C[z]$,
$$
\pi_f( \op{ev}^\ast \alpha, \op{ev}^\ast \beta) = f(0) \op{ev}^\ast \{\alpha,\beta\}
$$ 
where $\{-,-\}$ is the Poisson bracket on $\Oo(Y)$.   This fact is a little calculation: the main point is that the theory we are considering on $\C \times E$, when dimensionally reduced along $\C$, is classically equivalent to the theory describing chiral differential operators on the stack $\g/G$.  It follows that the chiral Poison bracket is the same as the one describing the classical limit of CDOs of $\g/G$, which lives inside the chiral Poisson algebra describing CDOs on $\g$.  
\end{enumerate}
From these properties, it is straightforward to complete the proof. There is a map $Y \to \mf{h}/W$ by projecting onto one factor of $\mf{h} \oplus \mf{h}$.   We need to show the following.  Suppose that $\alpha \in \Oo( Y[[z]])$ has the property that for all $\beta \in \Oo( \mf{h}/ W[[z]])$, and for all $f \in \C[z]$, then $\pi_f ( \alpha,\beta) = 0$.  Then, $\alpha \in \Oo ( \mf{h} / W[[z]] )$.  For brevity, we will describe this property by saying that $\Oo ( \mf{h}/W[[z]])$ is a maximal chiral Poisson commuting subalgebra.  

It suffices to prove the analogous result if we use polynomials instead of formal power series.  For the rest of the proof, $\Oo$ will denote polynomials. 

Let $\mf{h}^{rss} \subset \mf{h}$ be the locus of regular semi-simple elements.  Then, the inverse image of $\mf{h}^{rss}/ W$ in $Y$ is $T^\ast (\mf{h}^{rss} / W)$.  Because $\Oo(Y[[z]])$ is a subalgebra of $\Oo (((T^\ast (\mf{h}^{rss} / W) )[[z]])$, it suffices to show that $\Oo (\mf{h}^{rss}/W[[z]])$ is a maximal chiral Poisson commuting subalgebra of  $\Oo (((T^\ast (\mf{h}^{rss} / W) )[[z]])$.   We will show that the analogous fact holds if we replace $\mf{h}^{rss} / W$ by any smooth algebraic variety.  The statement is \'etale local, so it suffices to work on $\C^n$.  Then, the chiral Poisson algebra $\Oo ( T^\ast \C^n[[z]])$ is, as an algebra, freely generated by $\del^k x_i$ and $\del^k y_j$, where $x_i,y_j$ are coordinates on $T^\ast \C^n$.  The chiral Poisson bracket is 
$$
\pi_{f}  ( \del^k x_i, \del^l y_j) ) = \delta_{ij} \frac{1}{2 \pi i}\oint_{\abs{z} = 1} (-1)^{k+l} (k+l)! z^{-k-l} f(z) \d z . 
$$
It is easy to check from this expression that $\Oo(\C^n[[z]])$ is a maximal chiral Poisson commuting subalgebra. 
\end{proof}

\section{The operator product in the $z$-plane} 
\label{operator_product_z}
So far, we have seen that finite dimensional representations of the Yangian give Wilson loops for our theory; and that the operator product in the $w$-plane encodes the tensor product of such representations.

In this section we will analyze the operator product in the $z$-direction: we will see that this is encoded by the universal $R$-matrix for the Yangian.  

\subsection{}
Before we proceed, I should explain what I mean by the operator product in the $z$-direction.  We have constructed, for each disc $D(z_0,r)$ in the $z$-plane $\C_z$, an $E_2$ algebra $\F_{D(z_0,r)}$, associated by Lurie's theorem to the factorization algebra $\pi_\ast ( \F\mid_{D(z_0,r)})$ on $\C_w$.    We have also constructed a sub-$E_2$ algebra $\F_{z_0} \subset \F_{D(z_0,r)}$, as the direct sum of the $S^1$-eigenspaces of $\F_{D(z_0,r)}$.   We have seen that the $E_2$ algebra $\F_{z_0}$ is Koszul dual to the Yangian.  

All this structure came from a factorization algebra $\F$ on $\C_z\times \C_w$.   However, all of our constructions relied on pushing forward to give a factorization algebra on $\C_w$: we have not used the full structure of a factorization algebra on $\C_z \times \C_w$. 

The remaining structure is encoded by the following result.  If $U \subset \C_z$, let $\F_U$ denote the $E_2$ algebra associated to the factorization algebra $\pi_\ast \F \mid_{U \times \C_w}$ on $\C_w$.  
\begin{proposition}
Sending $U \to \F_U$ defines a holomorphically translation invariant prefactorization algebra on $\C_z$ valued in $E_2$ algebras.  
\end{proposition}
\begin{remark}
\begin{enumerate}
\item Unlike the factorization algebras we have discussed earlier, in this example the axioms of a factorization algebra only hold up to coherent homotopy.  This is because Lurie's theorem relating locally-constant factorization algebras and $E_n$ algebras takes place in the world of $\infty$-categories. 
\item The concept of a holomorphically translation invariant factorization algebra is explained in \cite{CosGwi11}.  Essentially, this means that the operator product in the $z$-direction is holomorphic. One can show that the cohomology of a holomorphically-translation invariant factorization algebra with a compatible $S^1$ action has the structure of a vertex algebra, as defined by Borcherds.  This is not quite good enough for our purposes, as this construction loses the homotopy coherences present in the structure of holomorphically-translation invariant factorization algebra. 
\end{enumerate}
\end{remark}
\begin{proof}
The cochain complex underlying $\F_U$ is $\F(U \times D_w)$, where $D_w \subset \C_w$ is any disc.  The structure of a factorization algebra on $\C_z \times \C_w$ gives us maps
$$
\F_{U_1} \times \dots \times \F_{U_k} \to \F_V
$$
if $U_1 \amalg \dots \amalg U_k \subset V$.  We need to show that these maps are maps of $E_2$ algebras, and that the associativity properties of a factorization algebra hold (up to coherent homotopy) in the category of $E_2$ algebras.  This follows immediately from Lurie's results. 

Indeed, the fact that we have a factorization algebra on $\C_z \times \C_w$ shows that we have maps, in the multicategory of prefactorization algebras on $\C_w$, 
$$
\pi_\ast (\F \mid_{U_1 \times \C_w}) \times \dots \times \pi_\ast (\F \mid_{U_k \times \C_w})  \to \pi_\ast (\F \mid_{V \times \C_w}) ,
$$
so that sending $U$ to $\pi_\ast (\F \mid_{U \times \C_w})$ defines a prefactorization algebra on $\C_z$ valued in prefactorization algebras on $\C_w$.  (See \cite{GinTraZei10} for more on this point).  Now, by Lurie's theorem, we can translate any structure in the world of locally-constant prefactorization algebras on $\C_w$ to one in the world of $E_2$ algebras, so that we have a prefactorization algebra on $\C_z$ valued in $E_2$ algebras.

The fact that we find a holomorphically translation invariant factorization algebra follows from the fact that the factorization algebra $\F$ on $\C_z \times \C_w$ is holomorphically translation invariant. 

\end{proof}

Let us unpack this result a little.  For each disc $D(z,r)$, we have an $E_2$ algebra $\F_{D(z,r)}$.  If $D(z_0,r_0), \dots, D(z_k,r_k) \subset D(z',r')$ are disjoint discs, we have a homomorphism of $E_2$ algebras 
$$
\F_{D(z_0,r_0)} \otimes \dots \otimes \F_{D(z_k,r_k)} \to \F_{D(z',r')},
$$
which doesn't depend on the ordering chosen on the $z_i$. 

Suppose that the first $i+1$ discs $D(z_0,r_0), \dots, D(z_i,r_i)$ are contained in $D(z'', r'')$, and that $D(z'',r'')$ is disjoint from the remaining discs $D(z_{i+1}, r_{i+1}), \dots D(z_k,r_k)$.  Then, the following diagram of $E_2$ algebras commutes up to homotopy: 
$$
\xymatrix{
\F_{D(z_0,r_0)} \otimes  \dots \otimes \F_{D(z_k,r_k)}    \ar[d] \ar[dr] &  \\
\F_{D(z'',r'')} \otimes \F_{D(z_{i+1},r_{i+1})} \dots \otimes \F_{D(z_k,r_k)}  \ar[r] & \F_{D(z',r')}. 
}
$$

The fact that we have a holomorphically translation invariant factorization algebra means that these product maps vary holomorphically with the position of the $z_i$.   Let 
$$P(r_1,\dots,r_k \mid s) \subset \C_z^k$$
be the open subset consisting of those $z_1,\dots, z_k$ where the closures of the discs $D(z_i,r_i)$ are disjoint and contained in $D(0,s)$.    For each $p \in P(r_1,\dots,r_k \mid s)$ we have a product map of $E_2$ algebras
$$
m_p : \F_{D(z_1,r_1)} \otimes \dots \otimes \F_{D(z_k,r_k)} \to \F_{D(0,s)}.
$$
This map varies holomorphically with $p$, in a certain rather technical sense which I will briefly describe. 

As explained in \cite{CosGwi11}, the cochain complex $\F(U)$ of observables on any open subset $U \subset \C_z \times \C_w$ has an extra structure: it is a cochain complex of differentiable vector spaces.   A differentiable vector space $V$ has enough structure to talk about smooth maps $M \to V$ from any manifold $M$, and to differentiate these smooth maps.   The differentiable structure on the factorization algebra $\F$ is compatible with products, and so makes $\F_{D(0,s)}$ into a differentiable $E_2$ algebra.  The differentiable structure on $\F_{D(0,s)}$ is enough structure that we can talk about the Dolbeault complex $\Omega^{0,\ast}( P(r_1,\dots,r_k \mid s) , \F_{D(0,s)})$: this is again an $E_2$ algebra. 

The map $m_p$ is holomorphic in the sense that it lifts to a map of $E_2$ algebras
$$
m_{P(r_1,\dots,r_k \mid s )} : \F_{D(0,r_1)} \otimes \dots \otimes \F_{D(0,r_k)} \to \Omega^{0,\ast}(P(r_1,\dots,r_k \mid s), \F_{D(0,s)}).
$$
(In this expression we have identified $\F_{D(z_i,r_i)}$ with $\F_{D(0,r_i)}$ using the translation isomorphism).  

\subsection{}
Let us now focus on the map of $E_2$ algebras that encodes the $R$-matrix. Fix $s,r_0,r_1 > 0$ with $2 r_0 + r_1 < s$.  Let  $A$ be the annulus consisting of those $z \in \C$ with 
$$
r_1 + r_0 < \abs{z} < s - r_0.
$$
The assumptions on $r, s_0, s_1$ guarantee that this set is non-empty.  

Note that $A$ is the locus where the discs $D(z, r_0)$ and $D(0,r_1)$ are disjoint and contained in $D(0,s)$. 

As above, the operator product in the $z$-direction gives a map in the homotopy category of $E_2$ algebras
\begin{equation*}
\F_{D(0,r_0)} \otimes \F_{D(0,r_1)} \to \F_{D(0,s)} \otimes \Omega^{0,\ast}(A). \tag{$\dagger$} 
\end{equation*}
Note that this map is $S^1$-equivariant, where $S^1$ acts on the left hand side by rotation the discs $D(0,r_i)$ and on the right hand side by rotating the disc $D(0,s)$ and the annulus $A$.  (This $S^1$-equivariance follows immediately from the fact that the factorization algebra $\F$ on $\C_z \times \C_w$ we started with is equivariant for the $S^1$ action which rotates around the origin in $\C_z$). 

By passing to the direct sum of $S^1$-eigenspaces on the right hand side, we get a map 
\begin{equation*}
\F_0 \otimes \F_0 \to \F_{D(0,s)} \otimes \Omega^{0,\ast}(A) \tag{$\ddagger$}.
\end{equation*}
\begin{lemma}
The map of $E_2$ algebras $\ddagger$ factors (up to homotopy) through $\F_{D(0,s)} \otimes \Omega^{0,\ast}(D(0,s))[z^{-1}]$. 

Further, there is a map in the homotopy category of $E_2$ algebras
$$
m_{OPE} : \F_0 \otimes \F_0 \to \F_0 \otimes \C((z)) 
$$
with the property that the diagram of $E_2$ algebras commutes up to homotopy:
$$
\xymatrix{
\F_0 \otimes \F_0 \ar[r] \ar[d]& \F_0 \otimes \C((z)) \ar[d] \\
\F_{D(0,s)} \otimes \Omega^{0,\ast}(D(0,r))[z^{-1}] \ar[r] & \F_{D(0,s)} \otimes \C((z)) }.
$$
Here the left hand vertical arrow is the map ($\ddagger$), and the bottom horizontal arrow is Laurent expansion (which we will see makes sense in this context).
\label{lemma_ope_exists} 
\end{lemma}
This lemma thus states that the operator product map ($\dagger$)
admits a nice Laurent expansion when $r_i, s \to 0$. The Laurent expansion of the operator product is, of course, the OPE: but it is now viewed as a map of $E_2$ algebras.  

\begin{proof}
There is a natural map of $E_2$ algebras
$$
\F_{D(0,s)} \otimes \Omega^{0,\ast}(D(0,r))[z^{-1}] \to \F_{D(0,s)} \otimes \Omega^{0,\ast}(A).
$$
This map induces a quasi-isomorphism on each $S^1$-eigenspace. Therefore, by taking the direct sum of $S^1$-eigenspaces on the left hand side of ($\dagger$), we get a map of $E_2$ algebras
$$
\F_0 \otimes \F_0 \to \F_{D(0,s)} \otimes \Omega^{0,\ast}(D(0,s))[z^{-1}]
$$
Next, we use the Taylor expansion map
$$
\Omega^{0,\ast}(D(0,s)) \to \Omega^{0,\ast}(\what{D}_0),
$$
where $\what{D}_0$ denotes the formal punctured disc around $0$.  Since $\Omega^{0,\ast}(\what{D}_0)$ is quasi-isomorphic to $\C[[z]]$, we get a map
$$
\F_0 \otimes \F_0 \to \F_{D(0,s)} \otimes \C((z)). 
$$
Finally, we can identify the direct sum of the $S^1$-eigenspaces on the right hand side with $\F_0 \otimes \C((z))$, giving the desired map. 
\end{proof}
\section{The $R$-matrix}
\label{section_R_matrix}
In this section, we will prove the following theorem (whose statement will be made precise shortly).
\begin{theorem}
After applying Koszul duality, the operator product map
$$
m_{OPE} : \F_0 \otimes \F_0 \to \F_0 \otimes \C((\lambda))
$$
(constructed in lemma \ref{lemma_ope_exists}) is encoded by the universal $R$-matrix $R \in Y(\g) \otimes Y(\g) ((\lambda))$. 
\end{theorem}

Before we can state this result more precisely, and prove it, we need to state some more facts about the Yangian.  

Recall that $Y(\g)$ is a topological Hopf algebra over $\C[[\hbar]]$ which $U(\g[[z]])$. We let
$$
Y(\g) ((\lambda)) = Y(\g) \what{\otimes} \C((\lambda))
$$
be the completed projective tensor product of $Y(\g)$ with $\C((\lambda))$, where $\C((\lambda))$ is also given its natural topology.  

Sending $z \mapsto z + \lambda$ gives an embedding $U(\g[[z]]) \into U(\g[[z]])((\lambda))$.  This embedding quantizes to an embedding of Hopf algebras
$$
T_\lambda : Y(\g) \into Y(\g)((\lambda))
$$
where $Y(\g)((\lambda))$ is viewed as a topological Hopf algebra over the field $\C((\lambda))$. 

We will let $T_0 : Y(\g) \into Y(\g)\otimes \C((\lambda))$ denote the obvious embedding, sending $\alpha \mapsto \alpha \otimes 1$.  

As above, $Y(\g)$ has a $\C^\times$-action where $\hbar$ has weight $1$ and, modulo $\hbar$, the parameter $z$ has weight $1$. 

\begin{theorem}[Drinfeld, \cite{Dri87} page 814; see also \cite{ChaPre95}, page 418]
\label{theorem_existence_R_matrix}
There is a unique element
$$
R(\lambda) \in Y(\g) \otimes_{\C[[\hbar]]} Y(\g) ((\lambda))
$$
with the following properties.
\begin{enumerate}
\item $R(\lambda)$ is invariant under the $\C^\times$-action which on $Y(\g)$ is the one discussed before, and which gives the parameter $\lambda$ weight $1$. 
\item $R(\lambda)$ satisfies the following identity.  Let $c$ denote the coproduct on $Y(\g)$, and $\sigma$ the isomorphism of $Y(\g) \otimes Y(\g)$ which switches the factors.  Then, for all $\alpha \in Y(\g)$, 
\begin{equation*}
(T_\lambda \otimes T_0) \sigma c (\alpha)  = R(\lambda)((T_\lambda \otimes T_0) c(\alpha)) R(\lambda)^{-1} \in Y(\g) \otimes Y(\g) ((\lambda)).  \tag{$R$-matrix equation $1$}
\end{equation*}
\item  The following equations hold:
\begin{align*}
(c \otimes 1) R(\lambda) &= R_{13}(\lambda)R_{23}(\lambda) \tag{$R$-matrix equation $2$}\\
(1 \otimes c) R(\lambda) &= R_{13}(\lambda)R_{12}(\lambda)  \tag{$R$-matrix equation $3$}.
\end{align*}
\end{enumerate}
\end{theorem}
\begin{remark}
The $R$-matrix $R(\lambda)$ is also a solution to the Yang-Baxter equation, but the equations listed above are enough to characterize it.  
\end{remark}
\begin{remark}
The $R$-matrix is of the form
$$
R(\lambda) = \sum \lambda^{-i} R_i
$$
where $R_i \in Y(\g)^{\otimes 2}$ is of weight $i$ with respect to the $\C^\times$ action on $Y(\g)$.  This sum is infinite.  However, because we are treating $Y(\g)$ as a topological Hopf algebra, and we are using the completed projective tensor product in defining $Y(\g) \otimes Y(\g) ((\lambda))$, we can still view $R(\lambda)$ as being in this ring.  Indeed, $Y(\g)$ is the product of it's $\C^\times$-eigenspaces, and the topology on $Y(\g)$ is the product topology.  The completed projective tensor product commutes with countable products, so 
$$Y(\g)((\lambda)) = \prod_i \left\{Y_i(\g)((\lambda)) \right\}$$
where $Y_i(\g)$ is the $i$'th $\C^\times$-weight space. Also, an element is a series of the form
$$
\sum_n \lambda^n \alpha_n
$$
where $\alpha_n \in Y_i(\g)$ and $\alpha_n = 0$ for $n << 0$. 

There are some subtle issues that arise when dealing with these completed tensor products. The main subtlety is that the product map
$$
\C((\lambda)) \times \C((\lambda)) \to \C((\lambda))
$$
is \emph{not} continuous, but only separately continuous.   Thus, it does not extend to the completed projective tensor product of $\C((\lambda))$ with itself. One can see this by observing that if it was continuous, it would extend to the completed projective tensor product; and that
$$
\sum_{i = 0}^\infty \lambda_1^i \lambda_2^{-i} \in \C[[\lambda_1]] \what{\otimes} \C[\lambda_2^{-1}]
$$
is in the completed projective tensor product. Clearly applying multiplication to this element does not converge.  
\end{remark}
\subsection{}
It will be helpful to understand in more categorical terms the meaning of the equations characterizing $R(\lambda)$.    If $V$ is a $Y(\g)$-module, let 
$$T_\lambda V = V \otimes_{Y(\g)} Y(\g)((\lambda))$$
where the tensor product is taken using the Hopf algebra homomorphism $T_\lambda : Y(\g) \to Y(\g)((\lambda))$.  Similarly, let $T_0 V = V((\lambda))$ denote the obvious extension of $V$ to a $Y(\g)((\lambda))$-module. 
\begin{remark}
One must be careful with the definition of $T_\lambda$: it \emph{does not} extend to a $\C((\lambda))$-linear map map
$$
Y(\g) ((\lambda)) \to Y(\g)((\lambda)).
$$
This is because the action map
$$
Y(\g) ((\lambda)) \times \C((\lambda)) \to Y(\g)((\lambda))
$$
is not continuous, only separately continuous. 
\end{remark}
The main result of this subsection is the following. 
\begin{theorem}
An $R$-matrix satisfying the $R$-matrix equations $1$, $2$ and $3$ gives rise to a natural transformation making the following functor monoidal: 
\begin{align*}
F: Y(\g)-\op{mod} \times Y(\g)-\op{mod} &\to Y(\g)((\lambda))-\op{mod} \\
V \times W & \mapsto T_0(V) \otimes_{\C[[\hbar]]((\lambda))} T_\lambda(W).
\end{align*}
\label{theorem_monoidal}
\end{theorem}
\begin{remark}
In the statement of the theorem, and what follows, we only work with  $Y(\g)$-modules $V$ which are of finite rank as $\C[[\hbar]]$-modules. Also, for simplicity, we will deal with ordinary (not dg) modules. Later we will see a generalization of this, whereby the $R$-matrix gives rise to a similar natural transformation on the category of all $Y^\ast(\g)$-comodules. (Recall that we can identify the category of $Y^\ast(\g)$-comodules with a certain category of inductive objects in the category of finite $Y(\g)$-modules). When we work in this generality, we will find that to give an $R$-matrix satisfying equations $1$,$2$ and $3$ is the same as to give a natural transformation making the functor $F$ monoidal. 
\end{remark}

The main technical result of this subsection is the following.
\begin{proposition}
An $R$-matrix $R \in Y(\g) \otimes Y(\g)((\lambda))$ satisfying $R$-matrix equation $1$ gives, for all finite $Y(\g)$-modules $V$ and $W$, an isomorphism of $Y(\g)((\lambda))$-modules
$$
\Phi:  T_0(V) \otimes_{\C[[\hbar]]((\lambda))} T_\lambda( W) \to T_\lambda(W) \otimes_{\C[[\hbar]]((\lambda))} T_0 (V)
$$
functorial in $V$ and $W$.   

Similarly, $R$-matrix equations $2$ and $3$ imply that for all $Y(\g)$-modules $V_1,V_2,V_3$, the following two diagrams  of $Y(\g)((\lambda))$-modules commute. In these diagrams, all tensor products are over  $\C[[\hbar]]((\lambda))$ unless otherwise noted. Further $\Phi_{ij}$ will mean $\Phi$ acting on the $ij$ factor.
\begin{equation*}
\xymatrix{
T_0(V_1) \otimes T_0(V_2) \otimes T_\lambda(V_3)  \ar[r]^{\Phi_{23} }   \ar[d]_{ = } &  T_0(V_1) \otimes T_\lambda (V_3)  \otimes T_0(V_2)  \ar[d]^{\Phi_{12}} \\
T_0 \left(V_1\otimes_{\C[[\hbar]]} V_2\right) \otimes T_\lambda(V_3) \ar[r]_{\Phi}  & T_\lambda(V_3) \otimes T_0(V_1) \otimes  T_0 (V_2)  
} \tag{$R$-matrix diagram $1$}
\end{equation*}
\begin{equation*}
\xymatrix{
T_0(V_1) \otimes T_\lambda(V_2) \otimes T_\lambda(V_3)  \ar[r]^{\Phi_{12} }   \ar[d]_{ = } &  T_\lambda(V_2) \otimes T_0 (V_1)  \otimes T_\lambda(V_3)  \ar[d]^{\Phi_{23}} \\
T_0(V_1) \otimes T_\lambda\left(V_2\otimes_{\C[[\hbar]]} V_3\right)  \ar[r]_{\Phi}  & T_\lambda(V_2) \otimes T_\lambda(V_3) \otimes  T_0 (V_1)  
} \tag{$R$-matrix diagram $2$}
\end{equation*}
\label{proposition_Rmatrix_categorical}
\end{proposition}
\begin{proof}
Let $F,F'$ be the functors
 \begin{align*}
F,F': Y(\g)-\op{mod} \times Y(\g)-\op{mod} &\to Y(\g)((\lambda))-\op{mod} \\
F (V \times W) &= T_0(V) \otimes_{\C[[\hbar]]((\lambda))} T_\lambda(W)\\
F' (V \times W) &= T_\lambda(W) \otimes_{\C[[\hbar]]((\lambda))} T_0(V)
\end{align*}
Let us first show that an $R$-matrix gives rise to a natural isomorphism $F \iso F'$. 

Note that, as $\C((\lambda))$-modules, $T_0(V) \otimes T_\lambda(W) = V \otimes W((\lambda))$.  Since $V \otimes W$ is a $Y(\g) \otimes Y(\g)$-module, multiplication by the $R$-matrix acts on $V \otimes W((\lambda))$.   We let
$$
R_{V,W} \in \op{End}_{\C((\lambda))}(V \otimes W((\lambda)))
$$
be this endomorphism. 

\begin{remark}
If $V$ and $W$ are not of finite rank over $\C[[\hbar]]$, then $R_{V,W}$ is not well-defined. 
\end{remark}
Let $\sigma : V \otimes W \to W \otimes V$ be the isomorphism which interchanges the two factors.

We will show that the map
$$
\sigma \circ R_{V,W} : T_0(V) \otimes T_\lambda(W) \to T_\lambda(W) \otimes T_0(V)
$$
is an isomorphism of $Y(\g)((\lambda))$-modules. 

Let $R' = \sigma(R) \in Y(\g) \otimes Y(\g)$.  Then,
$$
\sigma \circ R_{V,W} = R'_{W,V} \circ \sigma.
$$
In the above equations, we have used $\otimes$ to denote the tensor product on the category of $Y(\g)$-modules coming from the coproduct on $Y(\g)$.  Let us use the notation $\otimes^{op}$ for the tensor product coming from the opposite coproduct.  Then, $\sigma$ is an isomorphism
$$
\sigma : T_0(V) \otimes T_\lambda(W) \to T_\lambda(W) \otimes^{op} T_0(V) 
$$
of $Y(\g)((\lambda))$-modules. 

Thus, to verify that $R'_{W,V} \circ \sigma$ is a map of $Y(\g)((\lambda))$-modules, it suffices to verify that
$$
R'_{W,V} : T_\lambda(W) \otimes^{op} T_0(V) \to T_\lambda(W) \otimes T_0(V)
$$
is a map of $Y(\g)((\lambda))$-modules.   

Both sides of this equation are, as $\C((\lambda))$-modules, $W \otimes V((\lambda))$, but with different $Y(\g)$-module structures.  In one case, corresponding to $T_\lambda(W) \otimes T_0(V)$, the action of $\alpha \in Y(\g)$ is through the element
$$(T_{-\lambda} \otimes T_0) c(\alpha) \in Y(\g) \otimes Y(\g).$$
We are using here the fact that, if $V$ is a $Y(\g)$-module, an element $\beta \in Y(\g)$ acts on the $Y(\g)$-module $T_\lambda(V)$ by the action of $T_{-\lambda}(\beta)$ on $V((\lambda))$. 

In the other $Y(\g)$-module structure on $W \otimes V((\lambda))$, corresponding to $T_\lambda(W) \otimes^{op} T_0(V)$, the action is through the element 
$$
(T_{-\lambda}) \otimes T_0) c^{op}(\alpha)
$$
where $c^{op}$ is the opposite coproduct. 

Thus, to show that $R'_{W,V}$ gives a map of $Y(\g)((\lambda))$-modules, it suffices to show that, for all $\alpha \in Y(\g)$, we have the following identity in $Y(\g) \otimes Y(\g)((\lambda))$:
$$
R'(\lambda) (T_{-\lambda} \otimes T_0) (c^{op}(\alpha)) 
= (T_{-\lambda} \otimes T_0) c(\alpha) R'(\lambda) .
$$
Now, we can replace $\lambda$ by $-\lambda$ in this expression, and use the identity $R'(-\lambda) = R(\lambda)^{-1}$ (\cite{Dri87}).  This gives us the equation
$$
R(\lambda)^{-1} (T_{\lambda} \otimes T_0) (c^{op}(\alpha)) 
= (T_{\lambda} \otimes T_0) c(\alpha) R(\lambda)^{-1} .
$$
This is $R$-matrix equation $1$. 

It remains to check that commutativity of $R$-matrix diagrams $1$ and $2$ in the statement of the proposition correspond to $R$-matrix equations $2$ and $3$.  We will check that the equations for the $R$-matrix imply commutativity of the first diagram: the commutativity of the second diagram is straightforward.

In what follows, we will use the notation $\sigma$ with an appropriate subscript to denote an element of $S_3$ acting on a tensor product of $3$ vector spaces.  We will use the notation $R_{ij}$ to denote the $R$-matrix in the $i,j$ factors. 

Note that the map
$$
T_0(V_1 \otimes V_2) \otimes T_\lambda(V_3) \to T_\lambda(V_3) \otimes T_0(V_1 \otimes V_2) 
$$
is $\sigma_{123}R_{V_1 \otimes V_2, V_3}$.  Now, by the way the tensor product is defined, 
$$
R_{V_1 \otimes V_2, V_3} = (c \otimes 1)R
$$
acting on $V_1 \otimes V_2 \otimes V_3$. 

On the other hand, the composition
$$
T_0(V_1) \otimes T_0(V_2) \otimes T_\lambda(V_3) \to T_0(V_1) \otimes T_\lambda(V_3) \otimes T_0(V_2) \to T_\lambda(V_3) \otimes T_0(V_1) \otimes T_0(V_2)
$$
is the composition of $\sigma_{23}R_{23}$ with $\sigma_{12}R_{12}$. 

From this, we see that commutativity of $R$-matrix diagram $1$ amounts to the equation
$$
\sigma_{123} ((c \otimes 1) R)  = \sigma_{12} R_{12} \sigma_{23} R_{23}.
$$
This is equivalent to the equation
$$
(c \otimes 1)R = R_{13} R_{23},
$$
which is $R$-matrix equation $2$.  The prove that $R$-matrix diagram $2$ is equivalent to $R$-matrix equation $3$ is similar. 
\end{proof}

Now we can show how this proposition leads to a proof of theorem \ref{theorem_monoidal}.
\begin{lemma}
Define the functor 
$$
F : Y(\g)-\op{mod} \times Y(\g)-\op{mod} \to Y(\g)((\lambda))-\op{mod}
$$
by
$$
F(V \times W) = T_0(V) \otimes T_\lambda(W).
$$
Then, giving a natural isomorphism making the functor $F$ monoidal is equivalent to giving the natural transformations presented in the previous proposition.
\label{lemma_monoidal}  
\end{lemma}
\begin{proof}
Recall that, if $\mc{C}, \mc{D}$ are monoidal categories, a monoidal functor $F: \mc{C} \to \mc{D}$ is a functor with a natural isomorphism
$$
F(c \otimes c') \iso F(c) \otimes F(c')
$$
such that the following diagram commutes:
$$
\xymatrix{
F(c_0 \otimes c_1 \otimes c_2)  \ar[r] \ar[d] & F(c_0 \otimes c_1) \otimes F(c_2) \ar[d]  \\
F(c_0) \otimes F(c_1 \otimes c_2) \ar[r] & F(c_0) \otimes F(c_1) \otimes F(c_2).  
}$$

Thus, the first thing we need to check is that giving a natural isomorphism
\begin{equation*}
T_0(V) \otimes T_\lambda(W) \iso T_\lambda(W) \otimes T_0(V)  \tag{$\dagger$}
\end{equation*}
is equivalent to giving a natural isomorphism
\begin{equation*}
T_0 (V_0 \otimes V_1) \otimes T_\lambda (W_0 \otimes W_1) 
\iso (T_0(V_0) \otimes T_\lambda(W_0)) \otimes (T_0(V_1) \otimes T_\lambda(W_1)) \tag{$\ddagger$}. 
\end{equation*}
The natural isomorphism $\ddagger$ is obtained by applying the isomorphism in $\dagger$ to the middle two objects in the tensor product $T_0(V_0) \otimes T_\lambda(W_0) \otimes T_0(V_1) \otimes T_\lambda(W_1)$.

Conversely, the natural isomorphism $\dagger$ is the special case of $\ddagger$ when $V_0$ and $W_1$ are the identity objects. 

Next, we need to check that the natural isomorphism $\ddagger$ satisfies the coherence condition necessary for a monoidal functor; and this coherence condition is equivalent to the commutativity of $R$-matrix diagrams $1$ and $2$ above. 

Let us write out the diagram that needs to commute to ensure that $F$ is monoidal: 
$$
\xymatrix{
T_0(V_0 \cdot V_1 \cdot V_2) \cdot T_\lambda(W_0 \cdot W_1 \cdot W_2)   \ar[r] \ar[d] 
& T_0(V_0 \cdot V_1) \cdot T_\lambda (W_0 \cdot W_1) \cdot T_0(V_2) \cdot T_\lambda(V_2)  \ar[d]  \\
T_0(V_0) \cdot T_\lambda(W_0) \cdot T_0(V_1 \cdot V_2) \cdot T_\lambda(W_1 \cdot W_2) \ar[r] 
& T_0(V_0) \cdot T_\lambda(W_0) \cdot T_0(V_1) \cdot T_\lambda(W_1) \cdot T_0(V_2) \cdot T_\lambda(W_2).  
}$$
In this diagram, the symbol $\otimes$ has been replaced by $\cdot$ to save space.

Commutativity of this diagram follows from repeatedly using the commutativity of $R$-matrix diagrams $1$ and $2$ above.

Conversely, $R$-matrix diagrams $1$ and $2$ are specializations of this commutative diagram, to the case when $V_0,V_1$ and $W_2$ are the identity object, and to the case when $V_0, W_1$ and $W_2$ are the identity object.  

\end{proof}

\subsection{}
In this paper, the dual Yangian plays a more fundamental role than the Yangian.   In this section, we will prove an analog of theorem \ref{theorem_monoidal} for the dual Yangian.  When we work with the dual Yangian, the theorem applies to all modules, not just finite ones.  This allows us to deduce a converse: giving an $R$-matrix satisfying equations $1$, $2$ and $3$ above is \emph{equivalent} to giving a natural transformation making the functor $F$ monoidal.

There are a few subtle points in the statement and proof of this result, to do with various different completed tensor products we need to use.   For a vector space $V$, there are two different tensor products of $V$ with $\C((\lambda))$ we need to consider. The first is simply the algebraic tensor product, which we denote $V \br{\otimes} \C((\lambda))$.  In the algebraic tensor product only finite sums of decomposable elements $v \otimes f(\lambda)$ appear.

The second is the completed projective tensor product $V \what{\otimes} \C((\lambda))$, which uses the natural topology on $\C((\lambda))$.  Concretely, we can write $V \what{\otimes} \C((\lambda))$ as a colimit of a limit:
$$
V \what{\otimes} \C((\lambda)) = \colim_{k \to - \infty} \lim_{j \to \infty} \lambda^{-k} \left(V[\lambda] / \lambda^j  \right).
$$

Similarly, if $V$ is a complete filtered vector space, we let 
\begin{align*}
V \what{\otimes} \C((\lambda)) &= \liminv \left\{ (V / F^i V) \what{\otimes} \C((\lambda)) \right\} \\
V \br{\otimes} \C((\lambda)) &= \liminv \left\{ (V / F^i V) \br{\otimes} \C((\lambda) \right\} .
\end{align*}

In particular, we have two different Hopf algebras $Y^\ast(\g) \what{\otimes} \C((\lambda))$, and $Y^\ast(\g) \br{\otimes} \C((\lambda))$ over $\C((\lambda))$, associated to the dual Yangian. Note, however, that the coproduct on $Y^\ast(\g) \what{\otimes} \C((\lambda))$ lands in 
$$
(Y^\ast(\g) \otimes Y^\ast(\g) ) \what{\otimes} \C((\lambda)).
$$
Thus, we can view $Y^\ast(\g) \what{\otimes} \C((\lambda))$ as a \emph{lax} Hopf algebra over $\C((\lambda))$. In section \ref{subsection_laxHopf} I explain how to construct a dg category of comodules over a lax Hopf algebra of this form. 

The translation map $T_\lambda : Y(\g) \to Y(\g)[[\lambda]]$ has a dual, which is a map $T_\lambda : Y^\ast(\g) \to Y^\ast(\g) \what{\otimes} \C[[\lambda]]$.  It extends to a $\C((\lambda))$-linear map of Hopf algebras over $\C((\lambda))$
$$
Y^\ast(\g) \br{\otimes} \C((\lambda)) \to Y^\ast(\g) \what{\otimes} \C((\lambda)). 
$$

If $V$ is a $Y^\ast(\g)$-comodule, we let 
$$
T_\lambda(V) = V \br{\otimes} \C((\lambda))
$$
be the $Y^\ast(\g) \what{\otimes} \C((\lambda))$-comodule where the coaction map is that obtained from the natural $Y^\ast(\g) \br{\otimes} \C((\lambda))$-comodule structure via the homomorphism of coalgebras
$$
T_{-\lambda} : Y^\ast(\g) \br{\otimes} \C((\lambda)) \to Y^\ast(\g) \what{\otimes} \C((\lambda)). 
$$
Similarly, we let $T_0(V)$ be the $Y^\ast(\g) \what{\otimes} \C((\lambda))$ module obtained from $V \br{\otimes} \C((\lambda))$ via the inclusion
$$
T_0 : Y^\ast(\g) \br{\otimes} \C((\lambda)) \to Y^\ast(\g) \what{\otimes} \C((\lambda))
$$
which is the identity on $Y^\ast(\g)$. 

Define a functor $F$ by
\begin{align*}
F : Y^\ast(\g)-\op{comod} \times Y^\ast(\g)-\op{comod} & \to \left\{ Y^\ast(\g)\what{\otimes} \C((\lambda))\right\}-\op{comod} \\
F( V \times W) &= T_\lambda(V) \otimes T_0(W).  
\end{align*}

Finally, recall that the Yangian $Y(\g)$ is the $\C[[\hbar]]$-linear dual of $Y^\ast(\g)$, and thus has a filtration dual to the complete decreasing filtration on $Y^\ast(\g)$. In particular, $F^0 Y(\g)$ is the space of filtration-preserving $\C[[\hbar]]$-linear maps $Y(\g) \to C[[\hbar]]$. 
\begin{theorem}
To give an $R$-matrix 
$$
R(\lambda) \in F^0 \left( Y(\g) \otimes Y(\g) \right) ((\lambda))
$$
satisfying $R$-matrix equations $1$, $2$ and $3$ is the same as to give a natural transformation making the functor $F$ into a monoidal functor.  
\label{theorem_monoidal_dual}
\end{theorem}
\begin{proof}
According to lemma \ref{lemma_monoidal}, it suffices to show that an $R$-matrix is equivalent to giving a natural isomorphism
$$
T_0(V) \otimes T_\lambda(W) \iso T_\lambda(W) \otimes T_0(V)
$$
satisfying the diagrams in proposition \ref{proposition_Rmatrix_categorical}. 

Let $V,W$ be $Y^\ast(\g)$ modules. We let $c_V, c_W$ denote the coproduct on $V$ and $W$.  Any element $\alpha \in Y(\g) \otimes Y(\g)$ gives a linear map
$$
\alpha_{V,W} : V \otimes W \to V \otimes W
$$
by the composition
$$
V \otimes W \xto{c_V \otimes c_W} Y^\ast(\g) \otimes V \otimes Y^\ast(\g) \otimes W \xto{\alpha_{13}} V \otimes W 
$$
where we view $\alpha$ as a linear map $Y^\ast(\g) \otimes Y^\ast(\g) 
\to \C$.  

Note that
$$
F^0 \left( Y(\g) \what{\otimes} Y(\g)  \what{\otimes} \C((\lambda)) \right)
$$
is the space of filtration-preserving $\C[[\hbar]]$-linear maps 
$$
Y^\ast(\g) \otimes Y^\ast(\g) \to \C((\lambda)).
$$
Thus, the $R$-matrix $R(\lambda)$ gives such a map. In this way, we construct a linear map
$$
R_{V,W} : V \otimes W  \to V \otimes W \br{\otimes} \C((\lambda))
$$
which we extend by $\C((\lambda))$-linearity to a map
$$
R_{V,W} : V \otimes W \br{\otimes} \C((\lambda)) \to V \otimes W \br{\otimes} \C((\lambda)).
$$

Let $\otimes^{op}$ denote the opposite tensor product on the category of $Y^\ast(\g)$-modules.  Recall that $T_\lambda(W) \otimes T_0(V)$ and $T_\lambda(W) \otimes^{op} T_0(V)$ are both $Y^\ast(\g) \what{\otimes} \C((\lambda))$-comodules with the same underlying $\C((\lambda))$-module, namely $W \otimes V \br{\otimes} \C((\lambda))$.   However, the $Y^\ast(\g)\what{\otimes}\C((\lambda))$-comodule structure is different in the two cases. 

Let $R' = \sigma(R) \in Y(\g) \otimes Y(\g)((\lambda))$, and let 
$$
R'_{W,V} : W \otimes V \br{\otimes} \C((\lambda)) \to W \otimes V \br{\otimes} \C((\lambda)) 
$$
be the map constructed as above. 

We need to show that the map $R'_{W,V}$ intertwines the two different $Y^\ast(\g) \what{\otimes} \C((\lambda))$-comodule structures on $W \otimes V \br{\otimes} \C((\lambda))$, and so gives us a map of $Y^\ast(\g) \what{\otimes} \C((\lambda))$-comodules
$$
R'_{W,V} : T_\lambda(W) \otimes^{op} T_0(V) \to T_\lambda(W) \otimes T_0(V).
$$
As in proposition \ref{proposition_Rmatrix_categorical}, this follows from $R$-matrix equation $1$.

Conversely, we need to show how a natural isomorphism of $Y^\ast(\g) \what{\otimes} \C((\lambda))$-modules 
\begin{equation*}
T_\lambda(W) \otimes^{op} T_0(V) \to T_\lambda(W) \otimes T_0(V) \tag{$\dagger$}
\end{equation*}
arises from an $R$-matrix satisfying equation $1$. To see this, we consider the cobimodules corepresenting these two functors.  The functor
$$
T_\lambda(W) \otimes T_0(V)
$$
is represented by the $Y^\ast(\g) \what{\otimes} \C((\lambda)) - \left( Y^\ast(\g) \otimes Y^\ast(\g) \br{\otimes} \C((\lambda)) \right)$-cobimodule
$$
Y^\ast(\g) \otimes Y^\ast(\g) \br{\otimes} \C((\lambda)).
$$
The right coaction of $Y^\ast(\g) \otimes Y^\ast(\g) \br{\otimes} \C((\lambda))$ is the right regular corepresentation (composed with the automorphism $\sigma$ which switches the factors). The left coaction of $Y^\ast(\g) \what{\otimes} \C((\lambda))$ is obtained by composing the left regular corepresentation with the $\C((\lambda))$-linear homomorphism
\begin{multline*}
m \circ (T_\lambda \otimes T_0 ) : Y^\ast(\g) \otimes Y^\ast(\g) \br{\otimes} \C((\lambda)) \xto{T_\lambda \otimes T_0} (Y^\ast(\g)\what{\otimes} \C((\lambda)) ) \otimes_{\C((\lambda))} (Y^\ast(\g) \what{\otimes} \C((\lambda)) ) \\
 \xto{m} Y^\ast(\g) \what{\otimes} \C((\lambda)).  
\end{multline*}
The cobimodule representing the functor
$$
T_\lambda W \otimes^{op} T_0 V 
$$
is defined in the same way, except that we use $m^{op}$ in place of $m$ in defining the coaction of $Y^\ast(\g) \what{\otimes} \C((\lambda))$. 

A map of these cobimodules is necessarily colinear over $Y^\ast(\g) \otimes Y^\ast(\g) \br{\otimes} \C((\lambda))$.

In general, if $C$ is a coalgebra, a map
$$
\phi : C \to C
$$
which commutes with the right regular coaction of $C$ is necessarily of the form
$$
\phi(c) = (\alpha \otimes 1) \tr(c)
$$
for some $\alpha \in C^\ast$ (where $\tr(c)$ is the coproduct). 

Applying this to $Y^\ast(\g) \otimes Y^\ast(\g) \br{\otimes} \C((\lambda))$, we see that any natural isomorphism as in equation $(\dagger)$ is of the form $R'_{W,V}$ for some $R' \in F^0 (Y(\g) \otimes Y(\g) ) \what{\otimes} \C((\lambda))$.  The point is that the space of $\C((\lambda))$-linear filtration preserving linear maps
$$
Y^\ast(\g) \otimes Y^\ast(\g) \br{\otimes} \C((\lambda)) \to \C((\lambda))
$$
is $F^0 (Y(\g) \otimes Y(\g) ) \what{\otimes} \C((\lambda))$. 

The statement that $R'_{W,V}$ intertwines the left coactions of $Y^\ast(\g) \what{\otimes} \C((\lambda))$ is $R$-matrix equation $1$. 

The proof of proposition \ref{proposition_Rmatrix_categorical} shows that if $R$ satisfies $R$-matrix equations $1$ and $2$, then $R$-matrix diagrams $1$ and $2$ in \ref{proposition_Rmatrix_categorical}.  The converse is immediate. 
\end{proof}

\subsection{}
Our claim is that operator product expansion, which is a homomorphism of $E_2$ algebras 
$$
m_{OPE} : \F_{0}  \otimes \F_0 \to \F_0 \what{\otimes} \C((z)),
$$
is encoded by the $R$-matrix. In fact, there is a sign difference between the natural convention for the operator product expansion and the usual convention for the $R$-matrix. We thus let
$$
m'_{OPE} :\F_{0}  \otimes \F_0 \to \F_0 \what{\otimes} \C((\lambda))
$$
where we set $\lambda = -z$.  In other words, this is the OPE for the collision of an operator at $0$ with an operator at $-\lambda$, as $\lambda$ tends to zero. 

Let $F_{OPE}$ be the functor 
$$
\F_0-\op{mod} \times \F_0-\op{mod} \to \F_0((\lambda))-\op{mod} 
$$
associated to $m'_{OPE}$.  Since $m'_{OPE}$ is a homomorphism of $E_2$ algebras, $F_{OPE}$ is a monoidal functor.

We have seen in corollary \ref{cor:Yangian_equiv_categories} that there is a quasi-equivalence of dg monoidal categories
$$
\F_{0}-\op{mod} \simeq Y^\ast(\g)-\op{comod}.
$$
Further, by applying proposition \ref{proposition_equivalence_lambda} we see that there is an equivalence of monoidal dg categories
$$
\F_{0}((\lambda))-\op{mod} \simeq Y^\ast(\g)((\lambda))-\op{comod}
$$
where the category $Y^\ast(\g)((\lambda))-\op{comod}$ is defined to be the dg category consisting of certain lax $A_\infty$ comodules over $Y^\ast(\g)((\lambda))$.

After applying these equivalences, we see that the OPE map $m_{OPE}$ gives a monoidal functor which we continue to call $F_{OPE}$: 
$$
F_{OPE}(\lambda) : Y^\ast(\g)-\op{comod} \times Y^\ast(\g)-\op{comod} \to Y^\ast(\g) \what{\otimes} \C((\lambda))-\op{comod}. 
$$
\begin{theorem}
There is a natural equivalence of monoidal functors
$$
F_{OPE} \iso F_R
$$
where $F_R$ is the functor with monoidal structure constructed from the $R$-matrix in theorem \ref{theorem_monoidal}. 
\end{theorem}
\begin{remark}
By natural equivalence I mean a natural transformation which induces a quasi-isomorphism on all objects. 
\end{remark}
\begin{proof}
The first thing we need to verify is that there is such an equivalence just as functors (not as monoidal functors). 

Since $F_{OPE}$ is monoidal, we have
$$
F_{OPE}(V \times W) = F_{OPE}(V\times 1) \otimes F_{OPE}(1 \times W).
$$
Here $1$ denotes the trivial $Y(\g)$-module, that is, $\C[[\hbar]]$ with the trivial action. 

To check that there is an isomorphism of functors $F_{OPE} \iso F_R$, we need to verify that there is a natural quasi-isomorphism
$$
F_{OPE}(V \times 1) \otimes F_{OPE}(1 \times W) \iso T_0(V) \otimes T_\lambda(W).
$$
Thus, we need to construct natural quasi-isomorphisms
\begin{align*}
F_{OPE}(V \times 1) &\simeq T_0(V) \\
F_{OPE}(1 \times W) & \simeq T_\lambda(W). 
\end{align*}
Now $F_{OPE}$ was constructed from a homomorphism of $E_2$ algebras
$$
m'_{OPE} : \F_0 \otimes \F_0 \to \F_0((\lambda)).
$$
This homomorphism of $E_2$ algebras was the $\C((\lambda))$-linear extension of the identity when restricted to $\F_0 \otimes 1$: this makes it clear that $F_{OPE}(V \times 1) = T_0(V)$.

Note that infinitesimal translation in the $z$-direction is a derivation of $\F_0$.  In fact, it's a derivation of $\F_0$ as an $E_2$ algebra. Let us call this derivation $\til{\del}_z$.  Any derivation of $\F_0$ gives rise (by naturality) to a derivation of the Koszul dual Hopf algebra.   Let us use the notation $\til{\del}_z$ for the corresponding Hopf algebra derivation of $Y(\g)$. 

Let 
$$
\til{T}_\lambda : \F_0 \to \F_0[[\lambda]]
$$
be defined by $\exp (\lambda \del_z)$.   Thus, $\til{T}_\lambda$ is the one-parameter group of automorphisms associated to the derivation $\del_z$. 

By the construction of $m'_{OPE}$, the homomorphism of $E_2$ algebras $$m'_{OPE}(1 \otimes - ) : \F_0 \to \F_0((\lambda))$$
is simply $\til{T}_{-\lambda}$.  Thus, the functor 
$$
V \mapsto F_{OPE}(1 \times V)
$$
sends an $\F_0$-module $V$ to $\til{T}_{-\lambda}(V)$, which is defined by tensoring with the $\F_0 - \F_0$-bimodule $\F_0$ where the right action of an element $\beta \in \F_0$ is by right multiplication by $\til{T}_{-\lambda}(\beta)$. 

Let us use the notation $\til{T}_{-\lambda}$ for the one-parameter family of automorphisms of $Y(\g)$ and of $Y^\ast(\g)$ obtained by applying Koszul duality to the corresponding family of automorphisms of $\F_0$.  A small calculation shows that, under the equivalence between
$$
\F_0-\op{mod} \simeq Y^\ast(\g)-\op{comod} $$
the functor $\til{T}_{-\lambda}$ on $\F_0-\op{mod}$ corresponds to $\til{T}_{\lambda}$ on $Y^\ast(\g)-\op{comod}$: i.e. the functor on $Y^\ast(\g)-\op{comod}$ which sends a comodule $M$ to the module where the coaction of $Y(\g)$ is composed with $\til{T}_{-\lambda}$. 

Thus, in terms of functors on $Y^\ast(\g)-\op{comod}$, our discussion so far shows that
$$
F_{OPE}( 1 \times W ) \simeq \til{T}_{\lambda}(W).
$$
It remains to verify that $\til{T}_{\lambda}$ and $T_{\lambda}$ coincide. We will show the following.
\begin{lemma}
There is a unique homomorphism of Hopf algebras
$$
\Phi : Y(\g) \to Y(\g)[[\lambda]]
$$
with the following properties.
\begin{enumerate}
\item Classically, $\Phi$ arises from the map
\begin{align*}
\g[[z]] & \to \g[[z]][[\lambda]] \\
z &\mapsto z + \lambda.
\end{align*}
\item The map $\Phi$ is $\C^\times$-equivariant, where $\C^\times$ acts on the Yangian in the way we have discussed, and where $\lambda$ has weight $1$. 
\item Modulo $\lambda$, $\Phi$ is the identity. 
\end{enumerate}
\end{lemma}
\begin{proof}
We already know that $\Phi$ exists.  We need to verify that it is unique.  To prove this, we will work term by term in $\hbar$ and in $\lambda$.  If we have $\Phi$ modulo $\hbar^n$, then any two lifts of $\Phi$ to a homomorphism defined modulo $\hbar^{n+1}$ differ by a Hopf algebra derivation
$$
U(\g[[z]]) \to U(\g[[z]])[[\lambda]].
$$
Such a derivation is necessarily given by a derivation of Lie algebras 
$$
\Psi_n : \g[[z]] \to \g[[z]][[\lambda]].
$$
Modulo $\lambda$, $\Psi_n = 0$. 

Further, the homogeneity constraint implies that $\Psi_n$ must be of weight $-n$ with respect to the $\C^\times$ action (we can think of $\Psi_n$ as being accompanied by $\hbar^n$, and $\hbar$ has weight $1$). 

We will show $\Psi_n$ is unique by working order by order in $\lambda$. By considering the $\lambda^k$ term, for $k > 0$, we see that we need to show that any Lie algebra derivation
$$
\Gamma_{k+n} : \g[[z]] \to \g[[z]]
$$
which is of weight $\le -k-n$, is zero.  

Thus, let us fix $n > 1$, and let us consider a derivation $\Gamma_n$ of weight $-n$. We will show that $\Gamma_n = 0$. 

The map $\Gamma_n : z^k\g \to z^{k-n} \g$ must be a map of $\g$-modules.  Since $\g$ is simple, we see it it must be of the form $\Gamma_n(z^k X) = c_k z^{k-n} X$, for some sequence of constants $c_k$ with $c_k = 0$ for $k-n < 0$.

By considering $\Gamma_n( [z^k X, z^l Y] )$ and applying the derivation rule, we see that $c_k + c_l = c_{k+l}$. This implies $\Gamma_n = c z^{-n+1} \dpa{z}$ for some constant $c$.  If $n > 1$ this is not a well-defined derivation of $\C[[z]]$. 

\end{proof}

This lemma allows us to conclude that 
$$
F_{OPE}(V \otimes W) = T_0(V) \otimes T_{-\lambda}(W). 
$$
Thus, as functors, $F_{OPE}$ coincides with the functor $F_R(-\lambda)$ given by the $R(-\lambda)$.  According to \ref{theorem_monoidal_dual}, $F_{OPE}$ must be given by \emph{some} $R$-matrix satisfying $R$-matrix equations $1$, $2$ and $3$.  Uniqueness of the $R$-matrix tells us that $F_R$ and $F_{OPE}$ coincide as monoidal functors. 
\end{proof}

\section{Expectation value of Wilson operators and lattice models}
Detailed proofs have been given of almost all the statements made in the introduction. The exception is the statement relating the expectation values of Wilson operators and integrable lattice models.

Consider our theory on $\mbb{P}^1_z \times E_w$.  Recall that our theory can be defined at the quantum level on any complex surface equipped with a meromorphic volume form with quadratic poles along a divisor and no other poles or zeroes.  The meromorphic volume form here is $\d z \d w$.  

Recall also that 
$$
H^\ast ( \Obs(\mbb{P}^1_z \times E_w) ) = \C[[\hbar]].
$$
This means that we can define correlation functions.  Indeed, given any collection of disjoint open subsets $U_1,\dots, U_n \subset \mbb{P}^1_z \times E_w$, the structure map of the factorization algebra is a map
$$
m_{U_1,\dots,U_n}^{\mbb{P}^1_z \times E_w} : \Obs(U_1) \times \dots \times \Obs(U_n) \to \Obs(\mbb{P}^1) \times E_w).
$$ 
If $\alpha_i \in \Obs(U_i)$ are observables, which are closed
 and of cohomological degree $0$, then we define the expectation value or correlation function of the $\alpha_i$ by
$$
\ip{\alpha_1 , \dots, \alpha_n} = [m_{U_1,\dots, U_n}^{\mbb{P}^1} ] \in H^0 (\Obs(\mbb{P}^1 \times E)) = \C[[\hbar]]. 
$$
This expectation value has the usual properties: for example, if some $\alpha_i$ is cohomological exact, whereas the rest are closed, then the expectation value is zero. 

Let $V$ be a representation of the Yangian.   Let $a_1,\dots, a_m$ and $b_1,\dots,b_n$ be $a$- and $b$-cycles on $E$, such that the $a_i$ are disjoint from each other, as are the $b_i$. We have Wilson operators 
\begin{align*}
\chi_V(z, S^1_{a_i}) &\in H^0 ( \Obs( z \times S^1_{a_i} ) ) \\
\chi_V(0, S^1_{b_j}) &\in H^0 ( \Obs( 0 \times S^1_{b_j} ) ).
\end{align*}
Then, we can consider the expectation value
$$
\ip{\chi_V(z, S^1_{a_1}  ), \dots, \chi_V(z, S^1_{a_m}  ), \chi_V(0, S^1_{b_1}  ) , \dots, \chi_V(0, S^1_{b_n}  ) } \in \C[[\hbar]].
$$
We want to show that this is equal to the partition function of an integrable lattice model.   The partition function of this model is the trace of the $R$-matrix $R_{V^{\otimes m}, V^{\otimes n}}$ acting on $V^{\otimes m} \otimes V^{\otimes n}$.

\begin{theorem}
There is an equality
$$
\ip{\chi_V(z, S^1_{a_1}), \dots, \chi_V(0, S^1_{a_m}), \chi_V(0, S^1_{b})   } = \op{Tr}_{V^{\otimes m} \otimes V^{\otimes n}} R_{V^{\otimes m}, V^{\otimes n}}(z).
$$
\end{theorem}
\begin{remark}
The sign change on the right hand side comes from the sign change we discussed above between the natural conventions for the OPE and $R$-matrix. 
\end{remark}
\begin{proof}
We know that the operator product of Wilson operators corresponds to tensor product of the Yangian.  By colliding the $a$- and $b$-cycles together, we see that it suffices to show that
$$
\ip{\chi_{V^{\otimes m}}(z, S^1_{a}), \chi_{V^{\otimes n}}(0, S^1_{b})} = \op{Tr}_{V \otimes W} R(z).
$$
Equivalently, by translation invariance, we need to show that
$$
\ip{\chi_{V^{\otimes m}}(0, S^1_{a}), \chi_{V^{\otimes n}}(-z, S^1_{b})} = \op{Tr}_{V \otimes W} R(z).
$$
We will prove that this holds if we replace $V^{\otimes m}$ and $V^{\otimes n}$ by two arbitrary representations of the Yangian, which we call $W_1$ and $W_2$.  

As I sketched in the introduction, the proof relies on the Ayala-Rozenblyum \cite{AyaRoz13} model for the iterated Hochschild homology of a monoidal category. They define a theory they call \emph{composition homology} for an $(\infty,n)$-category with adjoints\footnote{Their theory is more general than this, but the case of categories with adjoints is what is relevant here.} (i.e.\ every $k$-morphism for $0 < k < n$ has an adjoint).  If $M$ is a compact framed manifold, and $\mc{C}$ is an $(\infty,n)$-category with adjoints, the composition homology is defined as a certain colimit over configurations consisting of stratified subsets $S \subset M$, with certain framing data, and labels on $S$ by objects and morphisms of $\mc{C}$.  More precisely, every connected component of $M \setminus S$ is labeled by an object of $\mc{C}$. Every codimension $1$ stratum of $S$ is labeled by a $1$-morphism of $\mc{C}$, in a direction determined by the framing data. (The framing data is, in particular, a trivialization of the Gauss map on each stratum).  These configurations are allowed to collide in various ways, as part of the colimit process.  The details will not concern us, however; we only need to deal with individual configurations. 

We are interested in the following special case.  The category 
$$\op{Fin}(Y(\g)) \simeq \op{Fin}(Y^\ast(\g)) \simeq \op{Perf}( \Obs_{0}) $$
is a monoidal dg category with duals, enriched in the category $\mc{C}^b$ of complete filtered $\C[[\hbar]]$-modules.  Let $\op{Fin}^{(2)}(Y(\g))$ be the $(\infty,2)$-category which as one object and whose $1$-morphisms are objects of $\op{Fin}(Y(\g))$. This $(\infty,2)$-category is enriched at the level of $2$-morphisms; the Ayala-Rozenblyum theory can deal with enriched categories. 

Then, the Ayala-Rozenblyum composition homology, applied to $\op{Fin}^{(2)}(Y(\g))$, yields the iterated Hochschild homology of $\op{Fin}(Y(\g))$.   We let $HH^{(2)}$ denote composition homology on the two-torus.

Let us denote composition homology of a the two-torus with coefficients in some two-category by $CH_\ast$.  Suppose we have the following data:
\begin{enumerate}
\item Two objects $V_1,V_2$ of $\op{Fin}(Y(\g))$.
\item A map
$$
f : V_1 \otimes V_2 \to V_2 \otimes V_1.
$$
\end{enumerate}
Then, we can construct an element of composition homology  $HH_\ast^{(2)}(\op{Fin}^{(2)}(Y(\g))$ as follows.  The stratified submanifold of the two-torus is 
$$
S^1_a \cup S^1_b \subset S^1_a \times S^1_b. 
$$
The stratification is the evident one: there are three closed strata, namely the circles and their intersections.  We label $S^1_a$  by $V_1$ and $S^1_b$ by $V_2$, both viewed as $1$-morphisms in $\op{Fin}^{(2)}(Y(\g))$.  The intersection point is labeled by $f$, which is a $2$-morphism in $\op{Fin}^{(2)}(Y(\g))$ making the following square commutative:
$$
\xymatrix{
1 \ar[r]^{V_1} \ar[d]_{V_2} & 1 \ar[d]^{V_2} \\
1 \ar[r]_{V_1} & 1
}
$$
where $1$ indicates the unique object of $\op{Fin}^{(2)}(Y(\g))$.   

The Wilson operator associated to the representation $W_1$ on $S^1_a$ is associated to such a configuration, where we label the $a$-cycle by $W_1$, the $b$-cycle by the identity object $1$ (i.e. $\C[[\hbar]]$ with the trivial action of the Yangian), and the intertwiner $f$ is the identity map
$$
f : 1 \otimes W_1 \to W_1 \otimes 1.
$$
The Wilson operator associated to the representation $W_2$ on the $b$-cycle is given by labeling $S^1_b$ by $W_2$, $S^1_a$ by the identity object, and the intertwiner again being the identity. 

We have a functor
$$
F_R : \op{Fin}(Y(\g)) \times \op{Fin}(Y(\g)) \to \op{Fin}(Y(\g)((\lambda)))
$$
associated to the operator product in the $z$-plane.  This induces a map on composition homology:
$$
HH^{(2)} (F_R) : HH^{(2)}( \op{Fin}^{(2)}(Y(\g)) ) \otimes HH^{(2)}( \op{Fin}^{(2)}(Y(\g)) ) \to HH^{(2)}(\op{Fin}(Y(\g)) ((\lambda)).
$$
We will apply $HH^{(2)} (F_R)$ to Wilson operators on the $a$- and $b$-cycles described above. 

What we find is the following configuration: on $S^1_a$, we put $T_0(W_1)$, on $S^1_b$ we put $T_{-\lambda}(W_2)$, and on the intersection we put
$$
\sigma \circ R_{W_1,W_2}(\lambda) : T_0(W_1) \otimes T_{\lambda}(W_2) \to T_{\lambda}(W_2) \otimes T_0(W_1). 
$$
We have seen that $\sigma \circ R_{W_1,W_2}$ is an isomorphism of $Y(\g)((\lambda))$-modules.  
 
Recall that
$$
H^\ast (\Obs (0 \times E)) ((\lambda)) =  HH^{(2)} (\op{Fin}^{(2)}(Y(\g)) ) ((\lambda)).
$$
Taking expectation values is the map
$$
H^\ast (\Obs (0 \times E)) ((\lambda)) \to H^\ast (\Obs (\mbb{P}^1 \times E)) ((\lambda))
$$
coming from the inclusion $0 \times E \into \mbb{P}^1 \times E$. 

We need to calculate the image of our element in $HH^{(2)} (\op{Fin}^{(2)}(Y(\g)) ) ((\lambda))$ under this map.

The map of $E_2$ algebras
$$
\Obs_{0} \to \Obs{\mbb{P}^1} \simeq \C[[\hbar]]
$$
is the augmentation map of the $E_2$ algebra $\Obs_0$.  Any homomorphism of $E_2$ algebras induces a monoidal functor; the corresponding monoidal functor
$$
\op{Perf}( \Obs_0) \simeq \op{Fin}(Y(\g)) \to \op{Fin}(\C[[\hbar]])
$$
is the forgetful functor, which takes a $Y(\g)$-module and views it as a $\C[[\hbar]]$-module.

The expectation value map arises by applying iterated Hochschild homology to this monoidal functor.  Since Ayala-Rozenblyum's composition homology is functorial, and agrees with iterated Hochschild homology in this case, it remains to perform the following calculation in composition homology of $\op{Fin}(\C[[\hbar]])$.  Consider the element in $HH^{(2)} (\op{Fin}^{(2)}(\C[[\hbar]])((\lambda))$ obtained by labeling the $a$-cycle by $W_1$, the $b$-cycle by $W_2$, and the intersection by the isomorphism
$$
\sigma \circ R_{W_1,W_2}(\lambda) : W_1 \otimes W_2 \to W_2 \otimes W_1.
$$
We need to check that this element is equivalent to $\op{Tr}_{W_1 \otimes W_2} R_{W_1, W_2}(\lambda)$. The argument is general: we need to show that if $F$ is any endomorphism of $W_1 \otimes W_2$, and we put $\sigma \circ F$ on the intersection, we find $\op{Tr} F$.  Thus, we can assume, without loss of generality, that $F$ is decomposable: $F = F' \otimes F''$ where $F'$, $F''$ are endomorphisms of $W_1$ and $W_2$ respectively.

Now, since $\op{Fin}(\C[[\hbar]])$ is a symmetric monoidal category, $HH^{(2)} (\op{Fin}^{(2)} (\C[[\hbar]])$ has a commutative product.  At the level of configurations on the torus $E$ labeled by objects and morphisms of $\op{Fin}^{(2)} (\C[[\hbar]])$, this tensor product is given by applying the symmetric monoidal structure to objects, morphisms, and two-morphisms.

The configuration where we label the $a$-cycle by $W_1$, the $b$-cycle by $W_2$ and the intersection by $\sigma \circ F$ can be realized as the tensor product (in the above sense) of two configurations. The first is where we label the $a$-cycle by $W_1$, the $b$-cycle by the identity object $1$ of $\op{Fin}(Y(\g))$, and the intersection by the endomorphism $F'$ which we view as an isomorphism
$$
W_1 \otimes 1 \iso 1 \otimes W_1.
$$
The second configuration labels the $a$-cycle by the identity object, the $b$-cycle by $W_2$, and the intersection by $F''$ viewed as an isomorphism
$$
1 \otimes W_2 \iso W_2 \otimes 1. 
$$

The configurations is in the image of composition homology along the $a$-cycle and and $b$-cycle, respectively.  The first configuration is easily seen to give $\op{Tr}_{W_1}(F')$, whereas the second gives $\op{Tr}_{W_2} F''$.   Therefore the answer is the product
$$
(\op{Tr}_{W_1} F' ) (\op{Tr}_{W_2} F'' ) = \op{Tr}_{W_1 \otimes W_2} (F' \otimes F'')
$$
as desired.

\end{proof}

\section{Normalization of the coupling constant}
In the introduction, I stated that to find the standard normalization of the spectral parameter in the theory of integrable systems, we need to normalize the Chern-Simons term in our theory as $\tfrac{1}{2 \pi i} \d z CS(A)$.  What we've proved so far tells us that we find some non-trivial homogeneous quantization of the Yangian, and the corresponding $R$-matrix.  There is precisely one quantization up to  change of coordinates of the form $z \to c z$ where $c$ is a constant. Thus, we need to normalize so that the coordinate $z$ matches the spectral parameter. This normalization amounts to a normalization of the one-loop $R$-matrix.

We will do this by a four-loop calculation. (It will be a very simple $4$-loop calculation, involving only the propagator). It turns out that, if $\g$ is simple, the one-loop contribution to the expectation value of any configuration of Wilson operators is zero. This is simply because the trace of any element of $\g$ in any representation is zero.
 
However, the one-loop $R$-matrix contributions to higher loop terms; we will isolate a certain $4$-loop expectation value which will be able to normalize the one-loop $R$-matrix.

Thus, let $\g$ be a simple Lie algebra. Let $V$ be a representation of $\g$ which lifts to a representation of $Y(\g)$. Let us assume that the invariant pairing on $\g$ defined by $\op{Tr}_V (X Y)$ is non-zero (and so non-degenerate).

Let $\chi_V$ be the Wilson operator corresponding to $V$. Let
$$
\chi'_V = \chi_V - \op{dim} V = \chi_V - \ip{\chi_V} .
$$ 
We are going to calculate the leading order contribution to
\begin{multline*}
C\left(\chi'_V(z, S^1_{a_1}), \chi'_V(z, S^1_{a_2}) , \chi'_V(0,S^1_{b_1}), \chi'_V(0,S^1_{b_2})\right) \\ 
\defeq \ip{\chi'_V(z, S^1_{a_1}), \chi'_V(z, S^1_{a_2}) , \chi'_V(0,S^1_{b_1}), \chi'_V(0,S^1_{b_2})} \\ 
-  2 \ip{\chi'_V(z, S^1_{a}), \chi'_V(0, S^1_{b}) }^2
\end{multline*}
(This is the cumulant of the four $\chi'_V$ operators we are considering). 

\begin{lemma}
We have
\label{lemma_Rmatrix_expectation_value}
\begin{multline*}
C\left(\chi'_V(z, S^1_{a_1}), \chi'_V(z, S^1_{a_2}) , \chi'_V(0,S^1_{b_1}), \chi'_V(0,S^1_{b_2})\right)
= \hbar^4 z^{-4} \alpha^{-4} (2 \pi i)^4 \left(\sum_{i,j} \op{Tr}_V c_i c_j \op{Tr}_V c^i c^j \right)^2 + O(\hbar^5)
\end{multline*}
where 
$$c = \sum c_i \otimes c^i \in \g \otimes \g$$ is the quadratic Casimir, dual to the chosen invariant pairing.  
\end{lemma} 
We will prove this lemma shortly.

Our main theorem tells us that this must be the same as a certain term in the partition function of the integrable lattice model associated to $V$ and $\g$, up to normalization of the spectral parameter $z$.   
\begin{lemma}
Modulo $\hbar^5$, we can compute the cumulant 
$$
C\left(\chi'_V(z, S^1_{a_1}), \chi'_V(z, S^1_{a_2}) , \chi'_V(0,S^1_{b_1}), \chi'_V(0,S^1_{b_2})\right)
$$
as the partition function of a $4 \times 4$ toroidal lattice, where a state as usual is a way of labelling an edge by a basis element of $V$, but the interaction at a vertex is given by 
$$R'_{V \otimes V} = R_{V \otimes V} - \op{Id}_{V \otimes V}.$$
\end{lemma}
\begin{proof}
Our main theorem gives an equality between expectation values of Wilson operators and the partition function of the lattice model with interaction $R_{V \otimes V}$. The cumulant is a linear combination of products of various Wilson operators, and so can be translated into a sum of products of partition functions of the integrable model.  A small combinatorial calculation gives the result.
\end{proof}

Using the standard integrable systems normalization of the spectral parameter, the partition function of this $4 \times 4$ lattice model with interaction $R'$ is easily computed to be
$$
\hbar^4 z^{-4} \left(\sum_{i,j} \op{Tr}_V c_i c_j \op{Tr}_V c^i c^j \right)^2 + O(\hbar^5).
$$
The point is that
$$
R_{V \otimes V} = 1 + \frac{\hbar}{z} c_{V \otimes V} + O(\hbar^2)
$$
where $c_{V \otimes V}$ is the quadratic Casimir acting on $V \otimes V$. 

Thus, to match the normalization with that seen in the theory of integrable systems, we will need to set the coupling constant $\alpha$ to be $\alpha = \tfrac{1}{2 \pi i}$. 

Let us now turn to the proof of the lemma about calculating the cumulant. 
\begin{proof}
Let 
$$
\E = \Omega^{0,\ast}(\mbb{P}^1, \Oo(-1)) \what{\otimes} \Omega^\ast(E)
$$
and let $\br{\E}$ be the distributional completion of $\E$. The space of fields of our theory is $\E \otimes \g$. 

We will calculate the expectation value using the propagator. 
The propagator is of the form
$$
P = c \otimes P_0 \in \g \otimes \g \otimes \br{\E} \what{\otimes} \br{\E} 
$$
where $c \in \g \otimes \g$ is the quadratic Casimir for $\g$. 

The propagator is defined with respect to a gauge fixing condition, which we choose to be $\dbar^\ast_z + \d^\ast_w$.  The propagator $P_0$ is uniquely characterized by the following properties.
\begin{enumerate}
\item $P_0$ is symmetric and of cohomological degree $0$. Thus, $P_0$ is a $2$-form on $\mbb{P}^1 \times E \times \mbb{P}^1 \times E$. 
\item 
$$(\dbar^\ast_{z_1}+ \d^\ast_{w_1}) P_0 = 0,$$
where $z_1,z_2$ and $w_1,w_2$ are coordinates in $\mbb{P}^1_{z_1} \times E_{w_1} \times \mbb{P}^1_{z_1} \times E_{w_2}$. 
\item 
$$
\alpha ( \d z_1 - \d z_2)  ( \dbar_{z_1} + \dbar_{z_2} + \d_{w_1} + \d_{w_2} ) P_0 = - \delta_{\op{Diag}}  
$$
is minus the delta function on the diagonal, viewed as a de Rham current on $\mbb{P}^1 \times E \times \mbb{P}^1 \times E$.  Here $\alpha$ is the coupling constant of our theory, so the action functional is $\alpha \d z CS(A)$. (This sign is a little tricky to get right, but it's not important: change of orientation of the cycles on which we're putting the Wilson loop changes the sign. Our choice of orientation on $E$ is such that the $S^1_a$ and $S^1_b$ intersect positively). 
\end{enumerate}
It is easy to see that, in the expectation value we are computing, only the propagator (and not the interaction terms) appears. Further, the expectation value only depends on the classical Wilson operator.

 We will compute this expectation value by using the path-ordered exponential formulation of the Wilson operator: 
$$
\chi_V(z,S^1_a) (A) = \op{dim} V + \sum_{n = 1}^{\op \infty} \tfrac{1}{n} \op{Tr}_V  \int_{\theta_1 <  \dots < \theta_n \in S^1_a} A_{\theta_1} \wedge \cdots \wedge A_{\theta_n}.
$$
Saying $\theta_1 < \dots < \theta_n$ indicates that these points are cyclically ordered.  The factor of $\tfrac{1}{n}$ appears because we are using cyclically ordered configurations in a circle instead of totally ordered configurations in a line.   $A$ is the partial-connection field of our theory, viewed as a one-form. $A_{\theta_i}$ refers to the one-form on the configuration space of the $\theta_i \in S^1_a$ obtained by pulling back $A$ via the map from this configuration space to $z \times E$, sending such a configuration to $\theta_i$. 

We are interested in the expectation value not of the Wilson operator $\chi_V$ but of  
$$\chi'_V = \chi_V -\op{dim} V.$$
This has the effect of removing the constant term ($n= 0$) from the path-ordered exponentation.  

Further, since the trace of any element of $\g$ in the representation $V$ is zero, we find
$$
\chi'_V(z,S^1_a) (A) = \sum_{n = 2}^{\op \infty} \frac{1}{n} \op{Tr}_V  \int_{\theta_1 <  \dots < \theta_n \in S^1_a} A_{\theta_1} \cdots A_{\theta_n}.
$$
Only the $n = 2$ term will contribute to the $4$-loop expectation value. 

Let 
$$
\Phi(z) = \int_{\substack{\theta \in z \times S^1_a \\ \gamma \in 0 \times S^1_b }}  P_0 (\theta,\gamma) .
$$
Note that 
$$
\int_{\substack{(\theta_1,\theta_2) \in z \times S^1_a \\ (\gamma_1,\gamma_2) \in 0 \times S^1_b }}  P_0 (\theta_1, \gamma_1) P_0(\theta_2,\gamma_2) \\
= \Phi(z)^2. 
$$

Standard Feynman diagrammatics tell us that the $4$-loop expectation value we are interested in is
\begin{multline*}
C\left( \chi'_V(z, S^1_{a_1}), \chi'_V(z, S^1_{a_2}) ,  \chi'_V(0,S^1_{b_1}), \chi'_V(0,S^1_{b_2}) \right) \\
= \hbar^4 \Phi(z)^4 \left(\sum_{i,j} \op{Tr}_V c_i c_j \op{Tr}_V c^i c^j \right)^2 + O(\hbar^5).
\end{multline*}
Indeed, we are computing the sum over diagrams which consist of $4$ propagators, where each propagator links one modified Wilson line on an $a$-cycle to one on a $b$-cycle. Diagrams which contain an interaction must have the property that, modulo $\hbar^5$, some Wilson line has at most one propagator attached to it, and so gives zero. 

Further, because we are using the cumulant $C$ instead of the expectation  value $\ip{-,-}$ we only count connected diagrams.  The only non-zero contributions to order $\hbar^4$ occur when each modified Wilson line is connected to precisely two others, each by a single propagator. Because we are using a total of 4 propagators we find $\Phi(z)^4$. There are $4$ diagrams that contribute

Thus, to calculate our expectation value, we need to calculate $\Phi(z)$. 

We can compute $\d z \dbar\Phi(z)$ using properties of $P_0$. Let $D$ be a disc around $0$ in $\mbb{P}^1$.  We have
\begin{align*}
\int_{z \in D} \d z \dbar \Phi(z) &=  \int_{(z_1,w_1,z_2,w_2) \in (D \times S^1_a \times 0 \times S^1_b)}  \d z_1  \dbar_{z_1} P_0 \\
&=  \int_{(z_1,w_1,z_2,w_2) \in (D \times S^1_a \times 0\times S^1_b)}  ( \d z_1 - \d z_2)  ( \dbar_{z_1} + \dbar_{z_2} + \d_{w_1} + \d_{w_2} ) P_0 \\
&= - \int_{(z_1,w_1,z_2,w_2) \in (D \times S^1_a \times 0 \times S^1_b)}\alpha^{-1} \delta_{Diag}
\end{align*}
where we have used one of the equations defining $P_0$. 

Since $D \times S^1_a$ intersects $0 \times S^1_b$ transversely in a single point, the integral of the delta-current for the diagonal over this region is one.  Therefore,
$$
\d z \dbar \Phi(z) =  -\alpha^{-1} \delta_{0}
$$
is  the delta-current for $0 \in \C$.  Since
$$
\d z \dbar z^{-1} = - 2 \pi i \delta_{0}
$$
and since $\Phi(z)$ is holomorphic away from $0$, and tends to zero as $z \to \infty$, we must have
we have
$$
\Phi(z) = \alpha^{-1} \frac{1}{2 \pi i} z^{-1}.
$$
This proves the lemma.
\end{proof}

Next, we will show derive from this lemma another proof of proposition \ref{proposition_bialgebra}.

\begin{proposition}
The Hopf algebra structure Koszul dual to the $E_2$ algebra $\Obs_{z_0}$ is described, to leading order, by the Lie bialgebra structure on the Yangian. 
\label{proposition_bialgebra_alternative}
\end{proposition}
\begin{proof}
Without assuming the statement of the proposition we are proving,  the general formalism we have developed shows that the Koszul dual Hopf algebra to $\Obs_{z_0}$ is a topological Hopf algebra quantizing $U(\g[[z]])$.  Let us call this Hopf algebra $\til{Y}(\g)$. This Hopf algebra has a filtration where $\hbar$ has filtered degree $2$	and $\g[[z]]$ has filtered degree $-1$.    Further, the general formalism implies that this Hopf algebra (whatever it is) has an $R$-matrix with spectral parameter
$$
R(\lambda) \in \til{Y}(\g) \otimes_{\C[[\hbar]]} \til{Y}(\g) ((\lambda))
$$
satisfying $R$-matrix equations $1$, $2$ and $3$ in theorem \ref{theorem_existence_R_matrix}. This $R$-matrix is in the $0$'th filtered piece, and is $\C^\times$-invariant, where $z$, $\hbar$ and $\lambda$ all have filtered weight $1$. Further, modulo $\hbar$, the $R$-matrix is easily seen to be $1$.

Thus, there is an element 
$$
R_1 \in U(\g[[z]]) \otimes U(\g[[z]]) ((\lambda))
$$
with the property that
$$
R - 1 = \hbar R_1 \op{mod} \hbar^2.
$$
The fact that $R$ is in the $0$'th filtered piece implies that 
$$
R_1 \in \g[[z]] \otimes \g[[z]] ((\lambda)).
$$
Let us assume that $\g$ is simple.  Then the $\C^\times$-equivariance property, together with $G$-invariance, implies that $R_1$ is of the form 
$$
R_1(z_1,z_2,\lambda) = \sum_{i,j \ge 0} \beta_{i,j}\lambda^{-i -j-1} z_1^i z_2^j c \in \g[[z_1]] \otimes \g[[z_2]]
$$
where $c \in \g \otimes \g$ is the quadratic Casimir. 

The general formalism we have developed tells us that the one-loop expectation value of Wilson operators on $\mbb{P}^1 \times E$ must be the trace of the $R$-matrix. If we normalize our action functional to be $\frac{1}{2 \pi i} \int \d z CS(A)$, the expectation value calculated above normalizes the coefficient of $\lambda^{-1}$ in $R_1$ to be $\lambda^{-1} c$. 

It follows from translation invariance of our field theory that 
\begin{align*}
\dpa{z_1} R + \dpa{z_2} R_1 &= 0 \\
\dpa{z_2} R + \dpa{\lambda} R_1 &= 0.
\end{align*}
(We can also see this directly by considering what translation invariance applies about the $4$-loop, $4$-Wilson operator expectation value computed above). 
 
The only solution to these equations, with initial condition the coefficient of $\lambda^{-1}$ in $R_1$, is
$$
R = c / (z_1 - z_2 + \lambda) = c \sum \lambda^{-i-1}(z_1 - z_2)^i.
$$

Let us now consider $R$-matrix equation $1$, which will relate the cocommutator of the Hopf algebra $\til{Y}(\g)$ to the $R$-matrix.  We have
$$
(T_\lambda \otimes T_0) \tr^{op} (\alpha)  = R(\lambda)((T_\lambda \otimes T_0) \tr(\alpha)) R(\lambda)^{-1} \in Y(\g) \otimes Y(\g) ((\lambda)),
$$
for all $\alpha \in \til{Y}(\g)$, where $\tr$ is the coproduct and $\tr^{op}$ is the opposite coproduct.  Let's work to leading order in in $\hbar$, and apply this to an element $\alpha \in \g[[z]]$, lifted arbitrarily modulo $\hbar^2$.  We find
$$
(T_\lambda \otimes T_0) ( \tr^{op}(\alpha) - \tr(\alpha) )  = \hbar [R_1,  (T_\lambda \otimes T_0) (\alpha \otimes 1 + 1 \otimes \alpha) ] \op{mod} \hbar^2. 
$$
Let $\delta(\alpha)$ denote the Lie bialgebra structure coming from the leading order cocommutator.  We have
$$
\tr^{op}(\alpha) - \tr(\alpha) =  - \hbar \delta(\alpha).
$$
Plugging this in, together with our calculation of $R_1$, gives us
$$
 \delta(\alpha)(z_1 + \lambda, z_2)   =  - \left[\frac{c}{z_1 - z_2 + \lambda} ,  (\alpha(z_1 + \lambda) \otimes 1 + 1 \otimes \alpha(z_2)) \right] .
$$
Here, we are viewing $\alpha$ as a map from the $z$-plane to $\g$.  

This implies
$$
\delta(\alpha) = -\left[ \frac{c}{z_1 - z_2}, \alpha \otimes 1 + 1 \otimes \alpha\right] 
$$
which is the standard expression for the Lie coalgebra structure on the Yangian; see \cite{Dri87}.
\end{proof}

\section{Appendix: Existence of quantum theories}
\label{appendix_existence}
In this appendix, we will prove that the twisted and deformed $N=1$ supersymmetric gauge theories can be quantized on any Calabi-Yau surface.   The precise statement will be that there exists a unique quantization compatible with certain natural symmetries.  Thus, we need to start by explaining what symmetries we will consider. 

First, let us consider the twisted $N=1$ supersymmetric gauge theory, i.e.\ the holomorphic BF theory on a complex surface $X$. We will consider the most general version we need, defined on a complex surface $X$ with a reduced divisor $D \subset X$ and a meromorphic volume form $\omega \in K_X(2 D)$ with quadratic poles on $D$ and no other poles or zeroes.  The case when $D$ is empty is allowed; then $\omega$ is just a holomorphic volume form. 

Let $P$ be a holomorphic $G$-bundle on $X$ trivialized along $D$. As we have seen, the Lie algebra describing the theory is 
$$\L = \Omega^{0,\ast}(X \otimes \g_P(-D))[\eps],$$ where $\eps$ is a parameter of degree $1$. 

Let $\op{At}_{P}$ be the Atiyah algebra of $P$. This is the sheaf on $X$ of infinitesimal symmetries of the pair $(X,P)$: equivalently, it is the sheaf on $X$ of $G$-equivariant holomorphic vector fields on $P$.  There is a short exact sequence
$$
0 \to \g_P \to \op{At}_P \to \op{Vect}_X \to 0
$$
of sheaves of Lie algebras, where $\op{Vect}_X$ is the sheaf of holomorphic vector fields on $X$. 

Let $\op{At}_{P,D} \subset \op{At}_P$ be the sub-sheaf consisting of those symmetries of $(X,P)$ which preserve the divisor $D$ and the trivialization of $P$ on $D$.  Thus, $\op{At}_{P,D}$ is the sheaf on $X$ of $G$-equivariant holomorphic vector fields on the total space $P$, which preserve the divisor $p^{-1} D = G \times D$, and which on this divisor come from vector fields on $D$.   There is a short exact sequence
$$
0 \to \g_P(-D) \to \op{At}_{P,D} \to \op{Vect}_{X,D} \to 0
$$
where $\op{Vect}_{X,D}$ is the sheaf of holomorphic vector fields on $X$ which preserve $D$. (Sometimes this is called the sheaf of logarithmic vector fields). 

Finally, let
$$
\op{At}^0_{P,D} \subset \op{At}_{P,D}
$$
be the subsheaf consisting of those symmetries which preserve the meromorphic volume form on $X$.  Thus, if $\op{Vect}^0_{X,D}  \subset \op{Vect}_{X,D}$ is the sheaf of divergence-free holomorphic vector fields on $X$ preserving $D$, there is a short exact sequence
$$
0 \to \g_P(-D) \to \op{At}^0_P  \to \op{Vect}^0_{X,D} \to 0. 
$$
We let 
$$
\mf{h} = H^0 (X, \op{At}^0_P)[\eps].
$$
This Lie algebra acts on $\L = \Omega^{0,\ast}(X \otimes \g_P(-D))[\eps]$ in an evident way, by dg Lie algebra derivations, preserving the invariant pairing on $\L$. 

The odd Lie algebra spanned by $\dpa{\eps}$ also acts on $\L$ in the obvious way.  Again, this action is by Lie algebra derivations, commutes with the differential, and is compatible with the invariant pairing, and so is an action by symmetries of the theory.    Putting these actions together, we have an action of the semi-direct product Lie algebra $\C \cdot \dpa{\eps} \oplus \mf{h}$ on $\L$. (Of course, $\dpa{\eps}$ is in cohomological degree $-1$. 

There is also a $\C^\times$ action on $\L$ given by rescaling the parameter $\eps$.  This action is by automorphisms of the dg Lie algebra $\L$, but it scales the invariant pairing.  As explained in \cite{Cos11a}, it makes sense to ask for $\C^\times$-equivariant quantizations, as long as we give the parameter $\hbar$ weight $-1$ under the $\C^\times$-action.  This way, the quantity $S / \hbar$ is $\C^\times$-invariant (where $S$ is the action functional).   

This $\C^\times$ also acts on the semi-direct product Lie algebra $\C \cdot \dpa{\eps} \oplus \mf{h}$ in a way compatible with the action of everything on $\L$.

\begin{theorem}
\label{theorem_existence_symmetries}
The twisted $N=1$ theory admits a unique quantization, compatible with the action of $\C \cdot \dpa{\eps} \oplus \mf{h}$ and $\C^\times$, on any Calabi-Yau surface $X$; where we perturb around any holomorphic $G$-bundle on $X$. 
\end{theorem}
This result will be strong enough to prove the existence of the twist of the deformation of the $N=1$ theory also. 

\begin{proof}

The obstruction-deformation complex for quantizing our theory has a slick description in terms of $D_X$-modules.  In what follows, we will always work with smooth (as opposed to holomorphic) differential operators on $X$, and we will denote this algebra by $D_X$.  Thus, a $D_X$-module is a (possibly infinite-dimensional) smooth vector bundle on $X$ with a flat connection.   Tensor products of $D_X$-modules are always tensor products over the sheaf $\cinfty_X$ of smooth functions on $X$; this corresponds to ordinary tensor product of vector bundles. 

The jets $J(\L)$ of $\L$ are a $D_X$ dg Lie algebra.First, we form the complex $C^\ast_{red}(J(\L))$ of reduced Chevalley-Eilenberg cochains of $J(\L)$, taken in the category of $D_X$-modules.  Thus $C^\ast_{red}(J(\L))$ is again a  $D_X$-module.   Then, we form the sheaf $\omega_X \otimes_{D_X}^{\mbb{L}} C^\ast_{red}(J(\L))$ (where $\omega_X = \Omega^{top}_X$ is the right $D_X$-modules of top forms on $X$).  Finally, we take invariants with respect to the symmetries we are interested in.    

We thus need to compute the $\C^\times$-invariants of 
$$
\hbar C^\ast \left( \C \cdot \partial_\eps \oplus \mf{h} ,  \omega_X \otimes_{D_X}^{\mbb{L}}  C^\ast_{red} (J(\L))  \right) [[\hbar]] . 
$$
We can rewrite this as the $\C^\times$-invariants of 
$$
\omega_X \otimes_{D_X}^{\mbb{L}} C^\ast \left( \C \cdot \partial_\eps \oplus \mf{h} ,   C^\ast_{red} (J(\L))  \right)  \otimes \hbar \C[[\hbar]]
$$
Note that the $D_X$-modules $J(\L)$ is quasi-isomorphic to $\op{J}(\g_P(-D))$, the $D$-module of jets of holomorphic sections of $g_P(-D)$.  

Away from the divisor $D$,the $D_X$-module $ C^\ast \left( \C \cdot \partial_\eps \oplus \mf{h},   C^\ast_{red} (J(\L))  \right) $ has stalks quasi-isomorphic to
$$
C^\ast \left(\C \cdot \partial_\eps \oplus \mf{h}, C^\ast_{red} ( \g[[z_1,z_2, \eps]] )  \right). 
$$
Near the divisor, if $f$ is a function generating the ideal $\Oo(-D)$, then the stalk is quasi-isomorphic to 
$$
C^\ast \left(\C \cdot \partial_\eps \oplus \mf{h}, C^\ast_{red} ( f \g[[z_1,z_2, \eps]] )  \right). 
$$

Since $\hbar \C[[\hbar]]$ is concentrated in negative $\C^\times$-weights, we need only show that the positive $\C^\times$ weight spaces of $C^\ast \left(\C \cdot \partial_\eps \oplus \mf{h}, C^\ast_{red} ( f \g[[z_1,z_2, \eps]] )  \right)$ all have vanishing cohomology. 

Note that we can rewrite this complex as 
$$
C^\ast \left( \C \cdot \partial_\eps, C^\ast ( H^0 (X, \op{At}_{P,D}^0 )[\eps], C^\ast_{red} ( f \g[[z_1,z_2,\eps]] ) ) \right). 
$$
Vanishing of positive $\C^\times$-weights of this complex follows from a general fact.

\begin{lemma}
Let $\mf{a},\mf{b}$ be any two Lie algebras where $\mf{a}$ acts on $\mf{b}$.    Then, the positive $\C^\times$ weight spaces in the cohomology of 
$$C^\ast  \left ( \C \cdot \partial_\eps ,C^\ast(\mf{a}[\eps],  C^\ast_{red}(\mf{b}[\eps]) )\right)$$ all vanish. 
\label{lemma_vanishing}
\end{lemma}
\begin{proof}
Observe that there is a short exact sequence 
$$
0 \to  C^\ast(\mf{a}[\eps],  C^\ast_{red}(\mf{b}[\eps]) )  \to C^\ast_{red} (\mf{a}[\eps] \oplus \mf{b}[\eps])  \to C^\ast_{red}(\mf{a}[\eps])\to 0.
$$
By applying $C^\ast(\C \cdot \partial_\eps, -)$ to the terms in this sequence, we see that it suffices to prove that, for any Lie algebra $\mf{a}$, the positive weight spaces of $C^\ast(\C\cdot \dpa{\eps}, C^\ast_{red}(\mf{a}[\eps]))$ vanish.  

For any weight $k \ge 0$, one can identify the weight $k$ part of the complex $C^\ast(\C\cdot \dpa{\eps}, C^\ast_{red}(\mf{a}[\eps]))$ with the Chevalley complex of the contractible dg Lie algebra  $\mathfrak{a}[\eps]$ with differential $\partial_\eps$.  This is, of course, zero. 
\end{proof}
\end{proof}
Now let us see why this result proves the existence of the deformed twisted theory.  

The twisted, deformed theory on $X$ is defined using a holomorphic vector field $V$ preserving the divisor $D$ and the meromorphic volume form $\omega$.    Let $P$ be a principal $G$-bundle on $X$, equipped with a lift of the divergence-free vector field $V \in H^0(X, \op{Vect}^0_{X,D})$ to an element $\nabla_V \in H^0(X, \op{At}^0_{P,D})$.   The choice of $\nabla_V$ is the choice of a holomorphic connection in the direction spanned by $V$, which is trivial along the divisor $D$. 

The Lie algebra describing the twisted and deformed theory is 
$$\L^V =\left(  \Omega^{0,\ast}(X, \g_P(-D)[\eps]), \dbar + \lambda \nabla_V \right)$$
where $\lambda$ is a coupling constant, which we can set to be $1$. 

We will show how, using the dictionary between symmetries and deformations, we can use theorem \ref{theorem_existence_symmetries} to produce a quantization of this deformed theory.   The point is that the quantum theory produced by \ref{theorem_existence_symmetries} (and the corresponding factorization algebra) lives in the world of $\C^\times$-equivariant $C^\ast(\mf{h})[[\hbar]]$-modules.   Recall that
$$
\mf{h}= H^0 ( \op{At}^0_{P,D} ) [\eps].
$$
Since $\nabla_V$ is an element of $H^0 ( \op{At}^0_P)$, $\eps \nabla_V$ is an element of $\mf{h}$.    Thus there's a map
$$
C^\ast(\mf{h} ) \to C^\ast ( \C \cdot (\eps \nabla_V) ) = \C[[\lambda]]. 
$$
By tensoring with $\C[[\lambda]]$ using this map, we get a $\C^\times$-equivariant family of theories over the ring $\C[[\lambda]]$.  Classically, this is the twisted deformed theory, because the Lie algebra $\L^V$ is obtained from $\L$ by adding $\lambda \eps \nabla_V$ to the differential.

Thus, we've shown how to construct the twisted deformed theory at the quantum level, with the coupling constant $\lambda$ treated as a formal parameter.  It remains to explain how to evaluate $\lambda$ to be $1$. 

The factorization algebra corresponding to this theory is a $\C^\times$-equivariant factorization algebra over $\C[[\lambda,\hbar]]$, where $\hbar$ has weight $-1$ and $\lambda$ has weight $1$.  Inverting $\lambda$ and taking $\C^\times$ invariants gives us a theory over $\C[[\lambda^{-1} \hbar]] \iso \C[[\hbar]]$. This has the effect of setting $\lambda = 1$ but keeping $\hbar$ to be a formal parameter, and so removing the redundancy in the parameters. 

\subsection{Holomorphically translation invariant theories on $\C^2$}
As a variant of this calculation, we will show that theory on $\C^2$ has some additional symmetries which means that the corresponding factorization algebra is ``holomorphically translation invariant'': that is, the operator product varies holomorphically with the location of the operators.  In this subsection the divisor $D$ is taken to be empty. 

On $\C^2$, the Lie algebra 
$$\L= \Omega^{0,\ast}(\C^2,\g[\eps]) = \cinfty(\C^2) [ \d \zbar, \d \br{w}, \eps] \otimes \g$$ is invariant under translation.  Thus, it has an action of the real group $\R^4$, and lift of this to an action of the complexified Lie algebra $\C^4$ with basis $\dpa{z}, \dpa{w}, \dpa{\zbar}, \dpa{\br{w}}$.    Further, we have an action of the odd Lie algebra $\C^2[1]$, spanned by $\dpa{\d \zbar}$ and $\dpa {\d \br{w}}$.   Note that $[\dbar, \dpa{\d \zbar}] = \dpa{\zbar}$ and similarly for $\br{w}$. 

A small generalization of the arguments of \cite{Cos11} (which we plan to include in \cite{CosGwi11}) show that the obstruction-deformation complex controlling holomorphically translation invariant quantizations is $\C \otimes_{\C[\dpa{z}, \dpa{w}]}^{\mbb{L}} C^\ast_{red}( \g[[z,w,\eps]])$.    The same argument as before allows us to conclude that the relevant summands of the obstruction-deformation complex have no cohomology.

\newcommand{\etalchar}[1]{$^{#1}$}
\def\cprime{$'$}

\end{document}